\documentclass{article}
\usepackage{srcltx}
\usepackage{subfigure}
\usepackage{epsfig}
\usepackage[T1]{fontenc}
\usepackage[utf8]{inputenc}
\usepackage[english]{babel}
\usepackage{color}
\usepackage{amsmath}
\usepackage{amssymb}
\usepackage{amscd}
\usepackage{latexsym}
\usepackage[ruled,vlined]{algorithm2e}
\usepackage{mathrsfs}
\usepackage{url}
\usepackage{hyperref}
\usepackage{varioref}
\usepackage{pdfsync}
\usepackage{graphicx}
\usepackage{tikz}
\usetikzlibrary{calc}

\usepackage{comment}

\newtheorem{theorem}{Theorem}
\newtheorem{corollary}[theorem]{Corollary}
\newtheorem{definition}[theorem]{Definition}

\newtheorem{conjecture}[theorem]{Conjecture}

\newtheorem{remark}[theorem]{Remark}
\newtheorem{claim}{Claim}
\newtheorem{subclaim}{Claim}[claim]

\newtheorem{fact}{Fact}

\newcommand{\proof}{\noindent\textbf{Proof.} }

\newcommand{\smallqed}{{\tiny ($\Box$)}}
\newcommand{\QED}{$\Box$}

\newcommand{\score}{{\rm score}}
\newcommand{\island}{{\rm island}}
\newcommand{\isthmus}{{\rm isthmus}}
\newcommand{\detour}{{\rm detour}}

\newcommand{\twostar}{{\rm $2$-star}}

\newcommand{\mdom}{\dot\gamma}
\newcommand{\mV}{\dot V}

\newcommand{\markn}{\dot n}

\newcommand{\w}{{\rm \dot w}}
\newcommand{\nw}{{\rm w}}

\newcommand{\upper}{{\rm upper}}
\newcommand{\lowr}{{\rm lower}}

\newcommand{\mds}{MD-set}
\colorlet{mygreen}{green!50!black}
\colorlet{myred}{red!70!black}
\colorlet{amber}{orange!70!black}

\newcommand{\redp}{\textcolor{myred}{red path}}
\newcommand{\redps}{\textcolor{myred}{red paths}}
\newcommand{\grep}{\textcolor{mygreen}{green path}}
\newcommand{\greps}{\textcolor{mygreen}{green paths}}
\newcommand{\red}{\textcolor{myred}{red}}
\newcommand{\green}{\textcolor{mygreen}{green}}
\newcommand{\black}{\textcolor{black}{\textbf{black}}}

\newcommand{\rede}{\textcolor{myred}{red edge}}
\newcommand{\redes}{\textcolor{myred}{red edges}}
\newcommand{\gree}{\textcolor{mygreen}{green edge}}
\newcommand{\grees}{\textcolor{mygreen}{green edges}}
\newcommand{\blae}{\textcolor{black}{\textbf{black edge}}}
\newcommand{\blaes}{\textcolor{black}{\textbf{black edges}}}

\newcommand{\grei}{\textcolor{mygreen}{green island}}
\newcommand{\redi}{\textcolor{myred}{red island}}

\newcommand{\greis}{\textcolor{mygreen}{green islands}}
\newcommand{\redis}{\textcolor{myred}{red islands}}

\newcommand{\gisthmus}{\textcolor{mygreen}{green isthmus}}
\newcommand{\risthmus}{\textcolor{myred}{red isthmus}}

\newcommand{\bdetour}{\textcolor{black}{\textbf{black detour}}}

\newcommand{\rspecial}{\textcolor{myred}{special red edge}}
\newcommand{\bspecial}{\textcolor{black}{\textbf{special black edge}}}

\newcommand{\greencyclegraph}{\textcolor{purple!50!black}{Green cycle graph}}
\newcommand{\link}{{\color{blue!50!black} link}}
\newcommand{\links}{{\color{blue!50!black} links}}

\newcommand{\Amber}{{\color{amber}{Amber}}}
\newcommand{\Blue}{{\color{blue} Blue}}

\newcommand{\setA}{{\color{amber}{A}}}
\newcommand{\setB}{{\color{blue} B}}

\newcommand{\modo}{{\rm mod \,}}

\newcommand{\2}{ \vspace{0.2cm} }
\newcommand{\1}{ \vspace{0.1cm} }
\newcommand{\5}{ \vspace{0.05cm} }

\let\oldenumerate\enumerate
\renewcommand{\enumerate}{
  \oldenumerate
  \setlength{\itemsep}{0.5pt}
  \setlength{\parskip}{0pt}
  \setlength{\parsep}{0pt}
}

\begin{document}

\title{
The $\frac{1}{3}$-Conjectures for Domination in Cubic Graphs}
\author{$^1$Paul Dorbec and $^2$Michael A. Henning\\ \\
$^1$Normandie Univ, UNICAEN, ENSICAEN, CNRS, GREYC,\\ 14000 Caen, France, \\
\small \tt Email: paul.dorbec@unicaen.fr \\
\\
$^2$Department of Mathematics and Applied Mathematics\\
University of Johannesburg \\
Auckland Park, 2006 South Africa \\
\small \tt Email: mahenning@uj.ac.za}

\date{}
\maketitle

\maketitle
\begin{abstract}
A set $S$ of vertices in a graph $G$ is a dominating set of $G$ if every vertex not in $S$ is adjacent to a vertex in~$S$. The domination number of $G$, denoted by~$\gamma(G)$, is the minimum cardinality of a dominating set in $G$. In a breakthrough paper in 2008, L\"{o}wenstein and Rautenbach~\cite{LoRa08} proved that if $G$ is a cubic graph of order~$n$ and girth at least~$83$, then $\gamma(G) \le \frac{1}{3}n$. A natural question is if this girth condition can be lowered. The question gave birth to two $\frac{1}{3}$-conjectures for domination in cubic graphs. The first conjecture, posed by Verstraete in 2010, states that if $G$ is a cubic graph on $n$ vertices with girth at least~$6$, then $\gamma(G) \le \frac{1}{3}n$. The second conjecture, first posed as a question by Kostochka in 2009, states that if $G$ is a cubic, bipartite graph of order~$n$, then $\gamma(G) \le \frac{1}{3}n$. In this paper, we prove Verstraete's conjecture when there is no $7$-cycle and no $8$-cycle, and we prove the Kostochka's related conjecture for bipartite graphs when there is no $4$-cycle and no $8$-cycle.
\end{abstract}



\section{Introduction}

A \emph{dominating set} in a graph $G$ is a set $S$ of vertices of $G$ such that every vertex outside $S$ is adjacent to at least one vertex in $S$. The \emph{domination number} of $G$, denoted by $\gamma(G)$, is the minimum cardinality of a dominating set in $G$. A thorough treatise on dominating sets can be found in the so-called ``domination books''~\cite{HaHeHe-20,HaHeHe-21,HaHeHe-23,book-total}.

Independent sets in cubic graph of given girth is very well studied, see, for example,~\cite{CaGoJo-20,HeTh06,PePe18}. The problem of determining a sharp upper bound on the domination number of a connected, cubic graph, of sufficiently large order, in terms of its order, remains one of the major outstanding problems in domination theory. A classical result due to Reed~\cite{Re96} states that a  cubic graph $G$ of order $n$ satisfies $\gamma(G) \le \frac{3}{8}n$. His proof uses ingenious counting arguments. (Another proof of Reed's $\frac{3}{8}$-bound can be found in~\cite{DoHeMoSo15}.) Kostochka and Stodolsky~\cite{KoSt-05} proved that the two non-planar, connected, cubic graphs of order~$8$ are the only connected, cubic graphs that achieve the $\frac{3}{8}$-bound, by proving that if $G$ is a connected cubic graph of order $n \ge 10$, then $\gamma(G) \le \frac{4}{11}n$.  The best general upper bound to date on the domination number of a connected cubic graph is due to Kostochka and Stocker~\cite{KoSt09}.

\begin{theorem}{\rm (\cite{KoSt09})}
\label{t:KoSt}
If $G$ is a connected, cubic graph of order~$n$, then
$\gamma(G) \le \frac{5}{14}n = \left( \frac{1}{3} + \frac{1}{42} \right)n$.
\end{theorem}

However, it is not known if there are graphs of large order that achieve the $\frac{5}{14}$-bound in Theorem~\ref{t:KoSt}. Reed~\cite{Re96} conjectured that the domination number of a connected, cubic graph of order~$n$ is $\lceil n/3 \rceil$. Kostochka and Stodolsky~\cite{KoSt-05} disproved this conjecture by constructing an infinite sequence $\{G_k\}^\infty_{k=1}$ of connected, cubic graphs with
\[
\lim_{k\to\infty} \frac{\gamma(G_k)}{|V(G_k)|} \ge \frac{1}{3} + \frac{1}{69}.
\]

Subsequently, Kelmans~\cite{Ke-06} constructed an infinite series of $2$-connected, cubic graphs $H_k$ with
\[
\lim_{k\to\infty} \frac{\gamma(H_k)}{|V(H_k)|} \ge \frac{1}{3} + \frac{1}{60}.
\]

Thus, there exist connected cubic graphs $G$ of arbitrarily large order~$n$ satisfying $\gamma(G) \ge ( \frac{1}{3} + \frac{1}{60} )n$. All known counterexamples to Reed's conjecture, including the above constructions of Kostochka and Stodolsky~\cite{KoSt-05} and Kelmans~\cite{Ke-06}, contain small cycles. Much discussion has centered on when Reed's conjecture becomes true with the additional condition that the girth, $g$, of the graph is sufficiently large, where the \emph{girth} is the length of a shortest cycle in $G$. The first such result was presented by Kawarabayashi, Plummer, and Saito~\cite{KaPlSa06}.

\begin{theorem}{\rm (\cite{KaPlSa06})}
\label{t:KaPlSa}
If $G$ is a connected cubic graph of order~$n$ and girth~$g \ge 3$ has a $2$-factor, then
\[
\gamma(G) \le \left(\frac{1}{3} + \frac{1}{9 \lfloor g/3 \rfloor +3} \right)n.
\]
\end{theorem}

Refining the ideas and techniques from Reed's seminal paper~\cite{Re96}, and using intricate discharging arguments, Kostochka and Stodolsky~\cite{KoSt09b} improved the upper bound in Theorem~\ref{t:KaPlSa}.

\begin{theorem}{\rm (\cite{KoSt09b})}
\label{t:KoSt-girth}
If $G$ is a connected cubic graph of order~$n$ and girth~$g \ge 3$, then
\[
\gamma(G) \le \left(\frac{1}{3} + \frac{8}{3g^2} \right)n.
\]
\end{theorem}

Rautenbach and Reed~\cite{RaRe07} and Kr\'{a}l, \v{S}koda, and Volec~\cite{KrSkVo12} established further upper bounds  on the domination number of a cubic graph in terms of its order and girth. The magic threshold of $\frac{1}{3}n$ for the domination number was first shown to hold for cubic graphs with large girth by L\"{o}wenstein and Rautenbach~\cite{LoRa08}.

\begin{theorem}{\rm (\cite{LoRa08})}
\label{t:LoRa}
If $G$ is a cubic graph of order~$n$ and girth~$g \ge 83$, then $\gamma(G) \le \frac{1}{3}n$.
\end{theorem}

Theorem~\ref{t:LoRa} presented a breakthrough back in 2008. The L\"{o}wenstein-Rautenbach girth condition of $g \ge 83$ in Theorem~\ref{t:LoRa} guaranteeing that the domination number of a cubic graph of order~$n$ is at most the magical threshold of $\frac{1}{3}n$ has yet to be improved, and has attracted considerable attention in the literature. Verstraete~\cite{JacVert10} conjectured in 2010 that the girth condition can be lowered significantly from $g \ge 83$ to $g \ge 6$ in order to guarantee that the $\frac{1}{3}$-bound will hold.

\begin{conjecture}{\rm (Verstraete)}
\label{conj2}
If $G$ is a cubic graph on $n$ vertices with girth $g \ge 6$, then $\gamma(G) \le \frac{1}{3}n$.
\end{conjecture}

We remark that the girth requirement in Verstraete's Conjecture~\ref{conj2} is essential, since the generalized Petersen graph $P(7,2)$, shown in Figure~\ref{f:GenP}, of order~$n = 14$ satisfies girth $g = 5$ and $\gamma(G) = 5 > \frac{1}{3}n$. We also remark that Kostochka and Stodolsky~\cite{KoSt-05} and Kelmans~\cite{Ke-06} constructed an infinite family of connected, cubic graphs of order~$n$ with girth~$4$ and $\gamma(G) > \frac{1}{3}n$.

\begin{figure}[htb]
\begin{center}
\tikzstyle{every node}=[circle, draw, fill=black!0, inner sep=0pt,
minimum width=.16cm]
\begin{tikzpicture}[thick,scale=.6]%
  \draw \foreach \x in {1,2,...,7}
  {
    (\x*360/7+90:2) node{} -- (\x*360/7+90+360/7:2)
    (\x*360/7+90:2) -- (\x*360/7+90:1) node{}
    (\x*360/7+90:1) -- (\x*360/7+90+720/7:1)
  };
\end{tikzpicture}
\caption{The generalized Petersen graph $P(7,2)$.}
\label{f:GenP}
\end{center}
\end{figure}
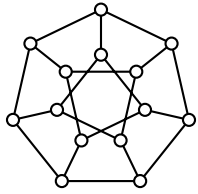

Another closely related $\frac{1}{3}$-conjecture for domination in cubic graphs can be attributed to Kostochka~\cite{Ko09} who announced the following question in the open problem session at the Third International Conference on Combinatorics, Graph Theory and Applications, held at Elgersburg, Germany, March 2009: \emph{Is it true that the domination number of a bipartite cubic graph is at most one-third its order?} Kostochka and Stodolsky comment in their paper in~\cite{KoSt09b} that it would be interesting to answer this question. This intriguing question of Kostochka was posed seven years later as a formal conjecture in~\cite{He16}.

\begin{conjecture}{\rm (\cite{He16})}
\label{conj1}
If $G$ is a cubic bipartite graph of order~$n$, then $\gamma(G) \le \frac{1}{3}n$.
\end{conjecture}

\section{Main results}

Our aim is this paper is to make significant progress on both the Verstraete Conjecture~\ref{conj2} and the Kostochka's inspired Conjecture~\ref{conj1} for domination in cubic graphs. Our main result proves Verstraete's $\frac{1}{3}$-Conjecture when the graph does not contain a $7$-cycle or $8$-cycle as a subgraph, that is, when the graph has girth at least~$6$ but is heptagon-free and octagon-free.

\begin{theorem}
\label{t:main1}
If $G$ is a cubic graph of order~$n$ and girth~$g \ge 6$ that does not contain a $7$-cycle or $8$-cycle, then $\gamma(G) \le \frac{1}{3}n$.
\end{theorem}

As an immediate consequence of Theorem~\ref{t:main1}, the $\frac{1}{3}$-Conjecture inspired by Kostochka (and posed formally as a conjecture in~\cite{He16}) holds when the cubic bipartite graph contains no $4$-cycle and no $8$-cycle.

\begin{theorem}
\label{t:main2}
If $G$ is a cubic bipartite graph of order~$n$ that does not contain a $4$-cycle or $8$-cycle, then $\gamma(G) \le \frac{1}{3}n$.
\end{theorem}

As an immediate consequence of Theorem~\ref{t:main1} we also have the following new lower bound on the girth guaranteeing that the domination number of a cubic graph is at most one-thirds its order. This result significantly lowers the previously best known girth condition, namely $g \ge 83$, due to L\"{o}wenstein-Rautenbach.

\begin{corollary}
\label{t:cor1}
If $G$ is a cubic graph of order~$n$ and girth~$g \ge 9$, then $\gamma(G) \le \frac{1}{3}n$.
\end{corollary}

\section{Notation}

For notation and graph theory terminology we generally follow~\cite{HaHeHe-23}. Specifically, let $G$ be a graph with vertex set $V(G)$ of order~$n(G) = |V(G)|$ and edge set $E(G)$ of size~$m(G) = |E(G)|$. If $G$ is clear from the context, we simply write $V$ and $E$ rather than $V(G)$ and $E(G)$. Let $v$ be a vertex in $V$. We denote the \emph{degree} of $v$ in $G$ by $\deg_G(v)$.
The \emph{open neighborhood} of $v$ is $N_G(v) = \{u \in V \, \colon uv \in E\}$ and the \emph{closed neighborhood of $v$} is $N_G[v] = \{v\} \cup N_G(v)$. For a set $S \subseteq V$, its \emph{open neighborhood} is the set $N_G(S) = \bigcup_{v \in S} N_G(v)$, and its \emph{closed neighborhood} is the set $N_G[S] = N_G(S) \cup S$. If the graph $G$ is clear from the context, we simply write $d(v)$, $N(v)$, $N[v]$, $N(S)$ and $N[S]$ rather than $\deg_G(v)$, $N_G(v)$, $N_G[v]$, $N_G(S)$ and $N_G[S]$, respectively.

A \emph{cycle} on $n$ vertices is denoted by $C_n$ and a \emph{path} on $n$ vertices by $P_n$. The \emph{girth} of $G$, denoted $g(G)$, is the length of a shortest cycle in $G$. For a set $S \subseteq V$, the subgraph of $G$ induced by $S$ is denoted by $G[S]$. Further if $S \ne V$, then we denote the graph obtained from $G$ by deleting all vertices in $S$ (as well as all incident edges) by $G - S = G[V \setminus S]$. If $S = \{v\}$, we also write $G - S$ simply as $G - v$. A graph $G$ is said to be \emph{subcubic} if its maximum degree is at most~$3$, and \emph{cubic} if every vertex has degree~$3$. A vertex of degree~$0$ is called an \emph{isolated vertex}. A \emph{double star} is a tree with exactly two vertices that are not leaves. We call these two vertices the \emph{central vertices} of the double star. Further, if the one central vertex has $s$ leaf neighbors and the other $t$ leaf neighbors, we denote the double star by $S(s,t)$.

\section{Statement of the main theorem}

In order to prove our main result, namely Theorem~\ref{t:main1}, we need to prove a much stronger result. For this purpose, we introduce the concept of a marked domination set in a graph. Let $G$ be a graph where every vertex in $G$ is either marked or unmarked. A \emph{marked dominating set}, abbreviated MD-set, is a set $S$ of vertices in $G$ such that every vertex that is not marked belongs to the set $S$ or is adjacent to a vertex in the set $S$. The \emph{marked domination number} of $G$, denoted by $\mdom(G)$, is the minimum cardinality of a MD-set in $G$. Thus, the marked domination number of $G$ is the minimum number of vertices needed to dominate all the unmarked vertices, where a vertex dominates itself and all its neighbors.

Let $\mV(G)$ denote the set of marked vertices in $G$. We note that $\mV(G) \subseteq V(G)$. We let $n_i(G)$ denote the number of vertices of degree~$i$ in $G$ that are unmarked. Further, we let $\markn(G)$ denote the number of marked vertices in $G$, and so $\markn(G) = |\mV(G)|$. Thus if $G$ is a subcubic graph, then
\[
n(G) = \markn(G) + \sum_{i=0}^3 n_i(G).
\]

If the graph $G$ is clear from context, we simply write $V$, $\mV$, $n$, $\markn$, and $n_i$ rather than $V(G)$, $\mV(G)$, $n(G)$, $\markn(G)$ and $n_i(G)$, respectively.
A \emph{marked subcubic graph} is a subcubic graph in which every vertex is either marked or unmarked.
We are now in a position to state our key result.

\begin{theorem}\label{theo:main}
If $G$ is a marked subcubic graph of girth at least~$6$ with no cycle of length~$7$ or~$8$, then
\[
12 \mdom(G) \le 4 \markn(G) + 4 n_3(G) + 5  n_2(G) + 8  n_1(G) + 12  n_0(G).
\]
\end{theorem}

We remark that if the set of marked vertices in a graph $G$ is empty, then the marked domination number of $G$ is precisely its domination number; that is, if $\mV = \emptyset$, then $\mdom(G) = \gamma(G)$. Theorem~\ref{t:main1} follows immediately from Theorem~\ref{theo:main} in the special case when $G$ is a cubic graph with no marked vertices. For a vertex $v$ in a marked subcubic graph $G$, we define its \emph{weight} in $G$ as

\[
\w_G(v) =
\begin{cases}
  12 & \mbox{\rm if $\deg_G(v) = 0$ and $v \notin \mV(G)$},
  \\ 8 & \mbox{\rm if $\deg_G(v) = 1$ and $v \notin \mV(G)$},
  \\ 5 & \mbox{\rm if $\deg_G(v) = 2$ and $v \notin \mV(G)$},
  \\ 4 & \mbox{\rm if $\deg_G(v) = 3$ or if $v \in \mV(G)$}.
\end{cases}
\]
For a subset $X \subseteq V(G)$, we define the \emph{weight} of $X$ in $G$, as

\[
\w_G(X) = \sum_{v \in X} \w_G(v),
\]

\noindent
that is, the weight of $X$ is the sum of the weights in $G$ of vertices in $X$. We refer to the weight of $V(G)$ (the entire vertex set) in $G$ simply as the \emph{weight of $G$}, denoted $\w(G)$. Thus,
\[
\w(G) = \sum_{v \in V(G)} \w_G(v) = 4 \markn(G) + 4 n_3(G) + 5  n_2(G) + 8  n_1(G) + 12  n_0(G).
\]

We note that each marked vertex in $G$ has weight~$4$, independent of its degree in $G$. Further, the weight of a marked vertex of degree~$3$ is the same as the weight of an unmarked vertex of degree~$3$. We now restate Theorem~\ref{theo:main} in terms of the weight of the graph.

\medskip
\noindent \textbf{Theorem~\ref{theo:main}} \emph{If $G$ is a marked subcubic graph of girth at least~$6$ with no cycle of length~$7$ or~$8$, then
\[
12 \mdom(G) \le \w(G).
\]}

\section{Proof of the main theorem}

The proof is based on a series of claims, which we distribute into thematic subsections to enhance readability of the proof. Suppose, to the contrary, that the theorem is false. Among all counterexample to Theorem~\ref{theo:main}, let $G$ be chosen so that $|V(G)|+|E(G)|$ is a minimum and, subject to this condition, the number of marked vertices of degree~$3$ is a minimum. We will frequently use the following fact.

\begin{fact} \label{fact1}
If $H$ is a proper subgraph of $G$, and $\alpha$ and $\beta$ are two integers such that $\mdom(G) \le \mdom(H) + \alpha$ and $\w(G) \ge \w(H) + \beta$, then $12\alpha > \beta$.
\end{fact}
\proof
Suppose, to the contrary, that $12\alpha \le \beta$. Since $G$ is a minimum counterexample and $H$ is a proper subgraph of $G$, we note that $12\mdom(H) \le \w(H)$. Thus, $12\mdom(G) \le 12(\mdom(H) + \alpha) = 12\mdom(H) + 12\alpha \le \w(H) + \beta \le \w(G)$, a contradiction.~\smallqed

\medskip
In what follows we present a series of claims
which culminate in the implication of the non-existence of the counterexample graph $G$.

\subsection{Fundamental structural properties}

In this subsection, we establish some fundamental structural properties of the graph $G$. We define a $2$-\emph{path} in $G$ as a path all of whose vertices have degree~$2$ in $G$. The claims we prove in this section show that the graph $G$ is a connected graph in which every vertex has degree~$2$ or~$3$. Further, we show that every marked vertex has degree~$3$, and that the marked vertices are at distance at least~$4$ apart in $G$. We also establish restrictions on the length of a $2$-path in $G$, and show that every maximal $2$-path in $G$ has length~$1$,~$2$,~$4$ or~$5$. We begin with the following property of marked vertices in the graph $G$.

\begin{claim}\label{marked-not-deg-3}
Every marked vertex in $G$ has degree~$1$ or~$2$.
\end{claim}
\proof  Suppose, to the contrary, that $G$ contains a marked vertex $v$ of degree~$0$ or $3$. Suppose, firstly that $v$ is isolated in $G$. In this case, let $H = G - v$. Every \mds\ of $H$ is a \mds\ of $G$, and conversely, implying that $\mdom(G) = \mdom(H)$. Since $\w_G(v)=4$, we have $\w(G) = \w(H) + 4$. Thus, $\mdom(G) \le \mdom(H) + \alpha$ and $\w(G) \ge \w(H) + \beta$, where $\alpha = 0$ and $\beta = 4$, contradicting Fact~\ref{fact1}. Suppose, next, that the vertex $v$ has degree~$3$ in $G$. In this case, let $H$ be the graph obtained from $G$ by unmarking the vertex~$v$. Since $|V(H)|+|E(H)| = |V(G)|+|E(G)|$ and $H$ has fewer marked vertices of degree~$3$ than does~$G$, our choice of $G$ implies that the graph $H$ is not a counterexample to our theorem. Therefore, $12 \mdom(H) \le \w(H)$. Every \mds\ of $H$ is a \mds\ of $G$, implying that $\mdom(G) \le \mdom(H)$. Since $\w_G(v) = \w_H(v) = 4$, we note that $\w(G) = \w(H)$. Thus, $12 \mdom(G) \le 12 \mdom(H) \le \w(H) = \w(G)$, contradicting the fact that $G$ is a counterexample.~\smallqed

\begin{claim}\label{G-connected}
The graph $G$ is connected.
\end{claim}
\proof  Suppose, to the contrary, that $G$ is disconnected with components $G_1,G_2,\ldots,G_k$ where $k \ge 2$. By linearity,
\[
\mdom(G) = \sum_{i=1}^k \mdom(G_i) \hspace*{0.5cm} \mbox{and} \hspace*{0.5cm} \w(G) = \sum_{i=1}^k \w(G_i).
\]

By the minimality of $G$, all components of $G$ satisfy Theorem~\ref{theo:main}, and so $12\mdom(G_i) \le \w(G_i)$ for $i \in [k]$, implying that $\mdom(G) \le \w(G)$, and so Theorem \ref{theo:main} therefore holds for $G$, a contradiction.~\smallqed

\begin{claim}\label{one-unmarked} 
The graph $G$ contains at least one unmarked vertex.
\end{claim}
\proof Suppose, to the contrary, that $G$  contains only marked vertices. Then, $\mdom(G) = 0$ and $\w(G) = 4n$, implying that Theorem \ref{theo:main} therefore holds for $G$, a contradiction.~\smallqed

\begin{claim}\label{no-unmarked-deg-1}
The graph $G$ contains no unmarked vertex of degree~$1$.
\end{claim}
\proof
Suppose, to the contrary, that $G$ contains an unmarked vertex, $u$ say, of degree~$1$. Let $v$ be the neighbor of $u$. If $n = 2$, then $\gamma(G) = \mdom(G) = 1$ and $\w(G) \ge 12$ noting that the vertex $v$ may possibly be marked, contradicting the fact that $G$ is a counterexample to the theorem. Hence, $\deg(v) \ge 2$. Let $H = G - \{u,v\}$, where we mark all neighbors of $v$ in $H$. Every \mds\ of $H$ can be extended to a \mds\ of $G$ by adding to it the vertex $v$, implying that $\mdom(G) \le \mdom(H)+1$. Since $\w_G(u)=8$ and $\w_G(v) \ge 4$ and $\w_G(w) \ge \w_H(w)$ for all vertices $w \in V(H)$, we have $\w(G) \ge \w(H) + 12$. Thus, $\mdom(G) \le \mdom(H) + \alpha$ and $\w(G) \ge \w(H) + \beta$, where $\alpha = 1$ and $\beta= 12$, contradicting Fact~\ref{fact1}.~\smallqed

\begin{claim}\label{no-adjacent-marked}
The graph $G$ does not contain adjacent marked vertices.
\end{claim}
\proof
Suppose, to the contrary, that $G$ contains two adjacent marked vertices. Let $H$ be obtained from $G$ by removing such an edge between marked vertices. Then, $\mdom(H) = \mdom(G)$ and $\w(H) = \w(G)$. Thus, $\mdom(G) \le \mdom(H) + \alpha$ and $\w(G) \ge \w(H) + \beta$, where $\alpha = \beta = 0$, contradicting Fact~\ref{fact1}.~\smallqed

\begin{claim}\label{marked-deg-2}
Every marked vertex in $G$ has degree~$2$.
\end{claim}
\proof
By Claim~\ref{marked-not-deg-3}, every marked vertex in $G$ has degree~$1$ or~$2$. Suppose, to the contrary, that $G$ contains a marked vertex $u$ of degree~$1$ in $G$ adjacent to a vertex $v$. By Claims~\ref{no-unmarked-deg-1} and~\ref{no-adjacent-marked}, the vertex $v$ is an unmarked vertex of degree~$2$ or~$3$. Let $H = G - u$. Every \mds\ of $H$ is a \mds\ of $G$, implying that $\mdom(G) \le \mdom(H)$. Since $\w_G(u) = 4$ and $\w_G(v)  \ge \w_H(v) -3$, we note that $\w(G) \ge \w(H)+ 4 - 3 = \w(H) + 1$. Thus, $\mdom(G) \le \mdom(H) + \alpha$ and $\w(G) \ge \w(H) + \beta$, where $\alpha = 0$ and $\beta = 1$, contradicting Fact~\ref{fact1}.~\smallqed

\begin{claim}\label{no-deg-1}
Every vertex in $G$ has degree~$2$ or~$3$.
\end{claim}
\proof
This follows immediately from Claims~\ref{no-unmarked-deg-1} and Claim~\ref{marked-deg-2}.~\smallqed

\begin{claim}\label{no-marked-next-to-deg-3}
The graph $G$ does not contain a marked vertex adjacent to a vertex of degree~$3$.
\end{claim}
\proof
Suppose, to the contrary, that $G$ contains a marked vertex $u$ adjacent to a vertex $v$ of degree~$3$. By Claim~\ref{marked-deg-2}, the vertex $u$ has degree~$2$. Let $w$ be the second neighbor of $u$. By Claims~\ref{no-adjacent-marked} and~\ref{no-deg-1}, the vertex $w$ is an unmarked vertex of degree~$2$ or~$3$. Let $H = G - u$. We note that $\w_G(u) = 4$, $\w_G(v) - \w_H(v) = 4 - 5 = - 1$, and $\w_G(w) - \w_H(w) \ge -3$, while the weights of all other vertices remain unchanged in $G$ and $H$. Thus, $\w(G) \ge \w(H) + 4 - 1 - 3 = \w(H)$. Further, every \mds\ of $H$ is a \mds\ of $G$, implying that $\mdom(G) \le \mdom(H)$. Thus, $\mdom(G) \le \mdom(H) + \alpha$ and $\w(G) \ge \w(H) + \beta$, where $\alpha = \beta = 0$, contradicting Fact~\ref{fact1}.~\smallqed

\begin{claim}\label{no-2path-of-length-3-with-extremity-marked}
There is no $2$-path of length~$3$ in $G$ that starts at a marked vertex with the remaining vertices of the path unmarked.
\end{claim}
\proof
Suppose, to the contrary, that $G$ contains a $2$-path $P \colon u_1u_2u_3u_4$ that emanates from a marked vertex $u_1$, where $u_2$, $u_3$, and $u_4$ are unmarked vertices of degree~$2$ in $G$. By Claim~\ref{marked-deg-2}, the marked vertex $u_1$ has degree~$2$. Let $u_0$ and $u_5$ be the neighbors of $u_1$ and $u_5$, respectively, not on $P$. The girth condition $g \ge 6$ implies that $u_0 \ne u_5$.
By Claim~\ref{no-deg-1},
the vertices $u_0$ and $u_5$ are of degree at least~$2$. We now consider the graph $H = G - V(P)$. Every \mds\ of $H$ can be extended to a \mds\ of $G$ by adding to it the vertex $v_3$, implying that $\mdom(G) \le \mdom(H)+1$. Since $\w_G(u_1) = 4$ and $\w_G(u_i) = 5$ for $i \in \{2,3,4\}$, while $\w_G(u_i) - \w_H(u_i) \ge -3$ for $i \in \{0,5\}$, we note that $\w(G) \ge \w(H)+ 4 + 3\times 5 - 2\times 3  = \w(H) + 13$. Thus, $\mdom(G) \le \mdom(H) + \alpha$ and $\w(G) \ge \w(H) + \beta$, where $\alpha = 1$ and $\beta = 13$, contradicting Fact~\ref{fact1}.~\smallqed

\begin{claim}\label{marked-at-least-4-apart}
The marked vertices in $G$ are at distance at least~$4$ apart from each other.
\end{claim}
\proof
Suppose, to the contrary, that $G$ contains two marked vertices, $u$ and $v$, at distance less than~$4$ apart. By Claim~\ref{no-adjacent-marked}, $d(u,v) \ge 2$.

\begin{subclaim}\label{marked-at-least-3-apart}
$d(u,v) = 3$.
\end{subclaim}
\proof
Suppose that $d(u,v) = 2$. By Claims~\ref{no-adjacent-marked}, \ref{marked-deg-2}, \ref{no-deg-1}, and~\ref{no-marked-next-to-deg-3}, and the girth $g \ge 6$ condition, $u$ and $v$ are internal vertices of a $2$-path $P \colon u_1u_2u_3u_4u_5$ of order~$5$, where $u_2=u$ and $u_4=v$. Recall that by definition of a $2$-path, every vertex of $P$ has degree~$2$ in $G$. Let $u_0$ and $u_6$ be the neighbors of $u_1$ and $u_5$, respectively, not on $P$.

Suppose that the vertex $u_0$ is either marked or has degree~$3$. In this case, let $H = G - \{u_1,u_2,u_3\}$. Every \mds\ of $H$ can be extended to a \mds\ of $G$ by adding to it the vertex $u_2$, implying that $\mdom(G) \le \mdom(H)+1$. We note that the vertex $u_0$ is marked or has degree~$2$ in $H$. Therefore, the weight of $u_0$ in $H$ either remains unchanged from its weight in $G$ or increases by~$1$, implying that $\w(G) \ge \w(H) + 4 + 2\times 5 - 1  = \w(H) + 13$. Thus, $\mdom(G) \le \mdom(H) + \alpha$ and $\w(G) \ge \w(H) + \beta$, where $\alpha = 1$ and $\beta = 13$, contradicting Fact~\ref{fact1}. Hence, $u_0$ is unmarked and has degree~$2$ in $G$. Analogously, $u_6$ is unmarked and has degree~$2$ in $G$.

If $u_0=u_6$, then $G \cong C_6$, $\mdom(G)=2$, and $\w(G) = 4 \times 5 + 2 \times 4 > 24 = 12\mdom(G)$, contradicting the fact that $G$ is a counterexample to our theorem. Hence,  $u_0 \ne u_6$. Since there is no cycle of length~$7$ in $G$, we note that $u_0$ and $u_6$ are not adjacent. Let $u_7$ be the neighbor of $u_6$ different from $u_5$. By Claim~\ref{no-2path-of-length-3-with-extremity-marked}, the vertex $u_7$ is either a marked vertex (of degree~$2$) or has degree~$3$. In this case, let $H = G - \{u_1,\ldots,u_6\}$ where we mark $u_7$ if it is an unmarked vertex in $G$. We note that the weight of $u_7$ in $H$ remains unchanged from its weight in $G$, while $\w_G(u_0) - \w_H(u_0) = -3$, implying that $\w(G) \ge \w(H) + 4\times 5 + 2 \times 4 - 3 = \w(H) + 25$.  Every \mds\ of $H$ can be extended to a \mds\ of $G$ by adding to it the vertices $u_2$ and $u_6$, implying that $\mdom(G) \le \mdom(H)+2$. Thus, $\mdom(G) \le \mdom(H) + \alpha$ and $\w(G) \ge \w(H) + \beta$, where $\alpha = 2$ and $\beta = 25$, contradicting Fact~\ref{fact1}. This completes the proof of Claim~\ref{marked-at-least-3-apart}.~\smallqed

\medskip
By Claim~\ref{marked-at-least-3-apart}, the marked vertices in $G$ are at distance at least~$3$ apart. In particular, $d(u,v) = 3$. By Claims~\ref{marked-deg-2}, \ref{no-deg-1} and~\ref{no-marked-next-to-deg-3}, and by the girth $g \ge 6$ condition, $u$ and $v$ are internal vertices of a $2$-path $P \colon u_1u_2u_3u_4u_5u_6$ of order~$6$, where $u_2=u$ and $u_5=v$.
If $u_1$ and $u_6$ are adjacent, then $G \cong C_6$, and $12\mdom(G) < \w(G)$, a contradiction. Hence, $u_1$ and $u_6$ are not adjacent. Let $u_0$ and $u_7$ be the neighbors of $u_1$ and $u_6$, respectively, not on $P$. As observed earlier, $u_1 \ne u_7$ and $u_0 \ne u_6$. Further, $u_0$ and  $u_7$ are unmarked. Since there is no cycle of length~$7$ in $G$, we note that $u_0 \ne u_7$.

Suppose that the vertex $u_7$ has degree~$2$. Let $u_8$ be the neighbor of $u_7$ different from $u_6$. By Claim~\ref{no-2path-of-length-3-with-extremity-marked}, the vertex $u_8$ is either a marked vertex (of degree~$2$) or has degree~$3$. In particular, we note that $u_1 \ne u_8$. We now let $H = G - \{u_2,\ldots,u_7\}$, where we mark $u_8$ if it is an unmarked vertex in $G$. We note that the weight of $u_8$ in $H$ remains unchanged from its weight in $G$, while $\w_G(u_0) - \w_H(u_0) = -3$, implying that $\w(G) \ge \w(H) + 4\times 5 + 2 \times 4 - 3 = \w(H) + 25$.  Every \mds\ of $H$ can be extended to a \mds\ of $G$ by adding to it the vertices $u_3$ and $u_7$, implying that $\mdom(G) \le \mdom(H)+2$. Thus, $\mdom(G) \le \mdom(H) + \alpha$ and $\w(G) \ge \w(H) + \beta$, where $\alpha = 2$ and $\beta = 25$, contradicting Fact~\ref{fact1}.

Hence, the vertex $u_7$ has degree~$3$. Analogously, the vertex $u_0$ has degree~$3$. We now let $H = G - V(P)$. Every \mds\ of $H$ can be extended to a \mds\ of $G$ by adding to it the vertices $u_2$ and $u_5$, implying that $\mdom(G) \le \mdom(H)+2$. Since $\w_G(u_0) - \w_H(u_0) = \w_G(u_7) - \w_H(u_7) = -1$, we note that $\w(G) \ge \w(H) + 4\times 5 + 2 \times 4 - 2 \times 1 = \w(H) + 26$. Thus, $\mdom(G) \le \mdom(H) + \alpha$ and $\w(G) \ge \w(H) + \beta$, where $\alpha = 2$ and $\beta = 26$, contradicting Fact~\ref{fact1}. This completes the proof of Claim~\ref{marked-at-least-4-apart}.~\smallqed

\medskip
By Claims~\ref{no-2path-of-length-3-with-extremity-marked} and~\ref{marked-at-least-4-apart}, a $2$-path in $G$ contains at most one marked vertex.

\begin{claim}\label{G-not-cycle}
The graph $G$ is not a cycle.
\end{claim}
\proof
Suppose, to the contrary, that $G \cong C_n$. By the girth condition and the assumption that there is no $7$-cycle in $G$, we note that $n \ge 6$ and $n \ne 7$. Further by Claims~\ref{no-2path-of-length-3-with-extremity-marked} and~\ref{marked-at-least-4-apart}, there is no marked vertex. Thus, $\w(G) = 5n$ and $\mdom(G) \le \gamma(C_n) = \lceil \frac n 3 \rceil$.
If $n \equiv 0 \, (\modo \, 3)$, then $12\mdom(G) \le 12 \times \lceil \frac n 3 \rceil = 4n < 5n \le \w(G)$. If $n \equiv 2 \, (\modo \, 3)$, then $12\mdom(G) \le 12 \times \lceil \frac n 3 \rceil = 4n+4 \le 5n \le \w(G)$. If $n \equiv 1 \, (\modo \, 3)$, then since $n \ge 6$ and $n \ne 7$, we note that $n \ge 10$, and therefore that $12\mdom(G) \le 12 \times \lceil \frac n 3 \rceil = 4n + 8 \le 5n \le \w(G)$. In all three cases, $12\mdom(G) \le \w(G)$, contradicting the fact that $G$ is a counterexample.~\smallqed

\begin{claim}\label{2path-length-at-most-5}
Every $2$-path in $G$ has order at most~$5$.
\end{claim}
\proof
Suppose, to the contrary, that $G$ contains a $2$-path $P \colon u_1u_2u_3u_4u_5u_6$ of order~$6$. By Claim~\ref{no-deg-1}, every vertex in $G$ has degree~$2$ or~$3$. By Claim~\ref{G-not-cycle}, $G$ is not a cycle, and so at least one vertex in $G$ has degree~$3$. The $2$-path $P$ can therefore be chosen so that the neighbor of $u_1$, say $u_0$, not on $P$ has degree~$3$. Let $H = G - V(P)$. Let $u_7$ be the neighbor of $u_6$ not on $P$. We note that $u_0 \ne u_7$, and that $\w_G(u_0) - \w_H(u_0) = -1$ and $\w_G(u_7) - \w_H(u_7) \ge -3$. Thus, since every $2$-path in $G$ contains at most one marked vertex, $\w(G)\ge \w(H) + 5 \times 5 + 1 \times 4 - 1 - 3 \ge \w(H)+25$. Every \mds\ of $H$ can be extended to a \mds\ of $G$ by adding to it the vertices $u_2$ and $u_5$, implying that $\mdom(G) \le \mdom(H)+2$. Thus, $\mdom(G) \le \mdom(H) + \alpha$ and $\w(G) \ge \w(H) + \beta$, where $\alpha = 2$ and $\beta = 25$, contradicting Fact~\ref{fact1}.~\smallqed

\begin{claim}\label{no-2path-of-length-3}
The graph $G$ does not contain a maximal $2$-path of order~$3$.
\end{claim}
\proof
Suppose, to the contrary, that $G$ contains a maximal $2$-path $P \colon u_1u_2u_3$ of order~$3$. Let $u_0$ and $u_4$ be the neighbor of $u_1$ and $u_3$, respectively, not on $P$. By the maximality of the $2$-path $P$, both  $u_0$ and $u_4$ are vertices of degree~$3$. Let $H = G - V(P)$. Since every $2$-path in $G$ contains at most one marked vertex, $\w(G) \ge \w(H) + 2 \times 5 + 1 \times 4 - 2 \times 1 \ge \w(H)+12$. Every \mds\ of $H$ can be extended to a \mds\ of $G$ by adding to it the vertex $u_2$, implying that $\mdom(G) \le \mdom(H)+1$. Thus, $\mdom(G) \le \mdom(H) + \alpha$ and $\w(G) \ge \w(H) + \beta$, where $\alpha = 1$ and $\beta = 12$, contradicting Fact~\ref{fact1}.~\smallqed

\subsection{The associated colored multigraph}

In this subsection, we define a colored cubic multigraph associated with the graph $G$. For this purpose, we introduce the following notation. By Claim~\ref{2path-length-at-most-5} and~\ref{no-2path-of-length-3}, every maximal $2$-path in $G$ has order~$1$, $2$, $4$ or $5$. We call a maximal $2$-path of order~$2$ or~$5$ a \emph{\redp}, and we call a maximal 2-path of order~$1$ or~$4$ a \emph{\grep}. We call a path a \emph{colored path} if it is a \redp\ or a \grep.  We call the \emph{ends} of a colored path the vertices of degree~$3$ adjacent to the extremities of the corresponding $2$-path.

Let $u$ and $v$ be the ends of a colored path $x_1\ldots x_k$, where $u$ and $v$ are adjacent to $x_1$ and $x_k$, respectively, and so $ux_1\ldots x_kv$ is a path in $G$ which we denote by $[uv]$.  Further, we denote by $(uv)$ the subpath $x_1\ldots x_k$ of the path $[uv]$ consisting of its internal vertices, and we denote by $(uv)_i$ the vertex $x_i$ for $i \in [k]$. We note that $(uv)$ is a maximal $2$-path in $G$, and is either a \redp\ or \grep. We call a colored path \textbf{short} if it is a \redp\ of order~$2$ or a \grep\ of order~$1$ and \textbf{long} if it is a \redp\ of order~$5$ or a \grep\ of order~$4$. By Claim~\ref{no-marked-next-to-deg-3} and~\ref{marked-at-least-4-apart}, we note that a short colored path contains no marked vertex, and a long colored path contains at most one marked vertex.

We now $3$-color the edges of $G$ as follows. We color every edge of $G$ joining two vertices of degree~$3$ with the color \black. Every edge of $G$ that belongs to a path $[uv]$, where the subpath $(uv)$ is a \emph{\redp} or a \emph{\grep}, we color \red\
or \green, respectively. Thus, if $u$ and $v$ are vertices of degree~$3$ in $G$, and $u$ and $v$ are the ends of a \emph{\redp} (respectively, \emph{\grep}), say $x_1 \ldots x_k$, then we color every edge on the path $[uv] \colon u x_1 \ldots x_k v$ with the color \red\ (respectively, \green). In this way, every edge of $G$ is colored with one of the colors \black, \red\ or \green.

We now assimilate $G$ with the corresponding \textbf{cubic colored multigraph} $M_G$, where $V(M_G)$ consists of the vertices of degree~$3$ in $G$ and where the edges of $M_G$ correspond to the \black, \green\ and \red\ edges defined in $G$.
Thus, a \blae\ $uv$ in $M_G$ corresponds to a \blae\ $uv$ in $G$, while a \green\ (respectively, \red) edge $uv$ in $M_G$ corresponds to a path $[uv]$ in $G$ associated with a \green\ (respectively, \red) path $(uv)$ in~$G$. We show next that the multigraph $M_G$ is in fact a graph.

\begin{claim}\label{colored-simple-graph}
The colored multigraph $M_G$ does not contain loops or multiple edges.
\end{claim}
\proof
Suppose, to the contrary, that $M_G$ contains loops or multiple edges. We are six possible cases. We consider each case in turn.

\smallskip
\emph{Case~1. $M_G$ contains a loop.} By the girth condition, a loop in $M_G$ corresponds to a long \red\ edge, say $x_1\ldots x_5$, both of whose ends are adjacent to a common vertex, $u$ say, in $G$. Let $H = G - \{u,x_1,\ldots,x_5\}$, where we mark the neighbor, $v$ say, of $u$ in $G$ different from $x_1$ and $x_5$. Since every long colored path contains at most one marked vertex, this implies that $\w(G) \ge \w(H) + 4\times 5 + 2\times 4 = \w(H)+ 28$. Every \mds\ of $H$ can be extended to a \mds\ of $G$ by adding to it the vertices $u$ and $u_3$, implying that $\mdom(G) \le \mdom(H)+2$. Thus, $\mdom(G) \le \mdom(H) + \alpha$ and $\w(G) \ge \w(H) + \beta$, where $\alpha = 2$ and $\beta = 28$, contradicting Fact~\ref{fact1}.

\smallskip
\emph{Case~2. $M_G$ contains a parallel \black\ and \rede.} Let $u$ and $v$ be two (distinct) vertices in $M_G$ that are joined by a \black\ edge and a \rede\ in $M_G$. By the girth condition, such a \rede\ in $M_G$ corresponds to a long \redp, say $x_1\ldots x_5$, in $G$. However, $ux_1\ldots x_5vu$ is a cycle of length~$7$ in $G$, a contradiction.

\smallskip
\emph{Case~3. $M_G$ contains a parallel \black\ and \gree.} Let $u$ and $v$ be two (distinct) vertices in $M_G$ that are joined by a \blae\ and a \gree\ in $M_G$. By the girth condition, such a \gree\ in $M_G$ corresponds to a long \grep\ in $G$ (with $u$ and $v$ as its ends). Let $P \colon x_1x_2x_3x_4$ denote this long \grep\ in $G$, where $u$ is adjacent to $x_1$. Let $u'$ be the neighbor of $u$ different from $x_1$ and $v$, and let $v'$ be the neighbor of $v$ different from $x_4$ and $u$. Let $H = G - (V(P) \cup \{u,v\})$, where we mark the vertex $v'$ in $H$. We note that $\w_G(u') - \w_H(u') \ge - 3$. Since $P$ contains at most one marked vertex, $\w(G) \ge \w(H) + 3\times 5 + 3\times 4 - 3 = \w(H) + 24$. Every \mds\ of $H$ can be extended to a \mds\ of $G$ by adding to it the vertices $v$ and $x_2$, implying that $\mdom(G) \le \mdom(H)+2$. Thus, $\mdom(G) \le \mdom(H) + \alpha$ and $\w(G) \ge \w(H) + \beta$, where $\alpha = 2$ and $\beta = 24$, contradicting Fact~\ref{fact1}.

\smallskip
\emph{Case~4. $M_G$ contains two \green\ parallel edges.} Let $u$ and $v$ be two (distinct) vertices in $M_G$ that are joined by two \green\ parallel edges in $M_G$. By the girth condition and the non-existence of cycles of length~$7$, both \grees\ in $M_G$ corresponds to a long \greps\ in $G$ (with $u$ and $v$ as their ends). Let $P \colon x_1x_2x_3x_4$ and $Q \colon y_1y_2y_3y_4$ be these two long \greps\ in $G$, where $u$ is adjacent to $x_1$ and $y_1$ in $G$. Let $H = G - (V(P) \cup V(Q) \cup \{u\})$, where we mark the neighbor, $x$ say, of $u$ different from $x_1$ and $y_1$ (possibly, $x = v$). We note that $\w_G(v) - \w_H(v) \ge - 4$. Since every long colored path contains at most one marked vertex, this implies that $\w(G) \ge \w(H) + 6\times 5 + 3\times 4 - 4 = \w(H) + 38$. Every \mds\ of $H$ can be extended to a \mds\ of $G$ by adding to it the vertices $u$, $x_3$ and $y_3$, implying that $\mdom(G) \le \mdom(H)+3$. Thus, $\mdom(G) \le \mdom(H) + \alpha$ and $\w(G) \ge \w(H) + \beta$, where $\alpha = 3$ and $\beta = 38$, contradicting Fact~\ref{fact1}.

\smallskip
\emph{Case~5. $M_G$ contains two \red\ parallel edges.} Let $u$ and $v$ be two (distinct) vertices in $M_G$ that are joined by two \red\ parallel edges in $M_G$. Let $P$ and $Q$ denote the paths in $G$ associated with these two \red\ parallel edges in $M_G$.  By the girth condition and since $G$ has no cycle of length~$7$, we note that $u$ and $v$ are not adjacent in $G$. Let $\ell$ be the number of long \red\ paths among $P$ and $Q$. By the girth condition, $\ell \ge 1$. Let $H = G - (V(P) \cup V(Q) \cup \{u,v\})$, where we mark the neighbors of $u$ and $v$ that do not belong to $P$ or $Q$. Since each of $P$ and $Q$ contain at most one marked vertex, $\w(G) \ge \w(H) + (2\ell+4) \times 5 + (\ell+2)\times 4  = \w(H) + 14\ell + 28$. Every \mds\ of $H$ can be extended to a \mds\ of $G$ by adding to it $\ell + 2$ vertices, namely the vertices $u$ and $v$, and the $\ell$ vertices at distance~$3$ from both $u$ and $v$ on the paths $P$ and $Q$, if such . Thus, $\mdom(G) \le \mdom(H) + \alpha$ and $\w(G) \ge \w(H) + \beta$, where $\alpha = \ell + 2$ and $\beta = 14\ell + 28$, contradicting Fact~\ref{fact1}.

\smallskip
\emph{Case~6. $M_G$ contains a parallel \red\ and \gree.} Let $u$ and $v$ be two (distinct) vertices in $M_G$ that are joined by a parallel \red\ and \gree\ in $M_G$. Let $P$ and $Q$ denote the paths in $G$ associated with these two edges in $M_G$. Let $\ell$ be the number of long paths among $P$ and $Q$. By the girth condition and since $G$ has no cycle of length~$7$, we note that $u$ and $v$ are not adjacent in $G$ and $\ell \ge 1$.  Let $H = G - (V(P) \cup V(Q) \cup \{u,v\})$, where we mark the neighbors of $u$ and $v$ that do not belong to $P$ or $Q$. Since each of $P$ and $Q$ contain at most one marked vertex, $\w(G) \ge \w(H) + (2\ell+3) \times 5 + (\ell+2) \times 4  = \w(H) + 14\ell + 23 \ge  \w(H) + 12\ell + 25$. Every \mds\ of $H$ can be extended to a \mds\ of $G$ by adding to it $\ell + 2$ vertices, namely the vertices $u$ and $v$, and the vertex $(uv)_3$ on the path $P$ if it is long and the vertex $(uv)_3$ on the path $Q$ if it is long. Thus, $\mdom(G) \le \mdom(H) + \alpha$ and $\w(G) \ge \w(H) + \beta$, where $\alpha = \ell + 2$ and $\beta = 12\ell + 25$, contradicting Fact~\ref{fact1}.

Since all six cases produce a contradiction, this completes the proof of Claim~\ref{colored-simple-graph}.~\smallqed

\medskip
By Claim~\ref{colored-simple-graph}, the colored multigraph $M_G$ contain no loops or multiple edges, and therefore is a colored cubic graph.

\subsection{The green edges form a Matching, as do the red edges}

In this section, we show that the set of \grees\ form a matching in $M_G$, and the set of \redes\ form a matching in $M_G$.

\begin{claim}\label{c:green-matching}
The colored multigraph $M_G$ does not contain two adjacent \grees, and so the \grees\ form a matching in $M_G$.
\end{claim}
\proof
Suppose, to the contrary, that $M_G$ contains a vertex $u$  incident with two \grees\, say $uv$ and $uw$. Let $H$ be the graph obtained from $G$ by removing $u$ and all vertices on the \green\ colored paths $(uv)$ and $(uw)$ in $G$. Let $\ell$ be the number of long paths among the two \green\ colored paths $(uv)$ and $(uw)$ in $G$. Since each long path contains at most one marked vertex, we note that $\w(G) \ge \w(H) + (2+2\ell) \times 5 + (\ell+1)\times 4 - 2 \ge \w(H) + 14 \ell + 12$. Every \mds\ of $H$ can be extended to a \mds\ of $G$ by adding to it $\ell + 1$ vertices, namely the vertex $u$ and the vertices $(uv)_3$ and $(uw)_3$ if they exist, where we mark the third neighbor of $u$ (that does not belong to the deleted $2$-paths). Thus, $\mdom(G) \le \mdom(H) + \alpha$ and $\w(G) \ge \w(H) + \beta$, where $\alpha = \ell + 1$ and $\beta = 14\ell + 12$, contradicting Fact~\ref{fact1}.~\smallqed

\medskip
Before proving that the set of \redes\ form a matching in $M_G$, we first show that there is no \green-\red-\green\ path in $M_G$ and that there is no vertex that is incident with two \black\ edges and is the end of a long \redp.

\begin{claim}\label{c:no-green-red-green}
The colored multigraph $M_G$ does not contain \green-\red-\green\ paths.
\end{claim}
\proof
Suppose, to the contrary, that $M_G$ contains a \green-\red-\green\ path $u_1u_2u_3u_4$, where $u_1u_2$ and $u_3u_4$ are \grees\ in $M_G$ and $u_2u_3$ is a \red\ edge in $M_G$.  Thus, $(u_1u_2)$ and $(u_3u_4)$ are \greps\ in $G$, while $(u_2u_3)$ is a \redp\ in $G$. By Claims~\ref{colored-simple-graph} and~\ref{c:green-matching}, the vertices $u_1$, $u_2$, $u_3$, $u_4$ are all distinct. Let $\ell$ be the number of long paths among $(u_1u_2)$, $(u_2u_3)$, and $(u_3u_4)$. Let $H$ be obtained from $G$ by removing the vertices $u_2$ and $u_3$ as well as all vertices on the $(u_1u_2)$, $(u_2u_3)$ and $(u_3u_4)$ $2$-paths, where we mark the third neighbors of $u_2$ and $u_3$ (that do not belong to  the deleted $2$-paths). We note that $\w(G) \ge \w(H) + (2\ell+4) \times 5 + (\ell+2)\times 4 - 2 = \w(H) + 14\ell + 26$. Every \mds\ of $H$ can be extended to a \mds\ of $G$ by adding to it $\ell + 2$ vertices, namely the vertices $u_2$ and $u_3$, and the vertices $(u_2u_1)_3$, $(u_2u_3)_3$, and $(u_3u_4)_3$ if they exist. Thus, $\mdom(G) \le \mdom(H) + \alpha$ and $\w(G) \ge \w(H) + \beta$, where $\alpha = \ell + 2$ and $\beta = 14\ell + 26$, contradicting Fact~\ref{fact1}.~\smallqed

\begin{figure}[htb]
\begin{center}
\input{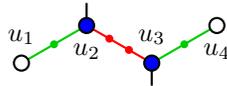}
\caption{A colored multigraph in the proof of Lemma~\ref{c:no-green-red-green}}
\label{f:no-green-red-green}
\end{center}
\end{figure}

\begin{claim}\label{no-long-red-black-black-star}
The graph $G$ does not contain a vertex that is incident with two \black\ edges and is the end of a long \redp.
\end{claim}
\proof
Suppose, to the contrary, that $G$ contains a vertex $u$ that is incident with two \black\ edges, say $uv_1$ and $uv_2$, and is the end of a long \redp\, say $(uv)$. By Claim~\ref{colored-simple-graph}, the vertices $v$, $v_1$ and $v_2$ are distinct. Let $H$ be obtained from $G$ by removing the vertex $u$ as well as all five vertices on the $(uv)$ $2$-path. We note that $v$, $v_1$ and $v_2$ all have degree~$3$ in $G$ and degree~$2$ in $H$, implying that $\w(G) \ge \w(H) + 4 \times 5 + 2 \times 4 - 3 = \w(H) + 25$. Every \mds\ of $H$ can be extended to a \mds\ of $G$ by adding to it the vertices $(uv)_1$ and $(uv)_4$. Thus, $\mdom(G) \le \mdom(H) + \alpha$ and $\w(G) \ge \w(H) + \beta$, where $\alpha = 2$ and $\beta = 25$, contradicting Fact~\ref{fact1}.~\smallqed

\medskip
We are now in a position to show that the \redes\  form a matching in $M_G$.

\begin{claim}\label{c:red-matching}
The colored multigraph $M_G$ does not contain two adjacent \redes, and so the \redes \, form a matching in $M_G$.
\end{claim}
\proof We proceed by a series of subclaims.

\begin{subclaim}\label{c:no-green-red-red-path}
The colored multigraph $M_G$ does not contain
\green-\red-\red\ paths.
\end{subclaim}
\proof
Suppose, to the contrary, that $M_G$ contains a \green-\red-\red\ path $u_1u_2u_3u_4$, where $u_1u_2$ is a \gree\ in $M_G$ and $u_2u_3$ and $u_3u_4$ are \redes\ in $M_G$. Thus, $(u_1u_2)$ is a \grep\ in $G$, and $(u_2u_3)$ and $(u_3u_4)$ are two \redps\ in $G$. By Claim~\ref{colored-simple-graph} the vertices $u_1,u_2,u_3$ are all distinct. Let $\ell$ be the number of long paths among $(u_1u_2)$ and $(u_2u_3)$. Let $H$ be obtained from $G$ by removing the vertices $u_2$, $u_3$, $(u_3u_4)_1$, as well as all vertices on the $(u_1u_2)$ and $(u_2u_3)$ $2$-paths, where we mark the third neighbors of $u_2$ and $u_3$ (that do not belong to the deleted $2$-paths). We note that $\w(G) \ge \w(H) + (2 \ell+4)\times 5 + (\ell+2)\times 4 - 1 - 3 = \w(H) + 14\ell + 24$. Every \mds\ of $H$ can be extended to a \mds\ of $G$ by adding to it $\ell + 2$ vertices, namely the vertices $u_2$ and $u_3$, and the vertices $(u_2u_1)_3$ and $(u_2u_3)_3$, if they exist. Thus, $\mdom(G) \le \mdom(H) + \alpha$ and $\w(G) \ge \w(H) + \beta$, where $\alpha = \ell + 2$ and $\beta = 14\ell + 24$, contradicting Fact~\ref{fact1}.~\smallqed

\begin{figure}[htb]
\begin{center}
\input{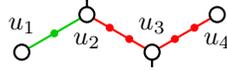}
\caption{A colored multigraph in the proof of Lemma~\ref{c:no-green-red-red-path}}
\label{f:no-green-red-red-path}
\end{center}
\end{figure}

\begin{subclaim}\label{c:no-red-cycle}
The colored multigraph $M_G$ does not contain \red\ cycles.
\end{subclaim}
\proof
Suppose, to the contrary, that $M_G$ contains a \red\ cycle $C \colon u_1u_2 \ldots u_ku_1$. Thus, the path $(u_1u_k)$ and the paths $(u_iu_{i+1})$ for $i \in [k-1]$, are all \redps\ in $G$. We note that every vertex of $C$ has degree~$3$ in $G$. Let $\ell$ be the number of long paths among $(u_1u_k)$ and $(u_iu_{i+1})$ for $i \in [k-1]$. Let $H$ be obtained from $G$ by removing the vertices in $V(C)$, as well as all vertices on the $(u_1u_k)$ and $(u_iu_{i+1})$ $2$-paths for $i \in [k-1]$, where we mark the third neighbors of each vertex $u_i$ for $i \in [k]$ (that do not belong to the deleted $2$-paths). We note that $\w(G) \ge \w(H) + (2k+2\ell)\times 5 + (k+\ell) \times 4 = \w(H) + 14(k + \ell)$. Every \mds\ of $H$ can be extended to a \mds\ of $G$ by adding to it $k + \ell$ vertices, namely the $k$ vertices in $V(C)$, and the $\ell$ vertices $(u_1u_k)_3$ and $(u_iu_{i+1})_3$ for $i \in [k-1]$, whichever exist. Thus, $\mdom(G) \le \mdom(H) + \alpha$ and $\w(G) \ge \w(H) + \beta$, where $\alpha = k + \ell$ and $\beta = 14(k + \ell)$, contradicting Fact~\ref{fact1}.~\smallqed

\begin{subclaim}\label{c:no-red-5-path}
The colored multigraph $M_G$ does not contain a
\red-\red-\red-\red\ path.
\end{subclaim}
\proof
Suppose, to the contrary, that $M_G$ contains a \red-\red-\red-\red\ path $u_1u_2u_3u_4u_5$, where $u_iu_{i+1}$ is a \rede\ in $M_G$ for $i \in [4]$.  Thus, $(u_1u_2)$, $(u_2u_3)$, $(u_3u_4)$ and $(u_4u_5)$ are all \redps\ in $G$. By Claim~\ref{colored-simple-graph}, the vertices $u_2$, $u_3$ and $u_4$ are distinct. Let $\ell$ be the number of long paths among $(u_2u_3)$ and $(u_3u_4)$. Let $H$ be obtained from $G$ by removing the vertices $u_2$, $u_3$, $u_4$, $(u_2u_1)_1$ and $(u_4u_5)_1$, as well as all vertices on the $(u_2u_3)$ and $(u_3u_4)$ $2$-paths, where we mark the third neighbors of $u_2$, $u_3$ and $u_4$ (that do not belong to the deleted $2$-paths). We note that $\w(G) \ge \w(H) + (6+2\ell)\times 5 + (3+\ell) \times 4 - 2 \times 3 = \w(H) + 14\ell + 36$. Every \mds\ of $H$ can be extended to a \mds\ of $G$ by adding to it $3 + \ell$ vertices, namely the vertices $u_2$, $u_3$ and $u_4$ together with the $\ell$ vertices $(u_2u_3)_3$ and $(u_3u_4)_3$, whichever exist. Thus, $\mdom(G) \le \mdom(H) + \alpha$ and $\w(G) \ge \w(H) + \beta$, where $\alpha = 3 + \ell$ and $\beta = 36 + 17\ell$, contradicting Fact~\ref{fact1}.~\smallqed

\begin{subclaim}\label{c:no-red-longred-red}
The colored multigraph $M_G$ does not contain
\red-\red-\red\ paths, where the central edge is a long \redp.
\end{subclaim}
\proof
Suppose, to the contrary, that $M_G$ contains a \red-\red-\red\ path $u_1u_2u_3u_4$, where $u_iu_{i+1}$ is a \red\ for $i \in [3]$ and where $u_2u_3$ is a long \red\ in $M_G$.  Thus, $(u_1u_2)$, $(u_2u_3)$, and $(u_3u_4)$ are all \redps\ in $G$. By Claims~\ref{colored-simple-graph} and~\ref{c:no-red-cycle}, the vertices $u_1$, $u_2$, $u_3$ and $u_4$ are distinct. Let $H$ be obtained from $G$ by removing the vertices $u_2$, $u_3$, $(u_2u_1)_1$ and $(u_3u_4)_1$, as well as all five vertices on the $(u_2u_3)$ $2$-path, where we mark the third neighbors of $u_2$ and $u_3$ (that do not belong to the $(u_1u_2)$, $(u_2u_3)$ and $(u_3u_4)$ paths). We note that $\w(G) \ge \w(H) + 6 \times 5 + 3 \times 4 - 2 \times 3 = \w(H) + 36$. Every \mds\ of $H$ can be extended to a \mds\ of $G$ by adding to it the three vertices $u_2$, $u_3$ and $(u_2u_3)_3$. Thus, $\mdom(G) \le \mdom(H) + \alpha$ and $\w(G) \ge \w(H) + \beta$, where $\alpha = 3$ and $\beta = 36$, contradicting Fact~\ref{fact1}.~\smallqed

\begin{subclaim}\label{c:no-red-red-black-triangle}
The colored multigraph $M_G$ does not contain
\red-\red-\black\ cycles.
\end{subclaim}
\proof
Suppose, to the contrary, that $M_G$ contains a \red-\red-\black\ cycle $T \colon uvwu$ where $uv$ and $uw$ are \redes\ and where $vw$ is a \black\ edge in $M_G$.  Thus, $(uv)$ and $(uw)$ are all \redps\ in $G$, while $vw$ is a \black\ edge in $G$. Since $G$ has no $7$-cycle, at least one of $uv$ and $uw$, say $uw$, is a long \red. Let $z$ be the third neighbor of $w$ in $M_G$ (not in the triangle $T$). Since $M_G$ is a colored multigraph, the edge $wz$ is colored with one of the three colors \green, \black\ or \red. By Claim~\ref{c:no-green-red-red-path}, $wz$ is not a \gree. By Claim~\ref{no-long-red-black-black-star}, $wz$ is not a \black\ edge. By Claim~\ref{c:no-red-longred-red}, $wz$ is not a \red. Since all three possibilities produce a contradiction, such a  \red-\red-\black\ cycle $T$ cannot exist.~\smallqed

\begin{subclaim}\label{c:no-double-star-w-2-red}
The colored multigraph $M_G$ does not contain a double star $S(2,2)$ with two \redes\,
one of which is the central edge, and the remaining edges \black.
\end{subclaim}
\proof
Suppose, to the contrary, that $M_G$ contains a double star $S(2,2)$ with vertex set $\{u_1,u_2,\ldots,u_6\}$, where $u_3$ and $u_4$ are the two central vertices and $u_1$ and $u_2$ the leaf-neighbors of $u_3$, and where $u_1u_3$ and $u_3u_4$ are \redes\ and $u_2u_3$, $u_4u_5$ and $u_4u_6$ are \black\ edges in $M_G$. Thus, $(u_1u_3)$ and $(u_3u_4)$ are \redps\ in $G$, and $u_2u_3$, $u_4u_5$ and $u_4u_6$ are \black\ edges in $G$. By  Claim~\ref{no-long-red-black-black-star}, the \redp\ $(u_3u_4)$ is short. By the girth condition, $u_2 \notin \{u_5,u_6\}$. By  Claim~\ref{c:no-red-red-black-triangle}, $u_1 \notin \{u_5,u_6\}$. By Claim~\ref{colored-simple-graph}, the vertices $u_1, u_2, \ldots, u_6$ are therefore distinct. By Claims~\ref{c:no-green-red-red-path},~\ref{no-long-red-black-black-star} and~\ref{c:no-red-longred-red}, the edge $u_1u_3$ is a short \red. Let $H$ be obtained from $G$ by removing the vertices $u_3$ and $u_4$, as well as the vertices on the $(u_1u_3)$ and $(u_3u_4)$ $2$-paths. We note that $\w(G) \ge \w(H) + 4 \times 5 + 2 \times 4 - 4 \times 1 = \w(H) + 24$. Every \mds\ of $H$ can be extended to a \mds\ of $G$ by adding to it the vertices $(u_3u_1)_1$ and $(u_3u_4)_2$. Thus, $\mdom(G) \le \mdom(H) + \alpha$ and $\w(G) \ge \w(H) + \beta$, where $\alpha = 2$ and $\beta = 24$, contradicting Fact~\ref{fact1}.~\smallqed

\begin{figure}
\begin{center}
\input{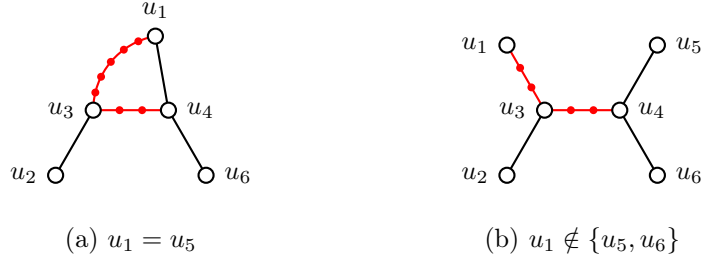}
\caption{Colored multigraphs in the proof of Lemma~\ref{c:no-double-star-w-2-red}}
\label{f:no-double-star-w-2-red}
\end{center}
\end{figure}

\begin{subclaim}\label{c:no-red-double-star}
The colored multigraph $M_G$ does not contain a double star $S(2,2)$ with all edges colored red.
\end{subclaim}
\proof
Suppose, to the contrary, that $M_G$ contains a double star $S(2,2)$ with vertex set $\{u_1,u_2,\ldots,u_6\}$, where $u_3$ and $u_4$ are the two central vertices and $u_1$ and $u_2$ the leaf-neighbors of $u_3$, and where all five edges are \redes\  in $M_G$. By  Claim~\ref{c:no-red-longred-red}, the \redp\ $(u_3u_4)$ is short. By Claim~\ref{c:no-red-cycle}, the vertices $u_1, u_2, \ldots, u_6$ are therefore distinct. Let $H$ be obtained from $G$ by removing the vertices $u_3$ and $u_4$, and all neighbors of $u_3$ and $u_4$. We note that $\w(G) \ge \w(H) + 6 \times 5 + 2 \times 4 - 4 \times 3 = \w(H) + 26$. Every \mds\ of $H$ can be extended to a \mds\ of $G$ by adding to it the vertices $u_3$ and $u_4$. Thus, $\mdom(G) \le \mdom(H) + \alpha$ and $\w(G) \ge \w(H) + \beta$, where $\alpha = 2$ and $\beta = 26$, contradicting Fact~\ref{fact1}.~\smallqed

\begin{subclaim}\label{c:no-red-K13-subd}
The colored multigraph $M_G$ does not contain a star $K_{1,3}$ with one edge subdivided and
 with all four edges colored red.
\end{subclaim}
\proof
Suppose, to the contrary, that $M_G$ contains a star $K_{1,3}$ with one edge subdivided with vertex set $\{u,u_1,u_2,u_3,u_4\}$, where $u$ is the vertex of degree~$3$, the vertices $u_1$, $u_2$ and $u_3$ are the three neighbors of $u$, the vertex $u_4$ is a leaf-neighbor of $u_3$, and where all four edges are \redes\ in $M_G$. By Claims~\ref{c:no-green-red-red-path} and~\ref{c:no-red-5-path}, the two edges incident with $u_4$ different from the \red\ $u_3u_4$ are both \black\ edges. Let $v_3$ be the third neighbor of $u_3$ (distinct from $u$ and $u_4$). By Claim~\ref{c:no-green-red-red-path}, $u_3v_3$ is not a \gree. By Claim~\ref{c:no-double-star-w-2-red}, $u_3v_3$ is not a \black\ edge. By Claim~\ref{c:no-red-double-star}, $u_3v_3$ is not a \red. Since all three possibilities produce a contradiction, such a  star with one edge subdivided and with all four edges colored red cannot exist.~\smallqed

\begin{subclaim}\label{c:no-red-4path}
The colored multigraph $M_G$ does not contain
\red-\red-\red\ paths.
\end{subclaim}
\proof
Suppose, to the contrary, that $M_G$ contains a \red-\red-\red\ path $u_1u_2u_3u_4$, where $u_iu_{i+1}$ is a \red\ for $i \in [3]$. By Claims~\ref{c:no-green-red-red-path} and~\ref{c:no-red-5-path}, the two edges incident with $u_4$ different from the \red\ $u_3u_4$ are both \black\ edges. Let $v_3$ be the third neighbor of $u_3$ (distinct from $u_2$ and $u_4$). By Claim~\ref{c:no-green-red-red-path}, $u_3v_3$ is not a \gree. By Claim~\ref{c:no-double-star-w-2-red}, $u_3v_3$ is not a \black\ edge. By Claim~\ref{c:no-red-K13-subd}, $u_3v_3$ is not a \red. Since all three possibilities produce a contradiction, such a  star with one edge subdivided and with all four edges colored red cannot exist.~\smallqed

\begin{subclaim}\label{c:no-red-star}
The colored multigraph $M_G$ does not contain a star with all three edges colored red.
\end{subclaim}
\proof
Suppose, to the contrary, that $M_G$ contains a star with vertex set $\{v,v_1,v_2,v_3\}$, where $v$ is the vertex of degree~$3$ and with all three edges colored red in $M_G$. By Claims~\ref{c:no-green-red-red-path} and~\ref{c:no-red-4path}, the two edges incident with $v_i$ different from the \red\ $vv_i$ are both \black\ edges for $i \in [3]$. Further, by Claim~\ref{c:no-red-5-path}, the three edges $vv_i$ are short \redes\ for $i \in [3]$.

Suppose that $v_1$ and $v_2$ have no common neighbor. Let $H$ be obtained from $G$ by removing the vertices $v$, $v_1$ and $v_2$, and all vertices on the $(vv_1)$, $(vv_2)$ and $(vv_3)$ $2$-paths. We note that $\w(G) \ge \w(H) + 6 \times 5 + 3 \times 4 - 5 \times 1 = \w(H) + 37$. Every \mds\ of $H$ can be extended to a \mds\ of $G$ by adding to it the three vertices $(vv_1)_2$, $(vv_2)_2$ and $(vv_3)_1$. Thus, $\mdom(G) \le \mdom(H) + \alpha$ and $\w(G) \ge \w(H) + \beta$, where $\alpha = 3$ and $\beta = 37$, contradicting Fact~\ref{fact1}. Hence, $v_1$ and $v_2$ have a common neighbor. Analogously, $v_1$ and $v_3$ have a common neighbor, and $v_2$ and $v_3$ have a common neighbor

Suppose that $v_1$, $v_2$ and $v_3$ all have the same common neighbor, say $w$. As observed earlier, the edges $wv_i$ are \black\ edges for all $i \in [3]$. In this case, we let $H$ be obtained from $G$ by removing the vertices $v$, $v_1$, $v_2$, $v_3$ and $w$ and all vertices on the $(vv_1)$, $(vv_2)$ and $(vv_3)$ $2$-paths. We note that $\w(G) \ge \w(H) + 6 \times 5 + 5 \times 4 = \w(H) + 50$. Every \mds\ of $H$ can be extended to a \mds\ of $G$ by adding to it the four vertices $v$, $v_1$, $v_2$ and $v_3$, where we mark the third neighbors of $v_1$, $v_2$ and $v_3$ (that were not deleted). Thus, $\mdom(G) \le \mdom(H) + \alpha$ and $\w(G) \ge \w(H) + \beta$, where $\alpha = 4$ and $\beta = 50$, contradicting Fact~\ref{fact1}. Therefore, $v_1$, $v_2$ and $v_3$ do not all have the same common neighbor.

Let $v_{ij}$ be the common neighbor of $v_i$ and $v_j$, where $1 \le i < j \le 3$. By our earlier observations, the vertices $v_{12}$, $v_{13}$ and $v_{23}$ are distinct. In this case, let $H$ be obtained from $G$ by removing the vertices $v$, $v_1$, $v_2$, $v_3$, $v_{12}$, $v_{13}$ and $v_{23}$ and all vertices on the $(vv_1)$, $(vv_2)$ and $(vv_3)$ $2$-paths. We note that $\w(G) \ge \w(H) + 6 \times 5 + 7 \times 4 - 3 \times 3 = \w(H) + 49$. Every \mds\ of $H$ can be extended to a \mds\ of $G$ by adding to it the four vertices $v$, $v_1$, $v_2$ and $v_3$. Thus, $\mdom(G) \le \mdom(H) + \alpha$ and $\w(G) \ge \w(H) + \beta$, where $\alpha = 4$ and $\beta = 49$, contradicting Fact~\ref{fact1}. Therefore, there is no star in $M_G$ with all three edges colored red.~\smallqed

\begin{subclaim}\label{c:no-red-red-green-star}
The colored multigraph $M_G$ does not contain a star with two \red\ edges and one \green\ edge.
\end{subclaim}
\proof
Suppose, to the contrary, that $M_G$ contains a star with vertex set $\{v,v_1,v_2,v_3\}$, where $v$ is the vertex of degree~$3$ and where the edge $vv_1$ and $vv_2$ are \redes\ and the edge $vv_3$ is a \gree. Let $H$ be obtained from $G$ by removing the vertex $v$ and its three neighbors in $G$. Every \mds\ of $H$ can be extended to a \mds\ of $G$ by adding to it the vertex $v$. We note that $\w(G) \ge \w(H) + 3 \times 5 + 1 \times 4 - 2 \times 3 - 1 = \w(H) + 12$. Thus, $\mdom(G) \le \mdom(H) + \alpha$ and $\w(G) \ge \w(H) + \beta$, where $\alpha = 1$ and $\beta = 12$, contradicting Fact~\ref{fact1}.~\smallqed

\medskip
We now return to the proof of Claim~\ref{c:red-matching} and show that the colored multigraph $M_G$ does not contain two adjacent \redes. Suppose, to the contrary, that $M_G$ contains a path $P \colon v_1v_2v_3$ where $v_1v_2$ and $v_2v_3$ are both \redes. By Claim~\ref{c:no-green-red-red-path} and~\ref{c:no-red-4path}, the two edges incident with $v_1$ (respectively, $v_3$) that do not belong to $P$ are both \black\ edges. Let $e$ be the edge incident with $v_2$ that does not belong to $P$. By Claim~\ref{c:no-double-star-w-2-red}, the edge $e$ cannot be a \black\ edge. By Claim~\ref{c:no-red-star}, the edge $e$ cannot be a \red. By Claim~\ref{c:no-red-red-green-star}, the edge $e$ cannot be a \gree. Since all three possibilities produce a contradiction, the colored multigraph $M_G$ does not contain \red-\red\ paths. Thus, the \redes\ form a matching in $M_G$. This completes the proof of Claim~\ref{c:red-matching}.~\smallqed

\medskip
By Claim~\ref{c:green-matching}, the \grees\ form a matching in $M_G$. By Claim~\ref{c:red-matching} the \redes\ form a matching in~$M_G$.

\subsection{There are no long edges}

In this section, we show that there are no long edges in the colored multigraph $M_G$; that is, every \grep\ is a maximal $2$-path in $G$ of order~$2$ and every \redp\ is a maximal $2$-path in $G$ of order~$1$.
As a consequence of Claim~\ref{c:red-matching}, we show that there is no long \rede\ in $M_G$.

\begin{claim}\label{c:no-long-red}
The colored multigraph $M_G$ does not contain any long \rede.
\end{claim}
\proof
Suppose, to the contrary, that $M_G$ contains a long \rede. By Claim~\ref{c:red-matching}, such a long \rede\ cannot be adjacent to another red edge. Further by Claim~\ref{c:no-green-red-green}, one end of the long \rede\ must be incident to two \blaes. But this contradicts Claim~\ref{no-long-red-black-black-star}.~\smallqed

\begin{claim}\label{c:gbr-triangle}
The colored multigraph $M_G$ does not contain a \green-\black-\red\ triangle.
\end{claim}
\proof
Suppose, to the contrary, that $M_G$ contain a \green-\black-\red\ triangle $T \colon uvwu$ where $uv$ is a \rede, $vw$ is a \gree, and $uw$ is a \blae. By Claim~\ref{c:no-long-red}, the \rede\ $uv$ is short. Let $\ell = 0$ if the \gree\ is short, and let $\ell = 1$ if the \gree\ is long. Let $H$ be obtained from $G$ by removing the vertices in $V(T)$, as well as all vertices on the $(uv)$ and $(vw)$ $2$-paths, where we mark the third neighbors of $u$ and $v$ that do not belong to the deleted $2$-paths. We note that $\w(G) \ge \w(H) + (3+2\ell) \times 5 + (3+\ell)\times 4 - 3 = \w(H) + 14\ell + 24$. Every \mds\ of $H$ can be extended to a \mds\ of $G$ by adding to it $\ell + 2$ vertices, namely the vertices $u$ and $v$ and, if $\ell = 1$, the vertex $vw_3$. Thus, $\mdom(G) \le \mdom(H) + \alpha$ and $\w(G) \ge \w(H) + \beta$, where $\alpha = \ell + 2$ and $\beta = 24 + 14\ell$, contradicting Fact~\ref{fact1}.~\smallqed

\begin{claim} \label{c:long-gre-rede}
No \rede\ is adjacent to a long \gree.
\end{claim}
\proof
Suppose, to the contrary, that $M_G$ contains a \rede\ $uv$ adjacent to a long \gree\ $vw$. By Claims~\ref{c:no-green-red-green} and~\ref{c:red-matching}, the two edges incident with $u$ different from the \rede\ $uv$ are \blaes. By Claims~\ref{c:green-matching} and~\ref{c:red-matching}, the third edge incident with $v$, different from the \rede\ $uv$ and the long \gree\ $vw$, is a black edge.

Suppose that the long \gree\ $vw$ contains no marked vertex. In this case, let $H$ be obtained from $G$ by removing the vertices $v$ and $w$, as well as all vertices on the $(uv)$ and $(vw)$ $2$-paths, where we mark the two neighbors of $w$ not on these $2$-paths. We note that $\w(G) \ge \w(H) + 6 \times 5 + 2 \times 4 - 2 = \w(H) + 36$. Every \mds\ of $H$ can be extended to a \mds\ of $G$ by adding to it the vertices $(uv)_2$, $(vw)_2$ and $w$. Thus, $\mdom(G) \le \mdom(H) + \alpha$ and $\w(G) \ge \w(H) + \beta$, where $\alpha = 3$ and $\beta = 36$, contradicting Fact~\ref{fact1}. Hence, the long \gree\ $vw$ contains one marked vertex, namely $(vw)_2$ or $(vw)_3$.

Suppose firstly that the vertex $(vw)_2$ is marked. In this case, let $H$ be obtained from $G$ by removing the vertices $v$ and $w$, as well the vertex $(uv)_2$ and all vertices on the $(vw)$ $2$-path, where we mark the neighbors of $v$ not on the $(uv)$ and $(vw)$ $2$-path. We note that at least one edge incident with $w$ is a black edge, and so $\w(G) \ge \w(H) + 4 \times 5 + 3 \times 4 - 2 \times 3 - 1 = \w(H) + 25$. Every \mds\ of $H$ can be extended to a \mds\ of $G$ by adding to it the vertices $v$ and $(vw)_4$. Thus, $\mdom(G) \le \mdom(H) + \alpha$ and $\w(G) \ge \w(H) + \beta$, where $\alpha = 2$ and $\beta = 25$, contradicting Fact~\ref{fact1}.

Suppose secondly that the vertex $(vw)_3$ is marked.  In this case, let $H$ be obtained from $G$ by removing the vertices $u$, $v$ and $w$, as well as all vertices on the $(uv)$ and $(vw)$ $2$-paths, where we mark the two neighbors of $w$ not on these $2$-paths. We note that $\w(G) \ge \w(H) + 5 \times 5 + 4 \times 4 - 3 \times 1 = \w(H) + 38$. Every \mds\ of $H$ can be extended to a \mds\ of $G$ by adding to it the vertices $(uv)_1$, $(vw)_1$ and $w$. Thus, $\mdom(G) \le \mdom(H) + \alpha$ and $\w(G) \ge \w(H) + \beta$, where $\alpha = 3$ and $\beta = 38$, contradicting Fact~\ref{fact1}. This completes the proof of Claim~\ref{c:long-gre-rede}.~\smallqed

\begin{claim}
\label{c:no-long-green}
The colored multigraph $M_G$ does not contain any long \gree.
\end{claim}
\proof
Suppose, to the contrary, that $M_G$ contains a long \gree\ $uv$.  By Claim~\ref{c:long-gre-rede}, every edge incident with $u$ or $v$ different from the long \gree\ $uv$ is a black edge. Let $H$ be obtained from $G$ by removing the vertices $u$ and $v$, as well as all vertices on the $(uv)$ and $(vw)$ $2$-paths

Suppose firstly that the long \gree\ $uv$ contains no marked vertex. In this case, we note that $\w(G) \ge \w(H) + 4 \times 5 + 2 \times 4 - 4 \times 1 = \w(H) + 24$. Every \mds\ of $H$ can be extended to a \mds\ of $G$ by adding to it the vertices $(uv)_1$ and $(uv)_2$. Thus, $\mdom(G) \le \mdom(H) + \alpha$ and $\w(G) \ge \w(H) + \beta$, where $\alpha = 2$ and $\beta = 24$, contradicting Fact~\ref{fact1}.

Suppose secondly that the long \gree\ $uv$ contains a marked vertex. Renaming the vertices $u$ and $v$ if necessary, we may assume that $(vw)_3$ is the marked vertex on the long \gree\ $uv$. In this case, we mark both neighbors of the vertex $w$ not on the long \gree\ $uv$, and note that $\w(G) \ge \w(H) + 3 \times 5 + 3 \times 4 - 2 \times 1 = \w(H) + 25$. Every \mds\ of $H$ can be extended to a \mds\ of $G$ by adding to it the vertices $(uv)_1$ and~$v$. Thus, $\mdom(G) \le \mdom(H) + \alpha$ and $\w(G) \ge \w(H) + \beta$, where $\alpha = 2$ and $\beta = 25$, contradicting Fact~\ref{fact1}. This completes the proof of Claim~\ref{c:no-long-green}.~\smallqed

\medskip
By Claims~\ref{c:no-long-red} and~\ref{c:no-long-green}, there are no long  \grees\ or \redes\ in the colored multigraph $M_G$. Thus by our earlier observations, there are no marked vertices.


\subsection{Key structural claims}

In this section, we prove several structural results that will be helpful when proving our main result.

\begin{claim}\label{c:claim23}
In the colored multigraph $M_G$, there is no \green-\black-\green\ path $uu'vv'$ where $u$ and $v$ are adjacent to the centers $u^*$ and $v^*$ of black stars.
\end{claim}

\begin{figure}[hbt]
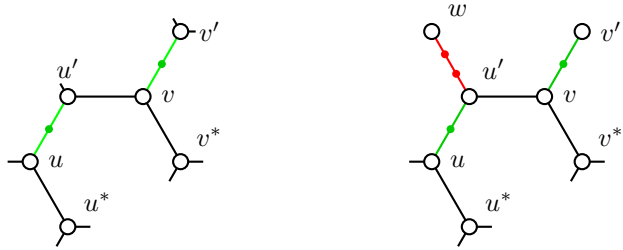

\begin{center}
\input{dotted-apart.tex} \hspace{2cm} \input{dotted-apart-red.tex}
\caption{The two cases of Claim~\ref{c:claim23}}
\label{f:dotted-apart}
\end{center}
\end{figure}

\proof
Suppose, to the contrary, that $M_G$ contains such a configuration. We consider two cases. The first case is when the third edge attached to $u'$ is a black edge. Is this case, let $H$ be obtained from $G$ by removing the vertices $u$, $v$, $u^*$, $v^*$, $u'$, and the vertices $(uu')_1$ and $(vv')_1$. We mark the third vertex of $u$ and cut all the six other outgoing edges. Since the graph $G$ does not contain a $7$-cycle, the vertices $u^*$ and $v^*$ do not have a common neighbor. We note that $\w(G) \ge \w(H) + 5\times 4 + 2 \times 5 - 6 = \w(H) + 24$. Every MD-set of $H$ can be extended to a MD-set of $G$ by adding the vertices $u$ and $v$. Thus, $\mdom(G) \le \mdom(H) + \alpha$ and $\w(G) \ge \w(H) + \beta$, where $\alpha = 2$ and $\beta = 24$, contradicting Fact~\ref{fact1}.

Now, consider the case when $u'$ is incident to a \rede\ $u'w$. In this case, let $H$ be obtained from $G$ by removing the vertices $u$, $v$, $u^*$, $v^*$, $u'$, $w$, and the vertices $(uu')_1$,  $(vv')_1$, $(u'w)_1$ and $(u'w)_2$. We mark the third vertex of $u$ and cut all the seven other outgoing edges. We note that these seven edges that are cut when constructing $H$ do not make a vertex of degree~$3$ drop to degree~$1$. Indeed since there is no $4$-cycle, the vertices $v'$ and $v^*$ are not adjacent and since there is no $7$-cycle, the vertex $v'$ is adjacent to neither $w$ nor $u^*$. Moreover since there is no $7$-cycle, the vertices $u^*$ and $v^*$ have no common neighbor and the vertices $w$ and $v^*$ have no common neighbor. Since there is no $8$-cycle, the vertices $u^*$ and $w$ have no common neighbor. Thus forbidding $7$- and $8$-cycles guarantees we do not create a vertex of degree~$1$ when constructing $H$. We therefore infer that $\w(G) \ge \w(H) + 6 \times 4 + 4 \times 5 - 7 = \w(H) + 37$. Every MD-set of $H$ can be extended to a MD-set of $G$ by adding the vertices $u$, $v$ and $(u'w)_2$. Thus, $\mdom(G) \le \mdom(H) + \alpha$ and $\w(G) \ge \w(H) + \beta$, where $\alpha = 3$ and $\beta = 36$, contradicting Fact~\ref{fact1}.~\smallqed


\begin{claim}\label{c:consecutives-reds}
In the colored multigraph $M_G$, there is no \green-\black-\green\ path $uu'vv'$ where $u'$ and $v'$ are incident to \redes\ $uu^\dagger$ and $vv^\dagger$ and where the third edge attached to $v$ is a black edge.
\end{claim}

\begin{figure}[hb]
\begin{center}
\input{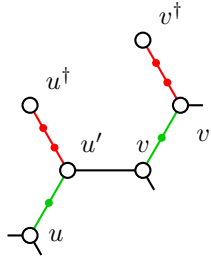}
\caption{A colored multigraph in the proof of Claim~\ref{c:consecutives-reds}}
\label{f:consecutives-reds}
\end{center}
\end{figure}

\proof
Suppose, to the contrary, that $M_G$ contains such a configuration. By Claim~\ref{c:gbr-triangle}, the edges $u'u^\dagger$ and $v'v^\dagger$ are distinct.
Let $H$ be obtained from $G$ by removing the vertices $u'$, $v$, $v'$, $(u'u^\dagger)_1$, $(v'v^\dagger)_1$, $(uu')_1$ and $(vv')_1$. We mark the third neighbor of $v'$, and cut all the four other outgoing edges. We note that $\w(G) \ge \w(H) + 3\times 4 + 4 \times 5 - 2 \times 3 - 2 \times 1 = \w(H) + 24$. Every MD-set of $H$ can be extended to a MD-set of $G$ by adding the vertices $u'$ and $v'$. Thus, $\mdom(G) \le \mdom(H) + \alpha$ and $\w(G) \ge \w(H) + \beta$, where $\alpha = 2$ and $\beta = 24$, contradicting Fact~\ref{fact1}.~\smallqed

\begin{claim}\label{c:pre-24-6cycle}
In the colored multigraph $M_G$, there is no $6$-cycle $u'vv'ww^*u'~$ (where $vv'$ is a \gree),
  such that $w$ is incident to a \gree, and $w^*$ is the center of a black star.
\end{claim}

\begin{figure}[hb]
\begin{center}
\input{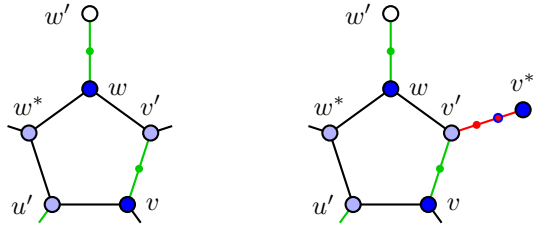}
\caption{Colored multigraphs in the proof of Claim~\ref{c:pre-24-6cycle}}
\label{f:pre-24-6cycle}
\end{center}
\end{figure}

\proof
Suppose, to the contrary, that $M_G$ contains such a configuration. Let $w'$ be the other extremity of the \gree\ incident to $w$.
Let $H$ be obtained from $G$ by removing the vertices $u'$, $v$, $v'$, $w$, $x$, $(vv')_1$ and $(ww')_1$.
We mark the third neighbor of $v$, and cut all four other outgoing edges,
namely the edge $(ww')_1w'$ and the edges incident to $w^*$, $u'$, and $v'$.
Possibly, the third edge incident to $u'$ is a \gree.
Moreover, $v'$ may possibly be incident to a \rede.

Suppose firstly that $v'$ is not incident to a \rede.
In this case, we note that
$\w(G) \ge \w(H) + 5 \times 4 + 2 \times 5 - 1 \times 3 - 3 \times 1 = \w(H) + 24$.
Every MD-set of $H$ can be extended to a MD-set of $G$ by adding the vertices $v$ and $w$.
(This is illustrated in the left hand figure of Figure~\ref{f:pre-24-6cycle} where the vertices in the  MD-set are dark blue and the other vertices of degree~$3$ in $G$ that were deleted when constructing $H$ are colored light blue.) Thus, $\mdom(G) \le \mdom(H) + \alpha$ and $\w(G) \ge \w(H) + \beta$, where $\alpha = 2$ and $\beta = 24$, contradicting Fact~\ref{fact1}.

Hence the vertex  $v'$ is incident to a \rede, say $vv^*$. Since there is no $7$-cycle,
the vertex $v^*$ has no common neighbor with $w^*$ and is not adjacent to $w'$. Moreover since there is no $8$-cycle, the vertex $v^*$ is not adjacent to $u'$. By Claim~\ref{c:gbr-triangle}, the vertex $v^*$ is not adjacent to~$v$. In this case, we let $H$ be obtained from $G$ by removing the vertices $u'$, $v$, $v'$, $w$, $w^*$, $(vv')_1$, $(ww')_1$, $(vv^*)_1$, $(vv^*)_2$, and $v^*$. By Claim~\ref{c:no-green-red-green}, the two edges incident to $v^*$ different that do not belong to the \rede\ $vv^*$ are \blaes. We mark the third neighbor of $v$, and cut all six other outgoing edges, namely the edge $(ww')_1w'$ and the edges incident to $w^*$, $u'$, and $v^*$. Every MD-set of $H$ can be extended to a MD-set of $G$ by adding the vertices $v$, $w$ and $v^*$. (This is illustrated in the right hand figure of Figure~\ref{f:pre-24-6cycle} where the vertices in the  MD-set are dark blue and the other vertices of degree~$3$ in $G$ that were deleted when constructing $H$ are colored light blue.) In this case, we note that $\w(G) \ge \w(H) + 6\times 4 + 4 \times 5 - 4 \times 1 - 1 \times 3 = \w(H) + 37$. Thus, $\mdom(G) \le \mdom(H) + \alpha$ and $\w(G) \ge \w(H) + \beta$, where $\alpha = 3$ and $\beta = 37$, contradicting Fact~\ref{fact1}.~\smallqed

\begin{claim}\label{c:pre-24-6cycle2}
In the colored multigraph $M_G$, if a 6-cycle $xu'vv'wx$ (where $vv'$ is a \gree) is such that the third neighbor $w^*$ of $w$ is the center of a black star, then the outgoing edge of $x$ is a \rede.
\end{claim}

\begin{figure}[htb]
\begin{center}
\input{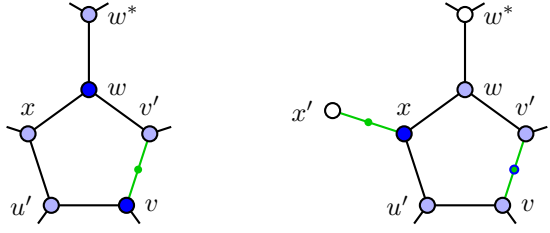}
\caption{Colored multigraphs in the proof of Claim~\ref{c:pre-24-6cycle2}}
\label{f:pre-24-6cycle2}
\end{center}
\end{figure}

\proof
Suppose, to the contrary, that $M_G$ contains such a configuration where the outgoing edge of $x$ is either a \blae\  or a \gree. By Claim~\ref{c:pre-24-6cycle}, the outgoing edge from $u'$ is a black edge.

Suppose firstly that the outgoing edge from $x$ is a \blae. Let $H$ be obtained from $G$ by removing the vertices $u'$, $v$, $v'$, $w$, $x$, $w^*$, and $(vv')_1$. We mark the third neighbor of $v$, and cut all five other outgoing edges, namely the edges incident to $x$, $u'$, $v'$ and $w^*$. In this case, we note that $\w(G) \ge \w(H) + 6 \times 4 + 1 \times 5 - 5 \times 1 = \w(H) + 24$. Every MD-set of $H$ can be extended to a MD-set of $G$ by adding the vertices $v$ and $w$. (This is illustrated in the left hand figure of Figure~\ref{f:pre-24-6cycle2} where the vertices in the  MD-set are dark blue and the other vertices of degree~$3$ in $G$ that were deleted when constructing $H$ are colored light blue.) Thus, $\mdom(G) \le \mdom(H) + \alpha$ and $\w(G) \ge \w(H) + \beta$, where $\alpha = 2$ and $\beta = 24$, contradicting Fact~\ref{fact1}.

Hence the outgoing edge from $x$ is a \gree, say $xx'$.  In this case, let $H$ be obtained from $G$ by removing the vertices $u'$, $v$, $v'$, $w$, $x$, $(vv')_1$ and $(xx')_1$. Since there is no $5$-cycle, the vertices $x'$ and $v'$ are not adjacent, and since there is no $8$-cycle, they do not have a common neighbor. Moreover since there is no $4$-cycle, the vertices $x'$ and $u'$ are not adjacent, and since there is no $5$-cycle, they do not have a common neighbor. Since there is no $5$-cycle, the vertices $x'$ and $v$ are not adjacent. Since there is no $5$-cycle, the vertices $u'$ and $v'$ do not have a common neighbor. We cut all five other outgoing edges, namely the edge $(xx')_1x'$ and the edges incident to $u'$, $v$, $v'$ and $w$. In this case, we note that $\w(G) \ge \w(H) + 5 \times 4 + 2 \times 5 - 5 \times 1 = \w(H) + 25$. Every MD-set of $H$ can be extended to a MD-set of $G$ by adding the vertices $x$ and $(vv')_1$. (This is illustrated in the right hand figure of Figure~\ref{f:pre-24-6cycle2} where the vertices in the  MD-set are dark blue and the other vertices of degree~$3$ in $G$ that were deleted when constructing $H$ are colored light blue.) Thus, $\mdom(G) \le \mdom(H) + \alpha$ and $\w(G) \ge \w(H) + \beta$, where $\alpha = 2$ and $\beta = 25$, contradicting Fact~\ref{fact1}.~\smallqed

\begin{claim}\label{c:pre-24-10cycle}
In the colored multigraph $M_G$, there is no 10-cycle $C \colon uu'vv'ww'w'^*u$ where the following properties hold. \\[-20pt]
\begin{enumerate}
\item[$\bullet$]  $uu'$, $vv'$ and $ww'$ are \grees,
\item[$\bullet$] $u'$ and $w$ are adjacent to the centers $u'^*$ and $w^*$, respectively, of black stars,
\item[$\bullet$] $w'^*$ is the center of a black star, and
\item[$\bullet$] $w'$ is adjacent to a vertex $x$ which is the end of a \gree\ $xx'$.
\end{enumerate}
\end{claim}

\begin{figure}[htb]
\begin{center}
\input{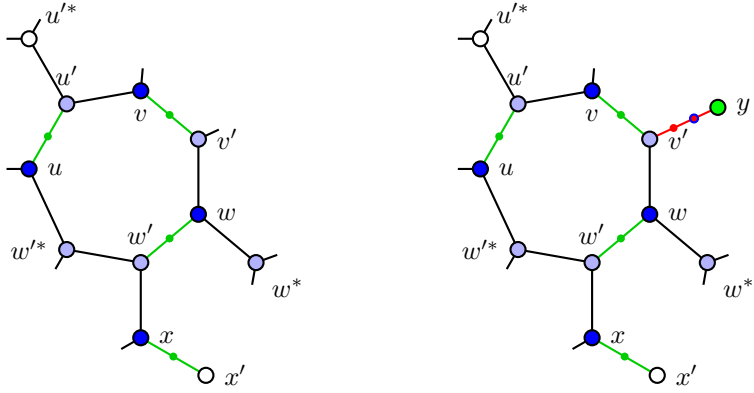}
\caption{Colored multigraphs in the proof of Claim~\ref{c:pre-24-10cycle}}
\label{f:pre-24-10cycle}
\end{center}
\end{figure}

\proof
Suppose, to the contrary, that $M_G$ contains such a configuration. Let $x$ be the third neighbor of $w'$ that does not belong to the cycle $C$. We note that the edge incident with $v'$ that does not belong to the cycle $C$ is either a \blae\ or a \rede. We consider the two cases in turn.

Suppose firstly that the edge incident with $v'$ that does not belong to the cycle $C$ is a \blae. By Claim~\ref{c:claim23}, the vertex $x$ is not the center of a black star. By our earlier properties and observations, the vertex $x$ is not incident with a \rede. Hence, the vertex $x$ is incident with a \gree, say $xx'$. In this case, let $H$ be obtained from $G$ by removing all the vertices of the cycle $C$, and removing the vertices $w^*$, $x$ and $(xx')_1$. We mark the neighbors of $u$ and $v$ that do not belong to the cycle, and we mark the neighbor of $x$ different from $w'$ and $(xx')_1$, and we cut all other outgoing edges.

Since there is no $5$-cycle, the vertices $x'$ and $w'^*$ are not adjacent and the vertex $u'^*$ is adjacent to neither $v'$ nor $w'^*$. Since there is no $7$-cycle, the vertices $u'^*$ and $w^*$ are not adjacent and the vertex $x'$ is adjacent to neither $w^*$ nor $v'$. Since $u'^*$ is the center of a black star, we note that the vertices $x'$ and $u'^*$ are distinct. Since there is no $4$-cycle, the vertices $w^*$ and $v'$ have no common neighbor. Since there is no $8$-cycle, the vertices $w'^*$ and $v'$ have no common neighbor. We note that the vertices $w'^*$ and $w^*$ may possibly have a common neighbor. However if a common neighbor of $w'^*$ and $w^*$ exists, then by Claim~\ref{c:pre-24-6cycle2} it is incident with a \rede.

From these structural properties we infer that if $w'^*$ and $w^*$ have no common neighbor, then $\w(G) \ge \w(H) + 9 \times 4 + 4 \times 5 - 6 \times 1 = \w(H) + 50$, while if $w'^*$ and $w^*$ do have a common neighbor, then $\w(G) \ge \w(H) + 9\times 4 + 4 \times 5 - 5 \times 1 - 1 \times 3 = \w(H) + 48$. Every MD-set of $H$ can be extended to a MD-set of $G$ by adding to it the vertices $u$, $v$, $w$ and $x$. (This is illustrated in the left hand figure of Figure~\ref{f:pre-24-10cycle} where the vertices in the  MD-set are dark blue and the other vertices of degree~$3$ in $G$ that were deleted when constructing $H$ are colored light blue.) Thus, $\mdom(G) \le \mdom(H) + \alpha$ and $\w(G) \ge \w(H) + \beta$, where $\alpha = 4$ and $\beta \ge 48$, contradicting Fact~\ref{fact1}.

Hence the edge incident with $v'$ that does not belong to the cycle $C$ is a \rede, say $v'y$. In this case, let $H$ be obtained from $G$ by removing all the vertices of the cycle $C$, and removing the vertices $w^*$, $x$, $y$, $(xx')_1$, $(v'y)_1$  and $(v'y)_2$. We mark the neighbors of $u$ and $v$ that do not belong to the cycle, we mark the neighbor of $x$ different from $w'$ and $(xx')_1$, we mark the vertex $y$, and we cut all other outgoing edges. Analogously as before, we infer that if $w'^*$ and $w^*$ have no common neighbor, then $\w(G) \ge \w(H) + 9 \times 4 + 6 \times 5 - 4 \times 1 = \w(H) + 61$, while if $w'^*$ and $w^*$ do have a common neighbor, then $\w(G) \ge \w(H) + 9\times 4 + 6 \times 5 - 3 \times 1 - 1 \times 3 = \w(H) + 60$. Every MD-set of $H$ can be extended to a MD-set of $G$ by adding to it the vertices $u$, $v$, $w$, $x$ and $(v'y)_2$. (This is illustrated in the right hand figure of Figure~\ref{f:pre-24-10cycle} where the vertices in the  MD-set are dark blue and the other vertices of degree~$3$ in $G$ that were deleted when constructing $H$ are colored light blue.) Thus, $\mdom(G) \le \mdom(H) + \alpha$ and $\w(G) \ge \w(H) + \beta$, where $\alpha = 5$ and $\beta \ge 60$, contradicting Fact~\ref{fact1}.~\smallqed

\begin{claim}\label{c:pre-24-other10cycle}
In the colored multigraph $M_G$, there is no $10$-cycle $C \colon u'vv'ww'w'^*yu'^* u'$ where the following properties hold. \\[-20pt]
\begin{enumerate}
\item[$\bullet$]  $vv'$ and $ww'$ are \grees,
\item[$\bullet$] $u'$ is incident with a \gree\ $uu'$,
\item[$\bullet$] $u'^*$ and $w'^*$ are center of black stars and have a common neighbor $y$,
\item[$\bullet$] $w$ is adjacent to the center $w^*$ of a black star, and
\item[$\bullet$] the neighbor $x$ of $w'$ that is not on the cycle $C$ is incident with a \gree\ $xx'$.
\end{enumerate}
\end{claim}

\begin{figure}[htb]
\begin{center}
\input{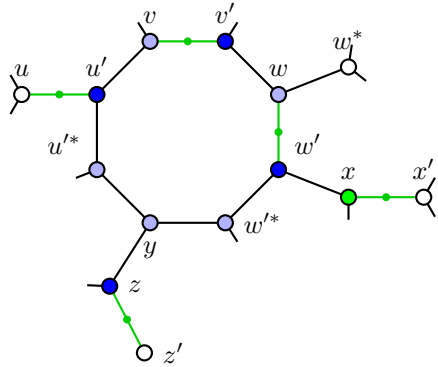}
\caption{A colored multigraph in the first case of Claim~\ref{c:pre-24-other10cycle}}
\label{f:pre-24-other10cycle}
\end{center}
\end{figure}

\proof
Suppose, to the contrary, that $M_G$ contains such a configuration. Since there is no $4$-cycle and no $5$-cycle, the vertex $v$ is not adjacent to $u$ and $w^*$, respectively. Moreover since there is no $4$-cycle, the vertices $v$ and $u'^*$ have no common neighbor, and since there is no $8$-cycle, the vertices $v$ and $w'^*$ have no common neighbor. We note, however, that the vertices $u$ and $w'^*$ may possibly be adjacent, in which they belong to a common $6$-cycle. Since there is no $4$-cycle, the vertex $u'^*$ is not adjacent to $u$ and has no common neighbor with $w'^*$. Let $e = yz$ be the edge in $M_G$ incident with $y$ that does not belong to the cycle $C$.

Suppose firstly that the edge $e$ is a \blae\ and that the vertex $z$ is incident to a \gree, say $zz'$, as illustrated in Figure~\ref{f:pre-24-other10cycle}. Since there is no $5$-cycle, the vertex $z'$ is adjacent to neither $u'^*$ nor $w'^*$. Since there is no $7$-cycle, the vertices $v$ and $z'$ are not adjacent. Since $w^*$ is the center of a black star, it cannot be equal to $u$ or to $z'$. Moreover since there is no $5$-cycle and no $7$-cycle, the vertex $w^*$ is adjacent to neither $w'^*$ nor $u'^*$, respectively. We also note that the vertices $u$ and $z'$ are distinct since the \grees\ form a matching.

Let $H$ be obtained from $G$ by removing all the vertices of the cycle $C$, and removing the vertices $z$, $(zz')_1$, and $(uu')_1$. We mark the neighbors of $v'$ and $w'$ that do not belong to the cycle $C$, and we mark the neighbor of $z$ different from $y'$ and $(zz')_1$, and we cut all other six outgoing edges. We note that when constructing $H$, we create at most one vertex of degree~$1$, which occurs when the vertices $u$ and $w'^*$ are adjacent. If we create no vertex of degree~$1$, then $\w(G) \ge \w(H) + 9 \times 4 + 4 \times 5 - 6 \times 1 = \w(H) + 50$. On the other hand, if we create a vertex of degree~$1$ (in the case when $u$ and $w'^*$ are adjacent), then $\w(G) \ge \w(H) + 9 \times 4 + 4 \times 5 - 5 \times 1 - 1 \times 3 = \w(H) + 48$. Every MD-set of $H$ can be extended to a MD-set of $G$
by adding to it the vertices $u'$, $v'$, $w'$ and $z$. (This is illustrated in the left hand figure of Figure~\ref{f:pre-24-other10cycle} where the vertices in the  MD-set are dark blue and the other vertices of degree~$3$ in $G$ that were deleted when constructing $H$ are colored light blue.) Thus, $\mdom(G) \le \mdom(H) + \alpha$ and $\w(G) \ge \w(H) + \beta$, where $\alpha = 4$ and $\beta \ge 48$, contradicting Fact~\ref{fact1}.

Hence the edge $e$ is a \rede\ or a \gree\ or a \blae\ with the vertex $z$ the center of a black star or a \blae\ with the vertex $z$ incident to an additional \blae\ and to a \rede, say $zz'$. This gives rise to four cases that may occur, as illustrated in Figure~\ref{f:pre-24-other10cycle-2}.

\begin{figure}[htb]
\begin{center}
\input{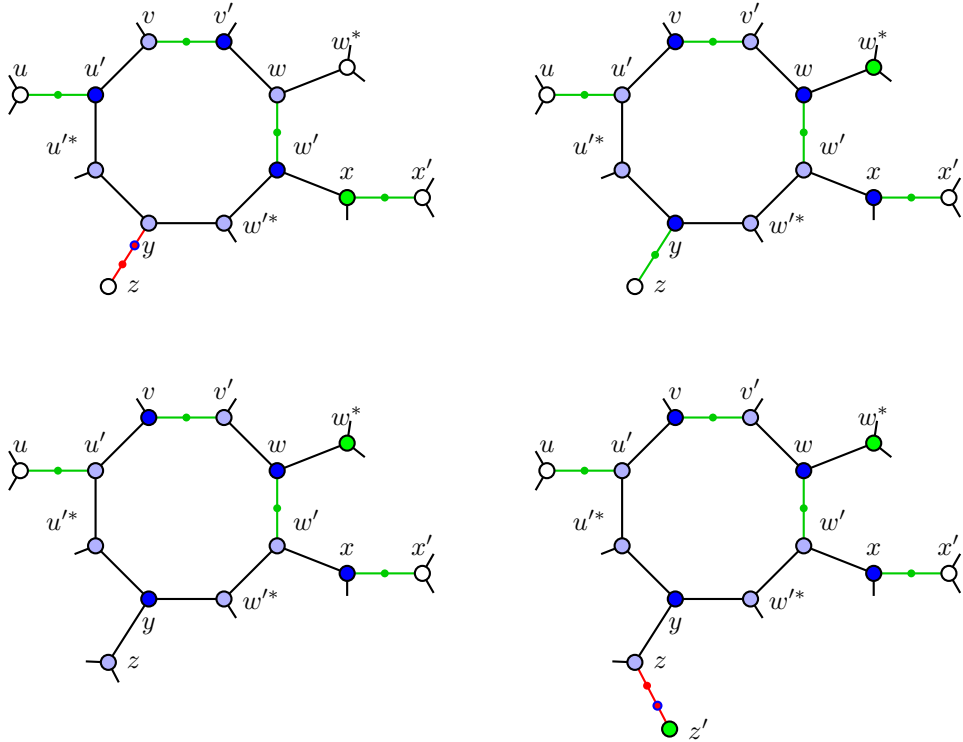}
\caption{Colored multigraph in the second case of Claim~\ref{c:pre-24-other10cycle}}
\label{f:pre-24-other10cycle-2}
\end{center}
\end{figure}

Since the vertex $(uu')_1$ is a vertex of degree~$2$, it is distinct from $x'$ and is not adjacent to $v'$, $w^*$, $w'^*$ or $u'^*$. Since there is no $4$-cycle, the vertices $u'^*$ and $w'^*$ have no common neighbor different from $y$ and the vertices $v'$ and $w^*$ have no common neighbor different from $w$. Since there is no $5$-cycle, the vertex $x'$ is not adjacent to $w'^*$. Since there is no $7$-cycle, the vertex $x'$ is not adjacent to $u'^*$, $v'$ or $w^*$. Since there is no $8$-cycle, the vertices $v'$ and $w'^*$ have no common neighbor, the vertices $u'^*$ and $w^*$ have no common neighbor, and the vertices $v'$ and $u'^*$ have no common neighbor. However, we note that the vertices $w^*$ and $w'^*$ may possibly have a common neighbor.

Suppose that the edge $e$ is a \rede. In this case, let $H$ be obtained from $G$ by removing all the vertices of the cycle $C$, and removing the vertices $x$, $(uu')_1$, $(yz)_1$ and $(yz)_2$. We mark the neighbors of $v'$ and $w'$ that does not belong to the cycle $C$, and we cut all other six outgoing edges. If the vertices $w^*$ and $w'^*$ have no common neighbor, then  $\w(G) \ge \w(H) + 8 \times 4 + 5 \times 5 - 6 \times 1 = \w(H) + 51$, while if the vertices $w^*$ and $w'^*$ have a common neighbor, then  $\w(G) \ge \w(H) + 8 \times 4 + 5 \times 5 - 5 \times 1 + 1 \times 3 = \w(H) + 48$.  Every MD-set of $H$ can be extended to a MD-set of $G$ by adding to it the vertices $u'$, $v'$, $w'$ and $(yz)_1$. (This is illustrated in the top left hand figure of Figure~\ref{f:pre-24-other10cycle-2} where the vertices in the  MD-set are dark blue and the other vertices of degree~$3$ in $G$ that were deleted when constructing $H$ are colored light blue.) Thus, $\mdom(G) \le \mdom(H) + \alpha$ and $\w(G) \ge \w(H) + \beta$, where $\alpha = 4$ and $\beta \ge 48$, contradicting Fact~\ref{fact1}. Hence, the edge $e$ is not a \rede.

In what follows suppose that the vertices $w^*$ and $w'^*$ do not have a common neighbor.

Suppose that the edge $e$ is a \gree. In this case, let $H$ be obtained from $G$ by removing all the vertices of the cycle $C$, and removing the vertices $x$, $(xx')_1$, $w^*$ and $(yz)_1$. We mark the neighbor of $v$ that does not belong to the cycle $C$, and we mark the neighbor of $x$ different from $w'$ and $(xx')_1$, and we cut all other eight outgoing edges. In this case, $\w(G) \ge \w(H) + 10 \times 4 + 4 \times 5 - 7 \times 1 - 1 \times 3 = \w(H) + 50$. Every MD-set of $H$ can be extended to a MD-set of $G$ by adding to it the vertices $v$, $w$, $x$ and $y$. (This is illustrated in the top right hand figure of Figure~\ref{f:pre-24-other10cycle-2} where the vertices in the  MD-set are dark blue and the other vertices of degree~$3$ in $G$ that were deleted when constructing $H$ are colored light blue.) Thus, $\mdom(G) \le \mdom(H) + \alpha$ and $\w(G) \ge \w(H) + \beta$, where $\alpha = 4$ and $\beta \ge 50$, contradicting Fact~\ref{fact1}. Hence, the edge $e$ is not a \gree.

Suppose that the edge $e$ is a \blae\ with the vertex $z$ the center of a black star. In this case, let $H$ be obtained from $G$ by removing all the vertices of the cycle $C$, and removing the vertices $x$, $(xx')_1$, $w^*$ and $z$. We mark the neighbor of $v$ that does not belong to the cycle $C$, and we mark the neighbor of $x$ different from $w'$ and $(xx')_1$, and we cut all other seven outgoing edges. In this case, $\w(G) \ge \w(H) + 11 \times 4 + 3 \times 5 - 8 \times 1 - 1 \times 3 = \w(H) + 48$. Every MD-set of $H$ can be extended to a MD-set of $G$ by adding to it the vertices $v$, $w$, $x$ and $y$. (This is illustrated in the bottom left hand figure of Figure~\ref{f:pre-24-other10cycle-2} where the vertices in the  MD-set are dark blue and the other vertices of degree~$3$ in $G$ that were deleted when constructing $H$ are colored light blue.) Thus, $\mdom(G) \le \mdom(H) + \alpha$ and $\w(G) \ge \w(H) + \beta$, where $\alpha = 4$ and $\beta \ge 48$, contradicting Fact~\ref{fact1}.

Suppose that the edge $e$ is a \blae\ with the vertex $z$ incident to an additional \blae\ and to a \rede, say $zz'$. In this case, let $H$ be obtained from $G$ by removing all the vertices of the cycle $C$, and removing the vertices $x$, $(xx')_1$, $w^*$, $z$, $z'$, $(zz')_1$ and $(zz')_2$. We mark the neighbor of $v$ that does not belong to the cycle $C$, and we mark the neighbor of $x$ different from $w'$ and $(xx')_1$, and we cut all other ten outgoing edges. In this case, $\w(G) \ge \w(H) + 12 \times 4 + 5 \times 5 - 9 \times 1 - 1 \times 3 = \w(H) + 61$. Every MD-set of $H$ can be extended to a MD-set of $G$ by adding to it the vertices $v$, $w$, $x$, $y$ and $(zz')_2$. (This is illustrated in the bottom right hand figure of Figure~\ref{f:pre-24-other10cycle-2} where the vertices in the  MD-set are dark blue and the other vertices of degree~$3$ in $G$ that were deleted when constructing $H$ are colored light blue.) Thus, $\mdom(G) \le \mdom(H) + \alpha$ and $\w(G) \ge \w(H) + \beta$, where $\alpha = 4$ and $\beta \ge 48$, contradicting Fact~\ref{fact1}.

Hence, we may assume in what follows that the vertices $w^*$ and $w'^*$ have a common neighbor, say $p$.

By Claim~\ref{c:pre-24-6cycle2}, the outgoing edge of $p$ from the $6$-cycle $C' \colon w (ww')_1 w' w'^* p w^* w$ is a \rede, that is, the vertex $p$ is incident with a \rede, say $pq$. By our earlier observations, there are three cases to consider, namely when the edge $e$ is a \gree\ or a \blae\ with the vertex $z$ the center of a black star or a \blae\ with the vertex $z$ incident to an additional \blae\ and to a \rede, say $zz'$.

Suppose that the edge $e$ is a \gree. In this case, let $H$ be obtained from $G$ by removing all the vertices of the cycle $C$, and removing the vertices $x$, $(xx')_1$, $w^*$, $(yz)_1$, $p$, $(pq)_1$ and $(pq)_2$. We mark the neighbor of $v$ that does not belong to the cycle $C$, and we mark the neighbor of $x$ different from $w'$ and $(xx')_1$, and we cut all other eight outgoing edges. In this case, $\w(G) \ge \w(H) + 11 \times 4 + 6 \times 5 - 7 \times 1 - 1 \times 3 = \w(H) + 64$. Every MD-set of $H$ can be extended to a MD-set of $G$ by adding to it the vertices $v$, $w$, $x$, $y$ and $(pq)_1$. Thus, $\mdom(G) \le \mdom(H) + \alpha$ and $\w(G) \ge \w(H) + \beta$, where $\alpha = 5$ and $\beta \ge 64$, contradicting Fact~\ref{fact1}. Hence, the edge $e$ is not a \gree. Therefore, the edge is a \blae.

Suppose that the edge $e$ is a \blae\ with the vertex $z$ the center of a black star. In this case, let $H$ be obtained from $G$ by removing all the vertices of the cycle $C$, and removing the vertices $x$, $(xx')_1$, $w^*$, $z$, $p$, $(pq)_1$ and $(pq)_2$. We mark the neighbor of $v$ that does not belong to the cycle $C$, and we mark the neighbor of $x$ different from $w'$ and $(xx')_1$, and we cut all other nine outgoing edges. In this case, $\w(G) \ge \w(H) + 12 \times 4 + 5 \times 5 - 8 \times 1 - 1 \times 3 = \w(H) + 61$. Every MD-set of $H$ can be extended to a MD-set of $G$ by adding to it the vertices $v$, $w$, $x$, $y$ and $(pq)_1$. Thus, $\mdom(G) \le \mdom(H) + \alpha$ and $\w(G) \ge \w(H) + \beta$, where $\alpha = 5$ and $\beta \ge 61$, contradicting Fact~\ref{fact1}.

The final case to consider is when the edge $e$ is a \blae\ with the vertex $z$ incident to an additional \blae\ and to a \rede, say $zz'$. In this case, let $H$ be obtained from $G$ by removing all the vertices of the cycle $C$, and removing the vertices $x$, $(xx')_1$, $w^*$, $z$, $z'$, $(zz')_1$, $(zz')_2$, $p$, $(pq)_1$ and $(pq)_2$. We mark the neighbor of $v$ that does not belong to the cycle $C$, and we mark the neighbor of $x$ different from $w'$ and $(xx')_1$, and we cut all other ten outgoing edges. In this case, $\w(G) \ge \w(H) + 13 \times 4 + 7 \times 5 - 9 \times 1 - 1 \times 3 = \w(H) + 75$. Every MD-set of $H$ can be extended to a MD-set of $G$ by adding to it the vertices $v$, $w$, $x$, $y$, $(zz')_2$ and $(pq)_1$. Thus, $\mdom(G) \le \mdom(H) + \alpha$ and $\w(G) \ge \w(H) + \beta$, where $\alpha = 6$ and $\beta \ge 75$, contradicting Fact~\ref{fact1}.~\smallqed

\begin{claim}\label{c:claim24}
In the colored multigraph $M_G$, there is no \green-\black-\green-\black-\green-\black-\green\ path $P \colon uu'vv'ww'xx'$ where $u'$, $w$ and $w'$ are adjacent to the centers $u'^*$, $w^*$ and $w'^*$ of black stars.
\end{claim}

\begin{figure}[htb]
\begin{center}
\input{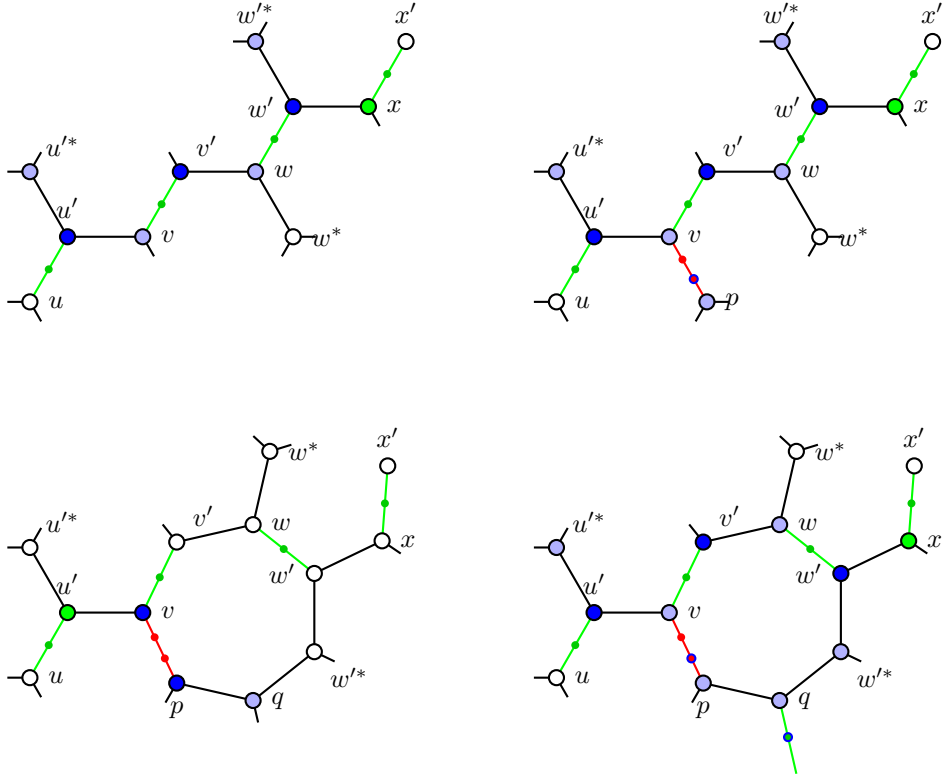}
\caption{Colored multigraph in the proof of Claim~\ref{c:claim24}}
\label{f:claim24}
\end{center}
\end{figure}

\proof
Suppose, to the contrary, that $M_G$ contains such a path $P$. We first obtain structural properties of the vertices on the path $P$.

Since there is no $4$-cycle, the vertex $u$ is not adjacent to $u'^*$ and $v$. Moreover the absence of $4$-cycles implies that the vertices $u'^*$ and $v$ have no common neighbor. Since there is no $5$-cycle, the vertex $w^*$ is not adjacent to $v$ and $w'^*$. Since there is no $7$-cycle, we note that $u \ne w^*$ (this also follows from the assumption that $w^*$ is the center of a black star, while $u$ is incident with a \gree) and the vertex $u'^*$ is not adjacent to $w^*$. Since there is no $8$-cycle, we note that $u'^* \ne w'^*$. Moreover the absence of $8$-cycles implies that $v$ and $w'^*$ have no common neighbor. If $u$ and $w'^*$ are adjacent, then they belong to a common $10$-cycle, contradicting  Claim~\ref{c:pre-24-10cycle}.
If $u'^*$ and $w'^*$ have a common neighbor, then they belong to a common $10$-cycle, contradicting  Claim~\ref{c:pre-24-other10cycle}.

Suppose that the edge incident with $v$ that does not belong to the path $P$ is a \blae. In this case, let $H$ be obtained from $G$ by removing the vertices $u'$, $v$, $v'$, $w$, $w'$, $u'^*$, $w'^*$, $(uu')_1$, $(vv')_1$ and $(ww')_1$. We mark $x$ the neighbor of $v'$ that does not belong to the path $P$, and we cut all the seven other outgoing edges. This yields $\w(G) \ge \w(H) + 7\times 4 + 3 \times 5 - 7 \times 1 = \w(H) + 36$. Every MD-set of $H$ can be extended to a MD-set of $G$ by adding to it the vertices $u'$, $v'$ and $w'$. (This is illustrated in the top left hand figure of Figure~\ref{f:claim24} where the vertices in the  MD-set are dark blue and the other vertices of degree~$3$ in $G$ that were deleted when constructing $H$ are colored light blue.) Thus, $\mdom(G) \le \mdom(H) + \alpha$ and $\w(G) \ge \w(H) + \beta$, where $\alpha = 3$ and $\beta = 36$, contradicting Fact~\ref{fact1}.

Hence the edge incident with $v$ that does not belong to the path $P$ is a \rede, say $vp$. By Claim~\ref{c:no-green-red-green}, the vertex $p$ is not incident with a \gree, and therefore $p$ is incident with two $\blaes$. Since the vertex $p$ is incident with a \rede, we note that $p$ is distinct from the vertices $u$, $u'^*$, $w^*$, and $w'^*$. Since there is no $7$-cycle, the vertices $p$ and $u'^*$ have no common neighbor and $p$ is not adjacent to $u$. Since there is no 8-cycle, $p$ is not adjacent to $w^*$ either.

Suppose that the vertices $p$ and $w'^*$ have no common neighbor. In this case, let $H$ be obtained from $G$ by removing the vertices $u'$, $v$, $v'$, $w$, $w'$, $u'^*$, $w'^*$, $p$, $(uu')_1$, $(vv')_1$, $(ww')_1$, $(vp)_1$ and $(vp)_2$. We mark $x$ and the neighbor of $v'$ that does not belong to the path $P$, and we cut all the eight other outgoing edges. This yields $\w(G) \ge \w(H) + 8 \times 4 + 5 \times 5 - 8 \times 1 = \w(H) + 49$. Every MD-set of $H$ can be extended to a MD-set of $G$ by adding to it the vertices $u'$, $v'$, $w'$ and $(vp)_2$. (This is illustrated in the top right hand figure of Figure~\ref{f:claim24} where the vertices in the  MD-set are dark blue and the other vertices of degree~$3$ in $G$ that were deleted when constructing $H$ are colored light blue.) Thus, $\mdom(G) \le \mdom(H) + \alpha$ and $\w(G) \ge \w(H) + \beta$, where $\alpha = 4$ and $\beta = 49$, contradicting Fact~\ref{fact1}.

Hence, the vertices $p$ and $w'^*$ have a common neighbor which we call~$q$. By our earlier observations, the edges $pq$ and $qw'^*$ are both $\blaes$. This produces a $7$-cycle in the colored multigraph $M_G$ that contains both $p$ and $w'^*$, namely the \black-\black-\red-\green-\black-\green-\black\ cycle $C_M \colon w'^* q p v v' w w' w'^*$. In the original graph $G$, this is an $11$-cycle, namely the cycle $C \colon w'^* q p (vp)_2 (vp)_1 v (vv')_1 v' w (ww')_1 w'^*$.

Suppose that the third edge in $M_G$ incident with $q$ (that does not belong to the cycle $C_M$) is a $\blae$. In this case, let $H$ be obtained from $G$ by removing the vertices $v$, $p$, $q$, $(vv')_1$, $(vp)_1$ and $(vp)_2$. We mark the neighbor $u'$ of $v$ that was not deleted, and we mark the neighbor of $p$ different from $q$ and $(vp)_2$, and we cut all the three other outgoing edges. This yields $\w(G) \ge \w(H) + 3 \times 4 + 3 \times 5 - 3 \times 1 = \w(H) + 24$. Every MD-set of $H$ can be extended to a MD-set of $G$ by adding to it the vertices $v$ and $p$. Thus, $\mdom(G) \le \mdom(H) + \alpha$ and $\w(G) \ge \w(H) + \beta$, where $\alpha = 2$ and $\beta = 24$, contradicting Fact~\ref{fact1}.

Hence the edge, $qr$ say, in $M_G$ incident with $q$ that does not belong to the cycle $C_M$ is either a \gree\ or a \rede. In this case, let $H$ be obtained from $G$ by removing all eleven vertices on the cycle $C$, removing the vertices $u'$, $u'^*$, $(uu')_1$ and $(qr)_1$. We mark $x$ and the neighbor of $v'$ that does not belong to the cycle $C$ and we mark the neighbor of $(qr)_1$ different from $q$. Thus, if $qr$ is a \gree, then we mark the vertex~$r$, while if $qr$ is a \rede, then we mark the vertex~$(qr)_2$. Further, we cut all the six other outgoing edges. This yields $\w(G) \ge \w(H) + 9 \times 4 + 6 \times 5 - 6 \times 1 = \w(H) + 60$. Every MD-set of $H$ can be extended to a MD-set of $G$ by adding to it the vertices $u'$, $v'$, $w'$, $(qr)_1$ and $(vp)_2$. (This is illustrated in the bottom figure of Figure~\ref{f:claim24} where the vertices in the  MD-set are dark blue and the other vertices of degree~$3$ in $G$ that were deleted when constructing $H$ are colored light blue.) Thus, $\mdom(G) \le \mdom(H) + \alpha$ and $\w(G) \ge \w(H) + \beta$, where $\alpha = 5$ and $\beta = 60$, contradicting Fact~\ref{fact1}.~\smallqed

\subsection{Maximal green-black alternating paths}
\label{S:max-green-black-paths}

In this section, we consider maximal alternating \green-\black\ paths in $M_G$ that are not cycles.

\subsubsection{Structural properties of extremities of green-black alternating paths}
\label{S:types}

In this subsection, we establish structural properties of the extremities of maximal alternating \green-\black\ paths that are not \green-\black\ cycles, and show that these structures can be classified into three types, which we will refer to as Type-(1), Type-(2), and Type-(3). In order to present these three different structural types, we prove three claims.

\begin{claim}\label{forbid-structure-I}
The colored multigraph $M_G$ does not contain a star with two \black\ edges and one \gree, and with two leaves of the star incident with three \black\ edges.
\end{claim}
\proof
Suppose, to the contrary, that $M_G$ contains a star with vertex set $\{v,v_1,v_2,v_3\}$, where $v$ is the vertex of degree~$3$ and where the edge $vv_1$ and $vv_2$ are \black\ edges and the edge $vv_3$ is a \gree, and where both $v_1$ and $v_2$ are incident with three \black\ edges. Let $H$ be obtained from $G$ by removing the vertex $v$ and its three neighbors in $G$. Every \mds\ of $H$ can be extended to a \mds\ of $G$ by adding to it the vertex $v$. We note that $\w(G) \ge \w(H) + 1 \times 5 + 3 \times 4 - 5 \times 1 = \w(H) + 12$. Thus, $\mdom(G) \le \mdom(H) + \alpha$ and $\w(G) \ge \w(H) + \beta$, where $\alpha = 1$ and $\beta = 12$, contradicting Fact~\ref{fact1}.~\smallqed

\begin{claim}\label{forbid-structure-II}
The colored multigraph $M_G$ does not contain a star with one \red, one \black\ edge, and one \gree, and with one leaf of the star incident with three \black\ edges.
\end{claim}
\proof
Suppose, to the contrary, that $M_G$ contains a star with vertex set $\{v,v_1,v_2,v_3\}$, where $v$ is the vertex of degree~$3$ and where the edge $vv_1$ is \black, the edge $vv_2$ is red and the edge $vv_3$ is \green, and where $v_1$ is incident with three \black\ edges. Let $H$ be obtained from $G$ by removing the vertex $v$ and its three neighbors in $G$. Every \mds\ of $H$ can be extended to a \mds\ of $G$ by adding to it the vertex $v$. We note that $\w(G) \ge \w(H) + 2 \times 5 + 2 \times 4 - 1 \times 3 - 3 \times 1 = \w(H) + 12$. Thus, $\mdom(G) \le \mdom(H) + \alpha$ and $\w(G) \ge \w(H) + \beta$, where $\alpha = 1$ and $\beta = 12$, contradicting Fact~\ref{fact1}.~\smallqed


\begin{claim}\label{allowable-structure}
If the colored multigraph $M_G$ contains a star centered at a vertex $u$ with two \black\ edges, $uv$ and $uw$, and with one \gree, $ux$, and such that neither $v$ nor $w$ is incident with a \gree, then the following properties hold. \\ [-26pt]
\begin{enumerate}
\item  If both $v$ and $w$ are incident with a \red, then the end of the \red\ incident with $v$ (respectively, $w$) in $M_G$ that is different from $v$ (respectively, $w$) is incident with a \gree.
\item If the vertex $v$ is incident with one \red\, $e_v$, and the vertex $w$ is incident with no \red, then the end of the edge $e_v$ different from $v$ is incident with a \gree.
\end{enumerate}
\end{claim}
\proof
We first prove part~(a). Suppose that there is a \red\ $vy$ incident with $v$ in $M_G$ and a \red\ $wz$ incident with $w$ in $M_G$, where $y$ is incident with two \black\ edges.
By supposition, neither $v$ nor $w$ in incident with a \gree, implying that the third edge incident with $v$ in $M_G$ (different from $uv$ and $vy$) and the third edge incident with $w$ in $M_G$ (different from $uw$ and $wz$) are \black\ edges. This structure is illustrated by the graph in Figure~\ref{f:c31}(a).
By Claim~\ref{no-long-red-black-black-star}, both \redes\ $vy$ and $wz$ are short. Since $G$ contains no $7$-cycle, the vertices $y$ and $w$ have no common neighbor. Let $H$ be obtained from $G$ by removing the vertices $u$, $v$, $w$, $x$ and $y$, and removing all vertices on the $(vy)$ and $(wz)$ $2$-paths in $G$, and marking the vertex $z$. We note that $\w(G) \ge \w(H) + 5 \times 5 + 4 \times 4 - 5 \times 1 = \w(H) + 36$. Every \mds\ of $H$ can be extended to a \mds\ of $G$ by adding to it the three vertices $u$, $(vy)_2$ and $(wz)_2$. Thus, $\mdom(G) \le \mdom(H) + \alpha$ and $\w(G) \ge \w(H) + \beta$, where $\alpha = 3$ and $\beta = 36$, contradicting Fact~\ref{fact1}. This completes the proof of part~(a).

\begin{figure}[htb]
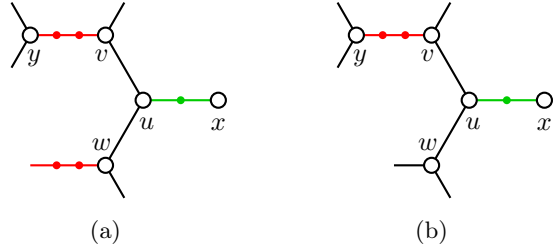

\begin{center}
    \input{green-black-end1n.tex} \hspace{1cm}
    \input{green-black-end2n.tex}
    \caption{Forbidden structures in the proof of  Lemma~\ref{allowable-structure}}\label{f:c31}
\end{center}
\end{figure}

We next prove part~(b). Suppose, to the contrary, that the end $y$ of the \red\ $e_v = yv$ different from $v$ is incident with two \blae. This structure is illustrated by the graph in Figure~\ref{f:c31}(b). Let $H$ be obtained from $G$ by removing the vertices $u$, $v$, $w$, $x$ and $y$, and removing both vertices on the $(vy)$ $2$-path in $G$. We note that $\w(G) \ge \w(H) + 3 \times 5 + 4 \times 4 - 6 \times 1 = \w(H) + 25$. Every \mds\ of $H$ can be extended to a \mds\ of $G$ by adding to it the two vertices $u$ and $(vy)_2$. Thus, $\mdom(G) \le \mdom(H) + \alpha$ and $\w(G) \ge \w(H) + \beta$, where $\alpha = 2$ and $\beta = 25$, contradicting Fact~\ref{fact1}.~\smallqed

\medskip
We are now in a position to present three structural properties of the extremities of maximal alternating \green-\black\ paths. These properties result in a classification of the extremities into three types, which we will refer to as Type-(1), Type-(2), and Type-(3).

Consider a maximal \green-\black\ path $P$ in $M_G$ that starts at a vertex $u$. Let $x$ be the neighbor of $u$ on $P$. We note that the $ux$ is a \gree. Let $v$ and $w$ be the neighbors of $u$ in $M_G$ that do not belong to the path $P$. We note that both $uv$ and $uw$ are \blae, or exactly one of $uv$ and $uw$ is a \rede. By the maximality of the path $P$, neither $v$ nor $w$ is incident with a \gree. There are two cases for us to consider.

Suppose firstly that both $uv$ and $uw$ are \blae. By Claim~\ref{forbid-structure-I}, at least one of $v$ and $w$ is incident with a \rede. Renaming $v$ and $w$ if necessary, we may assume that $v$ is incident with a \rede, say $vy$. We note that the edge incident with $v$, different from $uv$ and $vy$, is a \blae.

Suppose that $w$ is incident with a \rede, say $wz$. We note that the edge incident with $w$, different from $uw$ and $wz$, is a \blae. By Claim~\ref{allowable-structure}(b), the end $y$ (respectively, $z$) of the \red\ incident with $v$ (respectively, $w$) in $M_G$ that is different from $v$ (respectively, $w$) is incident with a \gree. We call the subgraph of $M_G$ induced by the vertices $\{u,v,w,x,y,z\}$ and their neighbors a Type-(1) structure of $M_G$ from the extremity~$u$ of a maximal alternating \green-\black\ path. This structure is illustrated in Figure~\ref{f:structures-types}(a).

Suppose that $w$ is incident only with \blaes. By Claim~\ref{allowable-structure}(b), the vertex $y$ is incident with a \gree. We call the subgraph of $M_G$ induced by the vertices $\{u,v,w,x,y\}$ and their neighbors a Type-(2) structure of $M_G$ from the extremity~$u$ of a maximal alternating \green-\black\ path. This structure is illustrated in Figure~\ref{f:structures-types}(b).

Suppose next that one of $uv$ and $uw$ is a \rede. Renaming vertices if necessary, we may assume that $uw$ is a \rede. Thus, the edge $uv$ is a \blae. By the maximality of the path $P$, the vertex $v$ is not incident with a \gree. By Claim~\ref{forbid-structure-I}, the vertex $v$ is incident with a \rede, say $vy$. We call the subgraph of $M_G$ induced by the vertices $\{u,v,w,x,y\}$ and their neighbors a Type-(3) structure of $M_G$ from the extremity~$u$ of a maximal alternating \green-\black\ path. This structure is illustrated in Figure~\ref{f:structures-types}(c).

\begin{figure}[htb]
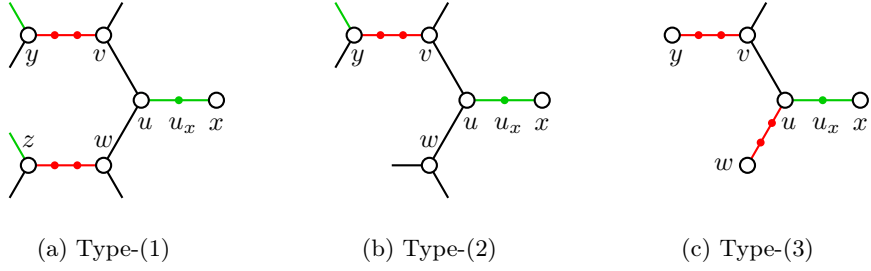

\begin{center}
    \input{green-black-end4n.tex} \hspace{1cm}
    \input{green-black-end3n.tex} \hspace{1cm}
    \input{green-black-end5n.tex}
    \caption{Structures for the extremity $u$ of a maximal alternating \green-\black\ path}
    \label{f:structures-types}
\end{center}
\end{figure}


Using our structural classification of the extremities of maximal alternating \green-\black\ paths, we now give two counting strategies that we will apply to each of these three structures. Let
\[
P \colon u_1,u_1',u_2,u_2',\ldots,u_k,u_k'
\]
be a maximal alternating \green-\black\ path in $M_G$, and so the edges \(u_iu_i'\) are \grees\ for $i \in [k]$ and edges \(u_i'u_{i+1}\) are \blaes\ for $i \in [k-1]$. We now define a strategy for handling the extremities $u_1$ and $u_k'$ of such a maximal alternating \green-\black\ path $P$. Let the \gree\ \(u_iu_i'\) correspond to the path $u_i, v_i, u_i'$ in the graph $G$, and so $v_i$ has degree~$2$ in $G$ with $u_i$ and $u_i'$ as its ends for $i \in [k]$.

In Strategy~1, we analyse the structure of the extremity of $u$ when we delete the vertex $u$ (and selected other vertices) but we do not delete the vertex $v_1$. In the resulting graph we do not count the weight change of the vertex $v_1$, since the idea is to subsequently dominate the vertex $v_1$ by carefully selecting vertices $u'_i$ for $i\in[k]$ on the maximal alternating \green-\black\ path $P$.

In Strategy~2, we analyse the structure of the extremity of $u$ when we delete both vertices $u$ and $v_1$ (and selected other vertices) but we do not delete the vertex $u_1'$. In the resulting graph we do not count the weight change of the vertex $u_1'$, since the idea is to subsequently dominate the vertex $u_1'$ by carefully selecting vertices $u_i$ for $i \in [k]$ on the maximal alternating \green-\black\ path $P$.
More precisely, we explain next these two strategies applied to each of the three structural types in turn.

\medskip
\emph{Strategies~1 and~2 applied to a Type-{\rm (1)} structure.} Suppose that the extremity $u$ belongs to a structure of Type-(1), as illustrated in Figure~\ref{f:structures-types}(a). Let the vertices in such a structure be named as in Section~\ref{S:types}, where $x = u_1'$ and $u_x = v_1$. Let $y'$ be the neighbor of $y$ on the \gree\ incident with $y$.

In Strategy~1, let $H$ be obtained from $G$ by removing the vertices $u$, $v$, $w$, and $y$, and removing all vertices on the $(vy)$, $(yy')$, and $(wz)$ $2$-paths in $G$, and marking the neighbors of $v$ and $y$ that are not deleted. Note that we remove four vertices of degree~$3$ and five vertices of degree~$2$ when constructing $H$. Further, five vertices in $H$ decrease their degrees from~$3$ to~$2$, and two of these vertices are marked in $H$. In Strategy~1, we do not count the cost of the vertex $v_1$, whose degree has decreased from~$2$ to~$1$ when constructing $H$. With this assumption, where we exclude the weight change of the vertex $v_1$, we have $\w(G) = \w(H) + 4 \times 4 + 5 \times 5 - 3 \times 1 = (\w(H) + 36) + 2$. Every \mds\ of $H$ can be extended to a \mds\ of $G$ by adding to it the three vertices $v$, $y$, and $(wz)_1$. Thus, $\mdom(G) \le \mdom(H) + \alpha$ and $\w(G) \ge \w(H) + \beta + 2$, where $\alpha = 3$ and $\beta = 36$. This results in a surplus weight of~$2$, which we denote by $+2$.

In Strategy~2, let $H$ be obtained from $G$ by removing the vertices $u$, $v$, $v_1$, and $w$, and removing all vertices on the $(vy)$ and $(wz)$ $2$-paths in $G$, and marking the vertices $y$ and $z$. Note that we remove three vertices of degree~$3$ and five vertices of degree~$2$ when constructing $H$. Further, four vertices in $H$  (different from~$u_1'$) decrease their degrees from~$3$ to~$2$, and two of these vertices are marked in $H$. In Strategy~2, we do not count the weight change of the vertex $u_1'$, whose degree has decreased from~$3$ to~$2$ when constructing $H$. With this assumption, where we exclude the weight change of the vertex $u_1'$, we have $\w(G) = \w(H) + 3 \times 4 + 5 \times 5 - 2 \times 1 = (\w(H) + 36) - 1$. Every \mds\ of $H$ can be extended to a \mds\ of $G$ by adding to it the three vertices $u$, $(vy)_2$, and $(wz)_2$. Thus, $\mdom(G) \le \mdom(H) + \alpha$ and $\w(G) \ge \w(H) + \beta - 1$, where $\alpha = 3$ and $\beta = 36$. This results in a weight shortage of~$1$, which we denote by $-1$.

\medskip
\emph{Strategies~1 and~2 applied to a Type-{\rm (2)} structure.} Suppose that the extremity $u$ belongs to a structure of Type-(2), as illustrated in Figure~\ref{f:structures-types}(b). Let the vertices in such a structure be named as in Section~\ref{S:types}, where $x = u_1'$ and $u_x = v_1$. Let $y'$ be the neighbor of $y$ on the \gree\ incident with $y$.

In Strategy~1, let $H$ be obtained from $G$ by removing the vertices $u$, $v$, and $y$, and removing all vertices on the $(vy)$ and $(yy')$ $2$-paths in $G$, and marking the neighbors of $v$ and $y$ that are not deleted. Note that we remove three vertices of degree~$3$ and two vertices of degree~$2$ when constructing $H$. Further, four vertices in $H$ decrease their degrees from~$3$ to~$2$, and two of these vertices are marked in $H$. As before, in Strategy~1 we do not count the cost of the vertex $v_1$, whose degree has decreased from~$2$ to~$1$ when constructing $H$. With this assumption, we have $\w(G) = \w(H) + 3 \times 4 + 3 \times 5 - 2 \times 1 = (\w(H) + 24) + 1$. Every \mds\ of $H$ can be extended to a \mds\ of $G$ by adding to it the two vertices $v$ and $y$. Thus, $\mdom(G) \le \mdom(H) + \alpha$ and $\w(G) \ge \w(H) + \beta + 1$, where $\alpha = 2$ and $\beta = 24$. This results in a surplus weight of~$1$, which we denote by $+1$.

In Strategy~2, let $H$ be obtained from $G$ by removing the vertices $u$, $v_1$, $v$, and $w$, and removing both vertices on the $(vy)$ $2$-path in $G$, and marking the vertex $y$. Note that we remove three vertices of degree~$3$ and three vertices of degree~$2$ when constructing $H$. Further, four vertices in $H$ (different from~$u_1'$) decrease their degrees from~$3$ to~$2$, and one of these vertices is marked in $H$. In Strategy~2, we do not count the weight change of the vertex $u_1'$, whose degree has decreased from~$3$ to~$2$ when constructing $H$. With this assumption, we have $\w(G) = \w(H) + 3 \times 4 + 3 \times 5 - 3 \times 1 = \w(H) + 24$. Every \mds\ of $H$ can be extended to a \mds\ of $G$ by adding to it the two vertices $u$ and $(vy)_2$. Thus, $\mdom(G) \le \mdom(H) + \alpha$ and $\w(G) \ge \w(H) + \beta$, where $\alpha = 2$ and $\beta = 24$. This results in neither a surplus weight nor a weight shortage. In this case, we say that the weight difference between $G$ and $H$ is \emph{balanced}, denoted by $0$.

\emph{Strategies~1 and~2 applied to a Type-{\rm (3)} structure.} Suppose that the extremity $u$ belongs to a structure of Type-(3), as illustrated in Figure~\ref{f:structures-types}(c). Let the vertices in such a structure be named as in Section~\ref{S:types}, where $x = u_1'$ and $u_x = v_1$. Let $y'$ be the neighbor of $y$ on the \gree\ incident with $y$.

In Strategy~1, let $H$ be obtained from $G$ by removing the vertices $u$ and $v$, and removing all vertices on the $(vy)$ and $(uw)$ $2$-paths in $G$, and marking the neighbors of $v$ and $y$ that are not deleted. Note that we remove two vertices of degree~$3$ and four vertices of degree~$2$ when constructing $H$. Further, three vertices in $H$ decrease their degrees from~$3$ to~$2$. As before, in Strategy~1 we do not count the weight change of the vertex $v_1$, whose degree has decreased from~$2$ to~$1$ when constructing $H$. With this assumption, we have $\w(G) = \w(H) + 2 \times 4 + 4 \times 5 - 3 \times 1 = (\w(H) + 24) + 1$. Every \mds\ of $H$ can be extended to a \mds\ of $G$ by adding to it the two vertices $(uw)_1$ and $(vy)_2$. Thus, $\mdom(G) \le \mdom(H) + \alpha$ and $\w(G) \ge \w(H) + \beta + 1$, where $\alpha = 2$ and $\beta = 24$. This results in a surplus weight of~$1$, which we denote by $+1$.

In Strategy~2, let $H$ be obtained from $G$ by removing the vertices $u$, $v_1$, $(uw)_1$, and $v$, and removing all vertices on the $(vy)$ $2$-path in $G$, and marking the vertex $y$. Note that we remove two vertices of degree~$3$ and four vertices of degree~$2$ when constructing $H$. Further, two vertices in $H$ (different from~$u_1'$) decrease their degrees from~$3$ to~$2$, and one of these vertices is marked in $H$, while one vertex decreases its degree from~$2$ to~$1$. In Strategy~2, we do not count the weight change of the vertex $u_1'$. With this assumption, we have $\w(G) = \w(H) + 2 \times 4 + 4 \times 5 - 1 - 3 = \w(H) + 24$. Every \mds\ of $H$ can be extended to a \mds\ of $G$ by adding to it the two vertices $u$ and $(vy)_2$. Thus, $\mdom(G) \le \mdom(H) + \alpha$ and $\w(G) \ge \w(H) + \beta$, where $\alpha = 2$ and $\beta = 24$. This results in neither a surplus weight nor a weight shortage. In this case, we say that the weight difference between $G$ and $H$ is \emph{balanced}, denoted by $0$.
The surplus or deficit weight resulting from applying our two strategies to the three structural types of extremities are summarized in Table~\ref{table:strategies}.

\begin{table}[ht!]
\begin{center}
\begin{tabular}{|c|c|c|c|}
\hline
& Type-(1) & Type-(2) & Type-(3) \\
	\hline
	Strategy~1 & $+2$ & $+1$ & $+1$ \\
	\hline
	Strategy~2 & $-1$ & $0$ & $0$ \\
	\hline
\end{tabular}
\end{center}
\vskip -0.4 cm
\caption{The two strategies and their corresponding surplus or deficit weights}\label{table:strategies}
\end{table}

\subsubsection{Structural properties of internal vertices of black-green alternating paths}
\label{S:types-internal}

In the previous section, Section~\ref{S:types}, we studied the structure of extremities of maximal alternating \green-\black\ path in $M_G$ that start and end with \grees, and gave two counting strategies that we will apply to each of these three structures. In this section, we study the structure of internal vertices of a maximal alternating \green-\black\ path. We adopt our notation in Section~\ref{S:types}. In particular,
\[
P \colon u_1,u_1',u_2,u_2',\ldots,u_k,u_k'
\]
is a maximal alternating \green-\black\ path in $M_G$  that starts and ends with \grees. We say that a vertex $v$ is a \emph{neighbor} of a path $P$ if it does not belong to the path but is a neighbor of some internal vertex on the path $P$ (different from one the extremities of $P$). Further, if $v$ has two neighbors or more neighbors on $P$, we say that $v$ is a \emph{close neighbor} of the path $P$. We call an edge $e$ a \emph{neighbor} of the path $P$ if it does not belong to the path, but is incident with some internal vertex on the path $P$.

Recall that in Section~\ref{S:types} we applied two strategies, namely Strategy~$1$ and Strategy~$2$, which yielded counting arguments to handle the extremities of the path $P$ depending on the type of structure, namely Type-(1), -(2) or -(3) structure of the extremity and on structure of the graph $G$ which we analyse in this section. Our initial aim is to apply Strategy~1 to one extremity of the path $P$ and Strategy~2 to the other extremity of the path $P$, and to remove from $G$ the vertices associated with the path $P$.

If we apply Strategy~1 to the extremity $u_1$ of the path $P$, then we wish to extend a MD-set in the resulting graph to a MD-set in the graph $G$ by adding to it, among possibly other vertices, the $k-1$ vertices $u_i'$ of the path $P$ for all $i \in [k-1]$. In this case, the vertices $u_j$ for all $j \in [k] \setminus \{1\}$ and the vertices of degree~$2$ associated with the internal vertices of \gree\ in the path $P$, are not added to the dominating set.

If we apply Strategy~2 to the extremity $u_k'$ of the path $P$, then we wish to extend a MD-set in the resulting graph to a MD-set in the graph $G$ by adding to it, among possibly other vertices, the $k-1$ vertices $u_i$ of the path $P$ for all $i \in [k] \setminus \{1\}$. In this case, the vertices $u_j'$ for all $j \in [k-1]$ and the vertices of degree~$2$ associated with the internal vertices of \gree\ in the path $P$, are not added to the dominating set.

As mentioned earlier, we focus next on the structure of the internal vertices of the maximal alternating \green-\black\ path $P$. We call an edge that does not belong to the path $P$ but joins two vertices of the path $P$ a \emph{well}-\emph{behaved edge} with respect to the path $P$. Such edges will not trouble us in anyway, since when we remove the vertices of $G$ associated with the path $P$, all well-behaved edges  with respect to $P$ are also removed. Hence, we may assume there are no well-behaved edges with respect to $P$.

In what follows, we call a vertex $u_i'$ for some $i \in [k-1]$ an \emph{upper vertex}, and we call a vertex $u_i$ for some $i \in [k] \setminus \{1\}$ a \emph{lower vertex}. Every neighboring edge of $P$ that is incident with an upper vertex we call an \emph{upper edge}, while every neighboring edge of $P$ that is incident with a lower vertex we call a \emph{lower edge}. By our assumption that there are no well-behaved edges  with respect to $P$, every neighboring edge of $P$ is either upper or lower (but not both).

We consider the contribution of upper and lower edges when we remove all vertices from the path $P$. For this purpose, we define a set $S_P^{\upper}$ (resp., $S_P^{\lowr}$) that associates with an upper (resp., lower) neighboring edge of $P$ a set of vertices that will be removed when we remove the vertices of the path $P$. We also associate a function $f_P^{\upper}$ (resp., $f_P^{\lowr}$) that determines the additional cost incurred by the removal of such vertices in $S_P^{\upper}$ or $S_P^{\lowr}$, where this function does not count the weight of the vertices of the path $P$ (which will be counted separately when the remove the vertices of $P$). We consider several types of neighboring edges of the path $P$, and analyse each edge in turn. By Claim~\ref{c:no-green-red-green}, we note that no neighboring \rede\ of the path $P$ has both its ends belonging to the path $P$.

We consider first several types of neighboring edges of the path $P$, and analyse each edge in turn.

\paragraph{Red islands.} We call a neighboring edge $e = uv$ of the path $P$ a \redi\ if $e$ is a \rede\ such that one end, say $u$, of the edge $e$ belongs to the path $P$, and the other end $v$ of the edge $e$ is adjacent to two additional vertices on the path $P$. Let $e'$ and $e''$ denote the  two edges incident with $v$ that are different from the \rede\ $e$. By supposition, the ends of $e'$ and $e''$ different from the vertex $v$ belongs to the path $P$. Necessarily, both edges $e'$ and $e''$ are \blaes\ by Claims~\ref{c:no-green-red-green} and~\ref{c:red-matching}. An example of a \redi\ is illustrated in Figure~\ref{f:red-island}. We call the  three edges $e, e'$ and $e''$ the edges associated with a \redi.

\begin{figure}[htb]
\begin{center}
\input{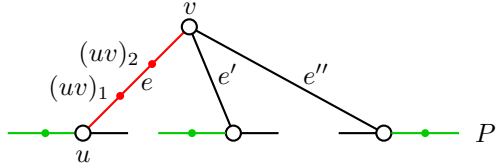}
\vskip -0.75 cm
\caption{A red island}\label{f:red-island}
\end{center}
\end{figure}

For each edge $e_v \in \{e,e',e''\}$ incident with $v$, we define both $S_P^{\upper}(e_v)$ and $S_P^{\lowr}(e_v)$ to consist of the vertex $v$ and both internal vertices of the path in $G$ corresponding to the \rede\ $e$ in $M_G$, that is
\[
S_P^{\upper}(e_v) = S_P^{\lowr}(v) = \{v, (uv)_1, (uv)_2\}.
\]

Let $H$ be obtained from $G$ by removing the vertices $v$, $(uv)_1$, and $(uv)_2$ that belong to the set $S_P^{\upper}(e_v)$ and $S_P^{\lowr}(e_v)$. Note that we remove one vertex of degree~$3$ and two vertices of degree~$2$. Not counting the weight change of the three vertices on the path $P$ that are adjacent with $v$ in $M_G$, we have $\w(G) = \w(H) + 1 \times 4 + 2 \times 5 = \w(H) + 14$. Every \mds\ of $H$ can be extended to a \mds\ of $G$ by adding to it the vertex $(uv)_2$. Thus, $\mdom(G) \le \mdom(H) + \alpha$ and $\w(G) \ge \w(H) + \beta + 2$, where $\alpha = 1$ and $\beta = 12$. This results in an excess weight of~$2$, which we share equally among the three edges $e$, $e'$ and $e''$, that is, each edge is assigned a weight of~$2/3$. Thus for each edge $e_v$ incident with $v$ in the multigraph $M_G$, we define
\[
f_P^{\upper}(e_v) = f_P^{\lowr}(e_v) = \frac{2}{3}.
\]

\paragraph{Green islands.} We call a \gree\ $e = uv$ with both ends adjacent to two vertices on the path $P$ a \grei. Let $e_1$, $e_2$, $e_3$ and $e_4$ be the four edges joining the \grei\ to the path $P$. By Claims~\ref{c:green-matching} and~\ref{c:no-green-red-green}, these four edges are all \blaes. An example of a \grei\ is illustrated in Figure~\ref{f:green-island}. We call the four edges $e_1$, $e_2$, $e_3$ and $e_4$ the edges associated with a \grei.

\begin{figure}[htb]
\begin{center}
\input{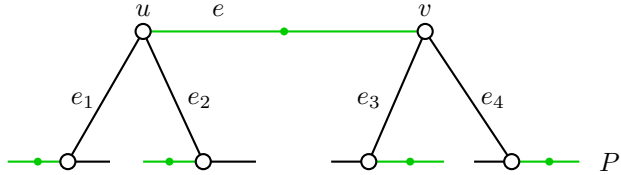}
\vskip -0.75 cm
\caption{A green island}\label{f:green-island}
\end{center}
\end{figure}

For each edge $e_v \in \{e_1,e_2,e_3,e_4\}$ adjacent with the \gree\ $e$, we define both $S_P^{\upper}(e_v)$ and $S_P^{\lowr}(e_v)$ to consist of the vertices $u$, $v$, and the internal vertex of the path in $G$ corresponding to the \gree\ $e$ in $M_G$, that is
\[
S_P^{\upper}(e_v) = S_P^{\lowr}(v) = \{u,v, (uv)_1\}.
\]

Let $H$ be obtained from $G$ by removing the vertices $u$, $v$, and $(uv)_1$ that belong to the set $S_P^{\upper}(e_v)$ and $S_P^{\lowr}(e_v)$. Note that we remove two vertices of degree~$3$ and one vertex of degree~$2$. Not counting the weight change of the four vertices on the path $P$ that are adjacent to $u$ or $v$ in $M_G$, we have $\w(G) = \w(H) + 2 \times 4 + 1 \times 5 = \w(H) + 13$. Every \mds\ of $H$ can be extended to a \mds\ of $G$ by adding to it the vertex $(uv)_1$. Thus, $\mdom(G) \le \mdom(H) + \alpha$ and $\w(G) \ge \w(H) + \beta + 1$, where $\alpha = 1$ and $\beta = 12$. This results in an excess weight of~$1$, which we share equally among the four edges $e_1$, $e_2$, $e_3$ and $e_4$, that is, each edge is assigned a weight of~$1/4$. Thus for each edge $e_v$ incident with $v$ in the multigraph $M_G$, we define
\[
f_P^{\upper}(e_v) = f_P^{\lowr}(e_v) = \frac{1}{4}.
\]

\paragraph{Green isthmus.} We call a \gree\ $e = uv$ that is not a \grei\, a \gisthmus\, if one end, say $u$, of the edge $e$ is adjacent to two vertices on the path $P$.  Since $e$ is not a \grei, the vertex $v$ is adjacent to at most one vertex of the path $P$. An example of a \gisthmus\, is illustrated in Figure~\ref{f:green-isthmus}.

\begin{figure}[htb]
\begin{center}
\input{green-isthmus.tex}
\vskip -0.75 cm
\caption{A \gisthmus}\label{f:green-isthmus}
\end{center}
\end{figure}

Suppose that at least one of $e_1$ and $e_2$ is an upper edge, say $e_1$. Thus, $e_1 = u_i'$ for some $u \in [k-1]$. In this case, we let $S_P^{\upper}(e_1) = \{ u \}$, and we let $S_P^{\upper}(e_2) = S_P^{\lowr}(e_2) = \{ u \}$.
Further, we let $H$ be obtained from $G$ by removing the vertex $u$ of degree~$3$. We note that the degree of vertex $(uv)_1$ decreases from~$2$ to~$1$ when constructing $H$, and hence its weight increases from~$5$ to~$8$. Therefore, not counting the weight change of the vertices on the path $P$ adjacent to $u$, we have $\w(G) = \w(H) + 4 - 3 = \w(H) + 1$. Every \mds\ of $H$ is a \mds\ of $G$ (noting that the vertex $u$ is subsequently dominated when we added the vertex $u_i$ to a \mds\ of $H$). Thus, $\mdom(G) \le \mdom(H) + \alpha$ and $\w(G) \ge \w(H) + 1$, where $\alpha = \beta = 0$. This results in an excess weight of~$1$, which we share among the two edges $e_1$ and $e_2$, that is, we define
\[
f_P^{\upper}(e_1) = \frac{1}{2} \hspace*{0.5cm} \mbox{and} \hspace*{0.5cm} f_P^{\upper}(e_2) = f_P^{\lowr}(e_2) = \frac{1}{2}.
\]

Suppose that both $e_1$ and $e_2$ are lower edges. In this case, we let $S_P^{\lowr}(e_1) = S_P^{\lowr}(e_2) = \{ u, (uv)_1 \}$. Further, we let $H$ be obtained from $G$ by removing the vertices $u$ and $(uv)_1$ and marking the vertex~$v$.  Note that we remove one vertex of degree~$3$ and one vertex of degree~$2$. As before, not counting the weight change of the vertices on the path $P$ adjacent to $u$, we have $\w(G) = \w(H) + 5 + 4 = \w(H) + 9$. Every \mds\ of $H$ can be extended to a \mds\ of $G$ by adding to it the vertex $(uv)_1$. Thus, $\mdom(G) \le \mdom(H) + \alpha$ and $\w(G) \ge \w(H) + \beta - 3$, where $\alpha = 1$ and $\beta = 12$. This results in a deficit weight of~$3$, which we share among the two edges $e_1$ and $e_2$, that is, we define
\[
f_P^{\lowr}(e_1) = f_P^{\lowr}(e_2) = -\frac{3}{2}.
\]

\paragraph{Red isthmus.} We call a \rede\, $e = uv$ a \risthmus\, if neither $u$ nor $v$ belongs to the path $P$, but at least one end, say $u$, of the edge $e$ is adjacent to two vertices on the path $P$. We note that the vertex $v$ is adjacent to zero, one or two vertices of the path $P$. An example of a \risthmus\, is illustrated in Figure~\ref{f:red-isthmus}.

\begin{figure}[htb]
\begin{center}
\input{red-isthmus-1.tex}
\vskip -0.751 cm
    \caption{A \risthmus}\label{f:red-isthmus}
    \end{center}
\end{figure}

We proceed now in an analogous way as in the previous case of \gisthmus. Thus if at least one of $e_1$ and $e_2$ is an upper edge, say $e_1$, then we let $S_P^{\upper}(e_1) = \{ u \}$, and we let $S_P^{\upper}(e_2) = S_P^{\lowr}(e_2) = \{ u \}$, and we define
\[
f_P^{\upper}(e_1) = \frac{1}{2} \hspace*{0.5cm} \mbox{and} \hspace*{0.5cm} f_P^{\upper}(e_2) = f_P^{\lowr}(e_2) = \frac{1}{2}.
\]

If both $e_1$ and $e_2$ are lower edges, then we let $S_P^{\lowr}(e_1) = S_P^{\lowr}(e_2) = \{ u, (uv)_1 \}$. In this case, we let $H$ be obtained from $G$ by removing the vertices $u$ and $(uv)_1$ and marking the vertex~$(uv)_2$. Every \mds\ of $H$ can be extended to a \mds\ of $G$ by adding to it the vertex $(uv)_1$. Thus, $\mdom(G) \le \mdom(H) + \alpha$ and $\w(G) \ge \w(H) + \beta - 3$, where $\alpha = 1$ and $\beta = 12$. This results in a deficit weight of~$3$, which we share among the two edges $e_1$ and $e_2$, that is, we define
\[
f_P^{\lowr}(e_1) = f_P^{\lowr}(e_2) = -\frac{3}{2}.
\]

\paragraph{Black detour.} We call a neighboring \blae\, $e = uv$ of the path $P$ a \bdetour\, if one end, say $u$, of the edge $e$ belongs to the path $P$ and the other end, $v$, is not adjacent to a vertex of the path, except for the vertex~$u$. Further, the vertex $v$ is incident with a \gree\, or a \rede\ (or both a \gree\, or a \rede). An example of a \bdetour\, is illustrated in Figure~\ref{f:black-detour}.

\begin{figure}[htb]
\begin{center}
\input{black-detour.tex}
\vskip -0.25 cm
\caption{A \bdetour}\label{f:black-detour}
\end{center}
\end{figure}

In this case, we let $S_P^{\upper}(e) = S_P^{\lowr}(e) = \emptyset$, and we let $H$ be obtained from $G$ by removing the edge $e$. Further, if $e$ is an upper edge, we mark the vertex~$v$ in $H$. Suppose that $e$ is an upper edge. Not counting the weight change of the vertex $u$ on the path $P$, we have $\w(G) = \w(H)$ noting that in this case the vertex $u$ is added to a MD-set of $H$. Thus, $\mdom(G) \le \mdom(H) + \alpha$ and $\w(G) \ge \w(H)$, where $\alpha = \beta = 0$. Hence, we define $f_P^{\upper}(e) = 0$. Suppose next that $e$ is a lower edge. In this case, noting that the degree of vertex $v$ decreases from~$3$ to~$2$ when constructing $H$, we have $\w(G) = \w(H) - 1$. Thus, $\mdom(G) \le \mdom(H) + \alpha$ and $\w(G) \ge \w(H) - 1$, where $\alpha = \beta = 0$. This results in a deficit weight of~$-1$, and we define $f_P^{\lowr}(e) = -1$. To summarize, we have
\[
S_P^{\upper}(e) = S_P^{\lowr}(e) = \emptyset, \hspace*{0.4cm} f_P^{\upper}(e) = 0, \hspace*{0.4cm} \mbox{and} \hspace*{0.4cm} f_P^{\lowr}(e) = -1.
\]

\paragraph{Special red edge.} We call a neighboring \rede\, $e = uv$ of the path $P$ that is not a \redi\, a \rspecial\, if one end, say $u$, of the edge $e$ belongs to the path $P$. Since $e$ is not a \redi\, the vertex $v$ is adjacent to at most one vertex of the path $P$ in addition to the vertex~$u$. (Possibly, the vertex $v$ is incident with a \gree.) An example of a \rspecial\, is illustrated in Figure~\ref{f:special-black-edge}. As mentioned earlier, our aim is to remove from $G$ the vertices associated with the path $P$, and to extend a MD-set in the resulting graph to a MD-set in the original graph $G$ by adding to it all $k-1$ vertices $u_i'$ (incident with upper edges) for all $i \in [k-1]$.

\begin{figure}[htb]
\begin{center}
\input{red-detour.tex}
\vskip -0.25 cm
    \caption{A \rspecial}\label{f:special-red-edge}
    \end{center}
\end{figure}

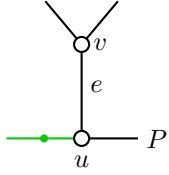
\paragraph{Special black edge.} We call a neighboring \blae\, $e = uv$ of the path $P$ a \bspecial\, if one end, say $u$, of the edge $e$ belongs to the path $P$ and the other end, $v$, is the center of a black star. An example of a \bspecial\, is illustrated in Figure~\ref{f:special-black-edge}.

\begin{figure}[htb]
\begin{center}
\input{special-black-edge.tex}
\vskip -0.25 cm
\caption{A \bspecial}\label{f:special-black-edge}
\end{center}
\end{figure}

\subsubsection{The score of a maximal alternating green-black path and its reverse}
\label{S:path-score}

Adopting our earlier notation, consider a maximal alternating \green-\black\ path in $M_G$
\[
P \colon u_1,u_1',u_2,u_2',\ldots,u_k,u_k'
\]
that starts and ends with \grees. In particular, the \gree\ \(u_iu_i'\) correspond to the path $u_i, v_i, u_i'$ in the graph $G$, and so $v_i$ has degree~$2$ in $G$ with $u_i$ and $u_i'$ as its ends for $i \in [k]$. Recall that a neighbor edge of the path $P$ is an edge that does not belong to the path, but is incident with some internal vertex on~$P$. By our earlier assumption, every such edge is incident with exactly one internal vertex of the path~$P$. 

\begin{claim}\label{c:no-special-edge}
The path $P$ contains a neighboring edge that is a \rspecial\, or a \bspecial.
\end{claim}
\proof Suppose, to the contrary, that no neighboring edge of the path $P$ is a \rspecial\ or a \bspecial. Thus, every neighboring edge of $P$ belongs to exactly one of the following structures: a \redi, a \grei, a \gisthmus, a \risthmus, or a \bdetour.

In Section~\ref{S:types}, we analysed the structure of the extremities of maximal alternating \green-\black\ paths, and showed that there are three types of structures, namely Type-(1), -(2), and -(3). We presented two strategies to apply on each of these three structural types, and we summarized in Table~\ref{table:strategies} the surplus or deficit weight resulting from applying each strategy. In order to define the score of the path $P$, which we denote by $\score(P)$, we introduce some additional notation.
\begin{itemize}
\item Let $t_1$ be the indicator variable to indicate whether the left extremity $u_1$ of the path $P$ belongs to a structure of type-(1), that is, $t_1 = 1$ if the structure containing $u_1$ is of Type-(1), and $t_1 = 0$ otherwise.
\item Let $t_1'$ be the indicator variable to indicate whether the right extremity $u_k'$ of the path $P$ belongs to a structure of type-(1), that is, $t_1' = 1$ if the structure containing $u_k'$ is of Type-(1), and $t_1' = 0$ otherwise.
\item Let $r_\island$ be the number of \redis, and let $g_\island$ be the number of \greis.
\item Let $n_\isthmus$ be the number of \risthmus\, or \gisthmus\, incident with two upper edges, and let $n_{\isthmus}'$ be the number of \risthmus\, or \gisthmus\, incident with two lower edges.
\item Let $m_\isthmus$ be the number of \risthmus\, or \gisthmus\, incident with one upper edge and one lower edge.
\item Let $b_\detour$ be the number of upper edges that are \bdetour\, edges, and let $b_{\detour}'$ be the number of lower edges that are \bdetour\, edges.
\end{itemize}

We are now in a position to define the score of the path $P$, which recall is denoted by $\score(P)$.
\begin{itemize}
\item Initially, we let $\score(P) = 0$.
\item Applying Strategy~1 to the extremity $u_1$ of the path $P$ results in a surplus weight $2$, $1$ or $1$ depending on whether $u_1$ belongs to a structure of Type-(1), -(2), and -(3), respectively (see Table~\ref{table:strategies}.) This adds $1 + t_1$ to $\score(P)$.
\item Applying Strategy~2 to the extremity $u_k'$ of the path $P$ results in a deficit weight of $1$ if $u_k'$ belongs to a Type-(1) structure, and otherwise it results in a balanced weight of~$0$ (see Table~\ref{table:strategies}). This removes $t_1'$ from $\score(P)$.
\item Removing all vertices of the path $P$, we add to a MD-set of the resulting graph all the upper vertices of $P$, that is, the $k-1$ vertices $u_i'$ where $i \in [k-1]$. For each such vertex $u_i'$ we can uniquely associate its two neighbors $v_i$ and $u_{i+1}$ on the path $P$. The sum of the weights of these three vertices $v_i, u_i', u_{i+1}$ is $2 \times 4 + 1 \times 5 = 13$, which results in a surplus weight of~$1$ for each $i \in [k-1]$. This adds $k-1$ to $\score(P)$.
\item For each \redi, we add $2$ to $\score(P)$, while for each \grei, we add $1$ to $\score(P)$. (See the earlier discussion on the counting associated with \redis\, and \greis.) This adds $2r_\island  + g_\island$ to $\score(P)$.
\item For each \risthmus\, and \gisthmus\, incident with two upper edges, we add $1$ to $\score(P)$, while for each \risthmus\, and \gisthmus\, incident with two lower edges we remove~$3$ from $\score(P)$. Further, for each \risthmus\, and \gisthmus\, incident with one upper edge and one lower edge, we add $1$ to $\score(P)$. This adds $n_\isthmus + m_\isthmus - 3n_\isthmus'$ to $\score(P)$.
\item For each \bdetour\, incident with a lower edge, we remove~$1$ from $\score(P)$.
\end{itemize}
We note that $r_\island \ge 0$, $g_\island \ge 0$, and $m_\isthmus \ge 0$. Thus, the score of the path $P$ is given by
\begin{equation}
\label{Eq1:scoreP}
\begin{aligned}
\score(P) = & \,(1 + t_1) - t_1' + (k-1) + 2r_\island  + 2g_\island \\
& \hspace*{0.5cm}  + (n_\isthmus + m_\isthmus - 3n_\isthmus') - b_{\detour}' \\
\ge & \,  t_1 - t_1' + k + (n_\isthmus - 3n_\isthmus') - b_{\detour}'.
\end{aligned}
\end{equation}

By the \emph{reverse} of the path $P$, denoted $\overleftarrow{P}$, we mean the path $P$ in the reverse direction that starts at $u_k'$ and ends at $u_1$, that is,
\[
\overleftarrow{P} \colon u_k', u_k, u_{k-1}', u_{k-1}, \ldots, u_2', u_2, u_1', u_1.
\]

We note that the roles of the vertices $u_i$ and $u_i'$ are reversed when we consider the reverse, $\overleftarrow{P}$, of the path $P$, as are the roles of the upper and lower edges, noting that in this case we apply Strategy~1 to the extremity $u_k'$ of $\overleftarrow{P}$, and Strategy~2 to the extremity $u_1$ of $\overleftarrow{P}$. Thus when we remove all vertices of the path $\overleftarrow{P}$, we now add to a MD-set of the resulting graph the $k-1$ vertices $u_i$ where $i \in [k] \setminus \{1\}$, etc. The score of the path $\overleftarrow{P}$ is given by
\begin{equation}
\label{Eq1:scorePr}
\begin{aligned}
\score(\overleftarrow{P}) = & \,(1 + t_1') - t_1 + (k-1) + 2r_\island  + 2g_\island \\
& \hspace*{0.5cm}  + (n_\isthmus' + m_\isthmus - 3n_\isthmus) - b_{\detour} \\
\ge & \,  t_1' - t_1 + k + (n_\isthmus' - 3n_\isthmus) - b_{\detour}.
\end{aligned}
\end{equation}

\begin{subclaim}\label{c:sum-score}
$\score(P) + \score(\overleftarrow{P}) \ge 2$.
\end{subclaim}
\proof By Equation~(\ref{Eq1:scoreP}) and~(\ref{Eq1:scorePr}), we have
\begin{equation}
\label{Eq1:score}
\score(P) + \score(\overleftarrow{P}) \ge 2k - 2n_\isthmus - 2n_\isthmus' - b_{\detour} - b_{\detour}'.
\end{equation}

Recall that an upper vertex of the path $P$ is a vertex $u_i'$ for some $i \in [k-1]$, and a lower vertex of $P$ is a vertex $u_i$ for some $i \in [k] \setminus \{1\}$. Let $m_\upper(P)$ and $m_\lowr(P)$ denote the number of upper and lower edges, respectively, of $P$. We note that
\[
k-1 = m_\upper(P) \ge 2n_\isthmus + b_{\detour}
\]
and
\[
k-1 = m_\lowr(P) \ge 2n_\isthmus' + b_{\detour}'.
\]
Thus,
\begin{equation}
\label{Eq1:sum-score}
2k \ge 2 + 2n_\isthmus + b_{\detour} + 2n_\isthmus' + b_{\detour}'. \2
\end{equation}
Hence, by Equations~(\ref{Eq1:score}) and~(\ref{Eq1:sum-score}), we have $\score(P) + \score(\overleftarrow{P}) \ge 2$.~\smallqed

\medskip
As an immediate consequence of Claim~\ref{c:sum-score}, we have the following result.

\begin{subclaim}\label{c:positive-score}
If the score of the path $P$ is negative, then the score of $\overleftarrow{P}$ is positive.
\end{subclaim}

Interchanging the names of the path $P$ and the reverse of $\overleftarrow{P}$ if necessary, we may assume by Claim~\ref{c:positive-score} that $\score(P) \ge 1$. This implies that there exists a set $R$ of vertices of $G$ whose removal from $G$ produces a graph $H = G - R$ such that $\mdom(G) \le \mdom(H) + \alpha$ and $\w(G) \ge \w(H) + \beta + 1$, where $\beta = 12 \alpha$, contradicting Fact~\ref{fact1}. This completes the proof of Claim~\ref{c:no-special-edge}.~\smallqed

\medskip
By Claim~\ref{c:no-special-edge}, the path $P$ contains a neighboring edge that is a \rspecial\, or a \bspecial. Let $v_e$ be the first vertex of the path $P$ that is incident with a special edge, and let $e$ be the special edge incident with~$v_e$. Thus, $v_e = u_q'$ for some $q \in [k-1]$ or $v_e = u_q$ for some $q \in [k] \setminus \{1\}$. We consider the two cases in turn.

\medskip
\emph{Case~1. $v_e = u_q'$ for some $q \in [k-1]$.} In this case, we consider the path
\[
Q \colon u_1,u_1',u_2,u_2',\ldots,u_q,u_q'.
\]

We proceed as in the proof of Claim~\ref{c:no-special-edge}, except that the score of $Q$, which we denote by $\score(Q)$, differs only on the extremity $u_q'$ of the path $Q$. We apply Strategy~1 to the extremity $u_1$ of the path $Q$, and we remove all vertices of the path $Q$, and add to a MD-set of the resulting graph all the upper internal vertices of $Q$, that is, the $q-1$ vertices $u_i'$ where $i \in [q-1]$.

By the minimality of the vertex $v_e = u_q'$, no upper or lower edge incident with an upper or lower internal vertex of the path $Q$ is a \rspecial\, or a \bspecial. Our scoring associated with internal vertices of the path $Q$ therefore proceeds exactly as in the proof of Claim~\ref{c:no-special-edge}, except for the extremity $u_q'$ of the path $Q$ which we discuss next. Let $v_e'$ be the neighbor of $v_e$ in $M_G$ that is incident with the special edge~$e$.

\medskip
\emph{Case~1.1 The edge $e$ is a \rspecial.} In this case, we remove the vertex $(v_ev_e')_1$, and we mark the vertex $u_{q+1}$. We can uniquely associate the vertices $v_q$ and $(v_ev_e')_1$ with the vertex $u_q'$. The sum of the weights of the three vertices $v_q, u_q', (v_ev_e')_1$ is $1 \times 4 + 2 \times 5 = 14$, which results in a surplus weight of~$2$ since we will add the vertex $u_q'$ to the MD-set (which decreases the weight by~$12$). However, the degree of the vertex $(v_ev_e')_2$ decreases from~$2$ to~$1$ with the removal of the vertex $(v_ev_e')_1$, and this decreases the weight by~$3$. The net gain/loss to the score of $Q$ is~$+ 2 - 3 = -1$. Thus, the score of the path $Q$ is given by
\begin{equation}
\label{Eq1:scoreQ}
\begin{aligned}
\score(Q) = & \,(1 + t_1) - 1 + (q-1) + 2r_\island  + 2g_\island \\
& \hspace*{0.5cm}  + (n_\isthmus + m_\isthmus - 3n_\isthmus') - b_{\detour}' \\
\ge & \,  t_1 + q - 1 + (n_\isthmus - 3n_\isthmus') - b_{\detour}',
\end{aligned}
\end{equation}

\noindent
where $t_1$ is as defined as before, and where the other parameters above, namely $r_\island$, $g_\island$, $n_\isthmus$, $n_\isthmus'$, etc., are defined analogously as before but restricted in their definition to the path $Q$. We now consider the reverse of the path $Q$, denoted $\overleftarrow{Q}$, that is,
\[
\overleftarrow{Q} \colon u_q', u_q, u_{q-1}', u_{q-1}, \ldots, u_1', u_1.
\]

As before, we apply Strategy~2 to the extremity $u_1$, and we remove all vertices from the path $Q$, and in this case add to a MD-set of the resulting graph all the lower internal vertices of $Q$, that is, the $q-1$ vertices $u_i'$ where $i \in [q] \setminus \{1\}$. We note that in Strategy~2 the vertex $u_1$ is added to the MD-set. Once again, our scoring associated with internal vertices of the path $Q$ therefore proceeds exactly as in the proof of Claim~\ref{c:no-special-edge}, except for the extremity $u_q'$ of the path $Q$ which we discuss next.

We remove the vertices $(v_ev_e')_1$ and $(v_ev_e')_2$, and add the vertex $(v_ev_e')_1$ to the MD-set. We can uniquely associate the vertices $(v_ev_e')_2$ and $u_q'$ with the vertex $(v_ev_e')_1$. The sum of the weights of the three vertices $u_q', (v_ev_e')_1,  (v_ev_e')_2$ is $1 \times 4 + 2 \times 5 = 14$, which results in a surplus weight of~$2$, since we will add the vertex $(v_ev_e')_1$ to the MD-set (which decreases the weight by~$12$). However, the degree of the vertices $v_e'$ and $u_{q+1}$ both decrease from~$3$ to~$1$ with the removal of the vertices $(v_ev_e')_1$ and $(v_ev_e')_2$, and this decreases the weight by~$2$. The net gain/loss to the score of $Q$ is~$+ 2 - 2 = 0$, and so the extremity $u_q'$ contributes~$0$ to the score of $\overleftarrow{Q}$. The score of the path $\overleftarrow{Q}$ is given by
\begin{equation}
\label{Eq1:scoreQr}
\begin{aligned}
\score(\overleftarrow{Q}) = & \, - t_1 + (q-1) + 2r_\island  + 2g_\island \\
& \hspace*{0.5cm}  + (n_\isthmus' + m_\isthmus - 3n_\isthmus) - b_{\detour} \\
\ge & \,  - t_1 + q - 1 + (n_\isthmus' - 3n_\isthmus) - b_{\detour}.
\end{aligned}
\end{equation}
By Inequalities~(\ref{Eq1:scoreQ}) and~(\ref{Eq1:scoreQr}), we have
\begin{equation}
\label{Eq1:sum-scoreQ}
\score(Q) + \score(\overleftarrow{Q}) \ge 2q - 2 - 2n_\isthmus - 2n_\isthmus' - b_{\detour} - b_{\detour}'.
\end{equation}

Let $m_\upper(Q)$ and $m_\lowr(Q)$ denote the number of upper and lower (neighboring) edges, respectively, of internal vertices of $Q$. We note that $q-1 = m_\upper(Q) \ge 2n_\isthmus + b_{\detour}$ and $q-1 = m_\lowr(Q) \ge 2n_\isthmus' + b_{\detour}'$, and so
\begin{equation}
\label{Eq1:scoreB}
2q - 2 \ge 2n_\isthmus + b_{\detour} + 2n_\isthmus' + b_{\detour}'.
\end{equation}

Hence, by Inequalities~(\ref{Eq1:sum-scoreQ}) and~(\ref{Eq1:scoreB}), we have $\score(Q) + \score(\overleftarrow{Q}) \ge 0$. Thus, $\score(Q) \ge 0$ or $\score(\overleftarrow{Q}) \ge 0$. This implies that there exists a set $R$ of vertices of $G$ whose removal from $G$ produces a graph $H = G - R$ such that $\mdom(G) \le \mdom(H) + \alpha$ and $\w(G) \ge \w(H) + \beta$, where $\beta = 12 \alpha$, contradicting Fact~\ref{fact1}.

\medskip
\emph{Case~1.2 The edge $e$ is a \bspecial.} In this case, we remove the vertex $v_e'$, and we mark the vertex $u_{q+1}$. We can uniquely associate the vertices $v_q$ and $v_e'$ with the vertex $v_e$ (where recall that here $v_e = u_q'$). The sum of the weights of these three vertices is $1 \times 5 + 2 \times 4 = 13$, which results in a surplus weight of~$1$ since we will add the vertex $u_q'$ to the MD-set (which decreases the weight by~$12$). However, the removal of the vertex $v_e'$ decreases the weight of its two neighbors different from $v_e$ by~$2 \times -1 = - 2$. The net gain/loss to the score of $Q$ is~$+ 1 - 2 = -1$. Thus, the extremity $u_q'$ contributes~$-1$ to the score of $Q$, which is the same as the previous case when $e$ is a \rspecial. Hence, the score of the path $Q$ again satisfies Inequality~(\ref{Eq1:scoreQ}), that is,
\begin{equation}
\label{Eq1:scoreQnew}
\score(Q) \ge t_1 + q - 1 + (n_\isthmus - 3n_\isthmus') - b_{\detour}'.
\end{equation}

The score for the reverse $\overleftarrow{Q}$ of the path $Q$ is as in Case~1.1 when $e$ is a \rspecial, except for the extremity $u_q'$ of the path $Q$. In this case, we do not remove the vertex $u_q'$ from the path $Q$. The degree of $u_q'$ decreases from~$3$ to~$2$, and results in the removal of~$-1$ from the score of $\overleftarrow{Q}$. Thus, the score of $\overleftarrow{Q}$ is one smaller than in Case~1.1, implying that
\begin{equation}
\label{Eq1:scoreQrnew}
\begin{aligned}
\score(\overleftarrow{Q}) \ge - t_1 + q - 2 + (n_\isthmus' - 3n_\isthmus) - b_{\detour}.
\end{aligned}
\end{equation}

By  Inequalities~(\ref{Eq1:scoreQnew}) and~(\ref{Eq1:scoreQrnew}), we have
\begin{equation}
\label{Eq1:sum-scoreQnew}
\score(Q) + \score(\overleftarrow{Q}) \ge 2q - 3 - 2n_\isthmus - 2n_\isthmus' - b_{\detour} - b_{\detour}'.
\end{equation}

Hence, by Inequalities~(\ref{Eq1:sum-scoreQnew}) and~(\ref{Eq1:scoreB}), we have $\score(Q) + \score(\overleftarrow{Q}) \ge -1$. Since $\score(Q)$ and $\score(\overleftarrow{Q})$ are integer values, we deduce that $\score(Q) \ge 0$ or $\score(\overleftarrow{Q}) \ge 0$. This implies that there exists a set $R$ of vertices of $G$ whose removal from $G$ produces a graph $H = G - R$ such that $\mdom(G) \le \mdom(H) + \alpha$ and $\w(G) \ge \w(H) + \beta$, where $\beta = 12 \alpha$, contradicting Fact~\ref{fact1}.

\medskip
\emph{Case~2. $v_e = u_q$ for some $q \in [k] \setminus \{1\}$.} In this case, we consider the path
\[
Q \colon u_1,u_1',u_2,u_2',\ldots,u_{q-1},u_{q-1}',u_q.
\]

By the minimality of the vertex $v_e = u_q$, no upper or lower edge incident with an upper or lower internal vertex of the path $Q$ is a \rspecial\, or a \bspecial. We proceed as in the proof of Claim~\ref{c:no-special-edge}, except that the score of $Q$ differs on the extremity $u_q'$ of the path $Q$. We apply Strategy~2 to the extremity $u_1$ of the path $Q$, and we remove all vertices of the path $Q$, and add to a MD-set of the resulting graph all the lower internal vertices of $Q$, that is, the $q-2$ vertices $u_i$ where $i \in [q-1] \setminus \{1\}$. We note that in Strategy~2 the vertex $u_1$ is added to the MD-set. Let $v_e'$ be the neighbor of $v_e$ in $M_G$ that is incident with the special edge $e$. We consider two cases, depending on whether the edge $e$ is a \rspecial\, or a \bspecial.

\medskip
\emph{Case~2.1 The edge $e$ is a \rspecial.} To take care of the extremity~$u_q$ of the path $Q$, we remove the vertices $(v_ev_e')_1$ and $(u_qu_q')$ (where recall that $v_e = u_q$). We can uniquely associate the vertices $u_{q-1}'$, $(v_ev_e')_1$ and $(u_qu_q')$ with the vertex $u_q$. The sum of the weights of these four vertices is $2 \times 4 + 2 \times 5 = 18$, which results in a surplus weight of~$6$ since we will add the vertex $u_q$ to the MD-set (which decreases the weight by~$12$). However, removing these four vertices decrease the degree of $(v_ev_e')_2$ from~$2$ to~$1$, and decreases the degrees of $u_q'$ and the neighbor of $u_{q-1}'$ not on the path $Q$ from~$3$ to~$2$. This decreases the weight by~$-1 - 1 - 3 = -5$. The net gain to the score of $Q$ is~$6 - 5 = 1$. Thus, the score of the path $Q$ is given by
\begin{equation}
\label{Eq1:scoreQ-case21}
\begin{aligned}
\score(Q) = & \,- t_1 + 1 + (q-2) + 2r_\island  + 2g_\island \\
& \hspace*{0.5cm}  + (n_\isthmus' + m_\isthmus - 3n_\isthmus) - b_{\detour} \\
\ge & \,  -t_1 + q - 1 + (n_\isthmus' - 3n_\isthmus) - b_{\detour}.
\end{aligned}
\end{equation}

\noindent
We now consider the reverse of the path $Q$, namely the path
\[
\overleftarrow{Q} \colon u_q, u_{q-1}', u_{q-1}, \ldots, u_1', u_1.
\]

We now apply Strategy~1 to the extremity $u_1$, and we remove all vertices from the path $Q$, except for the vertex $u_q$. In this case add to a MD-set of the resulting graph all the upper internal vertices of $Q$, that is, the $q-1$ vertices $u_i'$ where $i \in [q] \setminus \{1\}$. The degree of $u_q$ decreases from~$3$ to~$2$, and results in the removal of~$-1$ from the score of $\overleftarrow{Q}$. The score of the path $\overleftarrow{Q}$ is given by
\begin{equation}
\label{Eq1:scoreQ-case21r}
\begin{aligned}
\score(\overleftarrow{Q}) = & \,(1 + t_1) - 1 + (q-1) + 2r_\island  + 2g_\island \\
& \hspace*{0.5cm}  + (n_\isthmus + m_\isthmus - 3n_\isthmus') - b_{\detour}' \\
\ge & \,  t_1 + q - 1 + (n_\isthmus - 3n_\isthmus') - b_{\detour}'.
\end{aligned}
\end{equation}

By Inequalities~(\ref{Eq1:scoreQ-case21}) and~(\ref{Eq1:scoreQ-case21r}), we have
\begin{equation}
\label{Eq1:sum-scoreQ-case21}
\score(Q) + \score(\overleftarrow{Q}) \ge 2q - 2 - 2n_\isthmus - 2n_\isthmus' - b_{\detour} - b_{\detour}'.
\end{equation}

Let $m_\upper(Q)$ and $m_\lowr(Q)$ denote the number of upper and lower (neighboring) edges, respectively, of internal vertices of $Q$. We note that $q-1 = m_\upper(Q) \ge 2n_\isthmus + b_{\detour}$ and $q-2 = m_\lowr(Q) \ge 2n_\isthmus' + b_{\detour}'$, and so
\begin{equation}
\label{Eq1:scoreC-case21}
2q - 2 \ge 1 + 2n_\isthmus + b_{\detour} + 2n_\isthmus' + b_{\detour}'.
\end{equation}

Hence, by Inequalities~(\ref{Eq1:sum-scoreQ-case21}) and~(\ref{Eq1:scoreC-case21}), we have $\score(Q) + \score(\overleftarrow{Q}) \ge 1$. Thus, $\score(Q) \ge 1$ or $\score(\overleftarrow{Q}) \ge 1$. This implies that there exists a set $R$ of vertices of $G$ whose removal from $G$ produces a graph $H = G - R$ such that $\mdom(G) \le \mdom(H) + \alpha$ and $\w(G) \ge \w(H) + \beta$, where $\beta = 12 \alpha$, contradicting Fact~\ref{fact1}.

\smallskip
\emph{Case~2.2 The edge $e$ is a \bspecial.} In this case, to take care of the extremity~$u_q$ of the path $Q$, we remove the vertices $v_e'$ and $(u_qu_q')$ (where recall that $v_e = u_q$). We can uniquely associate the vertices $u_{q-1}'$, $v_e'$ and $(u_qu_q')$ with the vertex $u_q$. The sum of the weights of these four vertices is $3 \times 4 + 1 \times 5 = 17$, which results in a surplus weight of~$5$ since we will add the vertex $u_q$ to the MD-set (which decreases the weight by~$12$). However, removing these four vertices decrease the degrees of four neighboring vertices (the vertex $u_q'$, the neighbor of $u_{q-1}'$ not on the path $Q$, and the two neighbors of $v_e'$ different from $v_e$) from~$3$ to~$2$. This decreases the weight by~$-4$. The net gain to the score of $Q$ is~$5 - 4 = 1$. Thus, the score of the path $Q$ is again given by Inequality~(\ref{Eq1:scoreQ-case21}).

The score for the reverse $\overleftarrow{Q}$ of the path $Q$ is exactly as in Case~2.1 when $e$ is a \rspecial, except for the extremity $u_q$ of the path $Q$. Thus, we apply Strategy~1 to the extremity $u_1$, and we remove all vertices from the path $Q$, except for the vertex $u_q$, and we add to a MD-set of the resulting graph all the upper internal vertices of $Q$, that is, the $q-1$ vertices $u_i'$ where $i \in [q] \setminus \{1\}$. As before, Inequalities~(\ref{Eq1:scoreQ-case21r}),~(\ref{Eq1:sum-scoreQ-case21}) and~(\ref{Eq1:scoreC-case21}) hold, and we contradict contradicting Fact~\ref{fact1}.

We have therefore proven the following result.

\begin{claim}\label{c:no-max-green-black-path}
Every maximal alternating \green-\black\ path in $M_G$ is a cycle.
\end{claim}

\subsection{Maximal alternating green-black cycles}

In this section, we consider alternating \green-\black\ cycles in $M_G$. Let
\[
C \colon u_1,u_1',u_2,u_2',\ldots,u_k,u_k'
\]
be an alternating \green-\black\ cycle in $M_G$, and so the edges \(u_iu_i'\) are \grees\ for $i \in [k]$ and edges \(u_i'u_{i+1}\) are \blaes\ for $i \in [k]$ where addition is taken modulo~$k$. Adopting our earlier notation, we let the \gree\ \(u_iu_i'\) correspond to the path $u_i, v_i, u_i'$ in the graph $G$, and so $v_i$ has degree~$2$ in $G$ with $u_i$ and $u_i'$ as its ends for $i \in [k]$. We note that $v_i = (u_iu_i')_1$ for $i \in [k]$. The cycle $C$ when $k = 4$ is illustrated in Figure~\ref{f:gb-cycle}.

\begin{figure}[htb]
\begin{center}
\input{green-black-cycle.tex}
\caption{An alternating \green-\black\ cycle}\label{f:gb-cycle}
\end{center}
\end{figure}

We adopt our notation employed in the proof of Claim~\ref{c:no-special-edge}, except that we apply the notation to the cycle $C$ rather than the path $P$. In particular, a vertex $u_i'$ is an \emph{upper vertex} and a vertex $u_i$ is a \emph{lower vertex} for all $i \in [k]$. Every neighboring edge of $C$ that is incident with an upper vertex we call an \emph{upper edge}, while every neighboring edge of $C$ that is incident with a lower vertex we call a \emph{lower edge}. An edge $e$ is a \emph{neighbor} of the cycle $C$ if it does not belong to the cycle $C$, but is incident with a vertex on the cycle. Recall that a neighboring \blae\, $e = uv$ of the cycle $C$ is a \bspecial\, if one end, say $u$, of the edge $e$ belongs to the cycle $C$ and the other end, $v$, is the center of a black star.

In the following, we consider a coloring of the \grees\ in the \green-\black\ cycle $C$. This coloring uses two colors, \Amber\ and \Blue, and will be decided in Section~\ref{s:set-colors}. We now define the sets $\setA$ and $\setB$ of vertices in the \green-\black\ cycle $C$ as follows.

\begin{definition}
\label{defn:AB-sets}
{\rm
We define the set $\setA$ to consist of the lower vertices in $C$ that belong to \Amber\ edges and the upper vertices in $C$ that belong to \Blue\ edges, and, analogously, we define the set $\setB$ to consist of the upper vertices in $C$ that belong to \Amber\ edges and the lower vertices in $C$ that belong to \Blue\ edges.
We will also refer to vertices in the set $\setA$ as \emph{amber vertices} and vertices in the set $\setB$ as \emph{blue vertices}, thereby producing an amber-blue coloring of the vertices of $C$.

}
\end{definition}

We note that the sets $\setA$ and $\setB$ partition the vertices of the \green-\black\ cycle $C$, and so $\setA \cup \setB$ is the set of all upper and lower vertices of $C$. Further, one end of every \gree\ in the cycle $C$ belongs to the set $\setA$ and the other end belongs to the set $\setB$.
However, a \blae\ of $C$ may possibly have both its ends in the set $\setA$ (resp., $\setB$),
and the amber-blue coloring of the vertices of $C$ is not necessarily a proper coloring.


\begin{definition}
\label{defn:dotted}
{\rm
When a \blae\ is incident to two \grees\ of different colors (one \Amber\ and one \Blue), we say the \blae\ is a \emph{color change}. The extremities of this \blae\ are said to be \emph{dotted vertices}.
}
\end{definition}

Given a \blae\ corresponding to a change of color $u'_iu_{i+1}$,
if the \gree\ incident to the lower vertex $u_{i+1}$ is \Amber\
and the \gree\ incident to the upper vertex $u'_i$ is \Blue,
then both extremities of the \blae\ belong to \setA, and neither is dominated when choosing the set \setB.
Conversely, if the \gree\ incident to the lower vertex $u_{i+1}$ is \Blue\
and the \gree\ incident to the upper vertex $u'_i$ is \Amber,
then both extremities of the \blae\ belong to \setB, and neither is dominated when choosing the set \setA.
Note that there necessarily are the same number of dotted amber and dotted blue vertices,
which is equal to the number of change of colors.

When the set $\setA$ is chosen, we delete the vertices in $\setA$ and the blue vertices dominated by $\setA$. However, we do not delete the dotted blue vertices. Analogously, when the set $\setB$ is chosen, we delete the vertices in $\setB$ and the amber vertices dominated by $\setB$. However, we do not delete the dotted amber vertices.
In both cases, since all dotted vertices are initially of degree 3,
we lose a contribution of~$4$ for each change of color.

We are now in a position to explain the updated counting associated with the amber-blue coloring of the vertices of the \green-\black\ cycle $C$. Our strategy is to remove the set $\setA$ of vertices (colored amber) and all vertices in $\setB$ (colored blue) on $C$ dominated by the set $\setA$, together with the internal vertex of degree~$2$ from every \gree\ of $C$, noting that every such vertex is dominated by the set $\setA$ in the graph $G$. However, we do not delete the dotted blue vertices (noting that these vertices are not dominated by the set $A$). In addition, we delete or mark all neighboring vertices of the cycle $C$ that are dominated by the set $\setA$. Therefore in the graph that results from choosing the set $\setA$ and applying the above rules, every dotted blue vertex, which originally had degree~$3$ in $G$, is now a vertex of degree~$2$, and this decrease in its degree will contribute~$-1$ to the overall cost of selecting the set $\setA$. However if in this resulting graph, a dotted blue vertex becomes marked (which occurs if one of its neighbors outside the cycle $C$ is added to the dominating set $\setA$), then the contribution of the dotted blue vertex in the resulting graph is~$0$ noting that the weight of a degree~$3$ vertex and a marked vertex is the same, that is, we save~$1$ in our counting for each dotted blue vertex in the resulting graph that becomes marked. We will frequently use this fact in our counting arguments.

\begin{definition}\label{defn:greencyclegraph}
The \greencyclegraph\ is a subgraph of $G$ consisting of an alternating \green-\black\ cycle, together with some extra artifacts referred to as \links\ or fibers, respectively, joining two vertices (called the \emph{extremities} of the \link) on the cycle or attached to a single extremity.
Each added \link\ corresponds to a path of length~$2$ between
its extremities in the original graph, where the central vertex is the center of a black star. Each fiber corresponds to a black star attached to the fiber's extremity (as in Claim~\ref{c:claim23}).
\end{definition}


A \link\ is said to be \emph{well colored} if it has one end in $\setA$ and one end in $\setB$. It is said to be \emph{badly colored} if both ends are in $\setA$ or both in $\setB$. If any extremity of a \link\ is dotted, then we say the \link\ is dotted.
Similarly, a fiber is dotted if its extremity is dotted.

\subsubsection{Counting arguments associated with the \green-\black\ cycle}

In this section, we perform counting arguments, according to the different structures associated with the \green-\black\ cycle $C$.
The strategy is then to show that the counting when choosing at least one of the sets
$\setA$ and $\setB$ is nonnegative. Since the graph $G$ is a minimum counterexample, we apply the desired result to the graph resulting from choosing one of the sets $\setA$ or $\setB$ that yields the maximum weight gain (which is nonnegative).

We consider the different possible structures associated with the \green-\black\ cycle $C$. We usually make an assumption that some vertex in the cycle is Amber, the other case being symmetrical. Then, we represent the resulting scoring by $[a \mid b \,]$, where the first entry $a$ is the score when $\setA$ is played, and the second entry $b$ is the score when $\setB$ is played. We say $[a \mid b \,]$ is at least $[c \mid d \,]$ if $a\ge c$ and $b\ge d$.

At the end, we use an overall argument stating that since the sum of the scores
when $\setA$ is played and when $\setB$ is played is nonnegative,
one of these choices must be nonnegative. Anticipating this argument, we use an
\emph{average score} for $a$ and $b$, and we simplify the notation $[a \mid b \,]$ into $[ \frac{a+b}{2} ]$.
We also use addition and comparison and average scores, since these operations
are consistent in regards of the sign of the final score.

In the following, results are stated for the average score, but the proof details use
the scores $[a \mid b \,]$.

\begin{claim}\label{c:score-red-island}
The average score of a \redi\ is at least $[\frac{1}{2}]$.
\end{claim}
\proof  Consider a \redi\ shown in Figure~\ref{f:red-island2}, where $v'$ and $v''$ are the ends of the edge $e'$ and $e''$, respectively, different from the vertex~$v$.

\begin{figure}[htb]
\begin{center}
\input{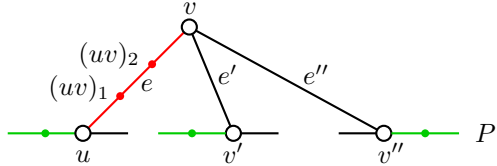}
\vskip -0.75 cm
\caption{A red island}\label{f:red-island2}
\end{center}
\end{figure}

For counting purposes, suppose that the vertex $u$ is colored amber (possibly, dotted), and so $u \in \setA$. We consider two possibilities, depending on the color of $v'$ and $v''$.

\begin{subclaim}
\label{c:score-red-island.1}
If $v'$ or $v''$ is colored blue, then the average score is at least $[1]$.
\end{subclaim}
\proof Suppose that $v'$ or $v''$, say $v'$, is colored blue. Let $H$ be obtained from $G$ by removing the vertices $v$, $(uv)_1$, and $(uv)_2$. If the set $A$ is chosen, then we add the vertex~$v$ to the dominating set $\setA$, and mark any dotted blue neighbor of $v$, if it exists. The score we obtain is $+5+5+4-12+\ell = +2+\ell$ where $\ell$ stands for the number of dotted blue neighbors of $v$. If the set $B$ is chosen, then we add the vertex~$(uv)_1$ to the dominating set $B$,     and mark the vertex $u$ if it is a dotted amber vertex. We remove the vertices $v$, $(uv)_1$, and $(uv)_2$. If $v''$ is a dotted amber vertex, we also cut the edge $vv''$, resulting in the vertex $v''$ having degree~$1$, which costs an extra $-3$. This yields a score of at least $+4 +5 +5 +1 -3 -12 = 0$. Thus, the score if $v'$ or $v''$ is in $\setB$ is at least $[2 \mid 0]$.~\smallqed

\begin{subclaim}
\label{c:score-red-island.2}
If both $v'$ and $v''$ are colored amber, then the average score is at least is $[\frac{1}{2}]$.
\end{subclaim}
\proof Suppose that both $v'$ and $v''$ are colored amber. If the set $A$ is chosen, then we let $H$ be obtained from $G$ by removing the vertices $v$, $(uv)_1$, and $(uv)_2$, and adding the vertex~$v$ to the dominating set $A$, resulting in a score of~$2$. Suppose the set $B$ is chosen. We consider three possible subcases.

Suppose that the vertex $u$ is a dotted amber vertex, and that at least one of $v'$ and $v''$  is a dotted amber vertex. In this case, we do nothing in the sense that we do not delete any of the vertices $v$, $(uv)_1$, and $(uv)_2$. This contributes~$0$ to the count, either because all three vertices are dotted, or if only one of $v'$ and $v''$ is not dotted, then the cutting of the corresponding edge was already taken into account. In this case, the score is $[2 \mid 0]$.

Suppose that the vertex $u$ is not a dotted amber vertex. In this case, we let $H$ be obtained from $G$ by removing the vertices $(uv)_1$, and $(uv)_2$, and adding the vertex~$(uv)_2$ to the dominating set $B$, and marking the vertex~$v$. We gain~$2 \times 5 = 10$ from deleting the vertices $(uv)_1$ and $(uv)_2$, and we lose~$12$ from adding the vertex $(uv)_2$ to the dominating set. Moreover,  we save the cutting of the outgoing edge joining $u$ and $(uv)_1$, gaining an extra~$1$. The overall contribution is~$10 - 12 + 1 = -1$, yielding the score $[2 \mid -1]$.

Suppose that the vertex $u$ is a dotted amber vertex, and none of $v'$ and $v''$ is a dotted amber vertex. In this case, we let $H$ be obtained from $G$ by removing the vertices $v$ and $(uv)_2$, and adding the vertex~$(uv)_2$ to the dominating set $B$, and marking the vertex~$(uv)_1$. We gain~$9$ from deleting the vertices $(uv)_2$ and $v$, and we lose~$12$ from adding the vertex $(uv)_2$ to the dominating set.  Moreover,  we save the cutting of the outgoing edges joining $v$ to $v'$ and $v''$, gaining an extra~$2$. We also gain~$1$ for marking the vertex $(uv)_1$ which was a vertex of degree~$2$. The overall contribution is~$9 - 12 + 2 + 1= 0$, yielding the score $[2 \mid 0]$.

Thus when both $v'$ and $v''$ are colored amber, the score is at least~$[2 \mid -1]$, yielding an average score of a \redi\ of at least $[\frac{1}{2}]$.~\smallqed

\medskip
In summary, by Claims~\ref{c:score-red-island.1} and~\ref{c:score-red-island.2}, the average score of a \redi\ is at least $[\frac{1}{2}]$.~\smallqed

\begin{claim}\label{c:score-green-island}
The average score of a \grei\ is at least $[0]$.
\end{claim}
\proof  Consider a \grei\ shown in Figure~\ref{f:green-island2}, where $u_1$ and $u_2$ are the ends of the edge $e_1$ and $e_2$, respectively, different from the vertex~$u$, and where $v_3$ and $v_4$ are the ends of the edge $e_3$ and $e_4$, respectively, different from the vertex~$v$. Let $W = \{u_1,u_2,v_3,v_4\}$.

\begin{figure}[htb]
\begin{center}
\input{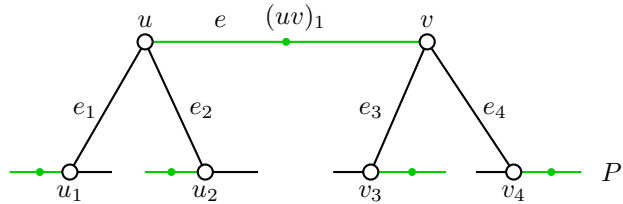}
\vskip -0.75 cm
\caption{A green island}
\label{f:green-island2}
\end{center}
\end{figure}

Suppose that no vertex in $W$ is a dotted amber or dotted blue vertex. In this case, let $H$ be obtained from $G$ by removing the vertices $u$, $v$, and $(uv)_1$. If the set $A$ (resp., $B$) is chosen, we add the vertex $(uv)_1$ to the dominating set $A$ (resp., $B$), resulting in an excess weight of at least~$2 \times 4 + 5 - 12 = 1$. We also increase the score by saving from cutting the edges $e_1$, $e_2$, $e_3$ and $e_4$, but in all cases, the score is always at least $[1 \mid 1]$. In view of the these observations, we may assume that at least one vertex in $W$ is a dotted vertex. By symmetry and for counting arguments, we may assume that the vertex $u_1$ is a dotted amber vertex. In particular, $u_1 \in A$. We now consider a few possibilities that may occur.

\begin{subclaim}
\label{c:score-green-island.1}
If at least one of $v_3$ and $v_4$ is a dotted amber vertex, then the average score is at least $[0]$.
\end{subclaim}
\proof Suppose that at least one of $v_3$ and $v_4$, say $v_3$, is a dotted amber vertex. Suppose firstly that the set $A$ is chosen. Suppose $u_2$ or $v_4$, say $v_4$ by symmetry, is amber. In this case, let $H$ be obtained from $G$ by removing the vertices $u$, $v$, and $(uv)_1$, adding the vertex $u$ to the set $A$, and marking the vertex $u_2$ if it is a dotted blue vertex. This results in an excess weight of~$13 - 12 = 1$ if $u_2$ is not a blue vertex, and an excess weight of~$13 - 12 + 1 = 2$ if $u_2$ is a blue vertex. Suppose $u_2$ and $v_4$ are both blue. (Possibly, both $u_2$ and $v_4$ are dotted blue vertices.) In this case, we do not delete any of the vertices $u$, $v$, and $(uv)_1$, but we mark both $u$ and $v$. This contributes~$0$ to the count.

Suppose next that the set $B$ is chosen. In this case, we do not remove any of the vertices $u$, $v$, and $(uv)_1$. We mark $u$ if $u_2$ is blue and we mark $v$ if $v_4$ is blue. If the vertex $u_2$ is an amber vertex that is not dotted, then we gain~$+1$ for the deletion of the outgoing edge $uu_2$ and we lose~$1$ for the degree of $u$ dropping from~$3$ to~$2$. The net contribution balances out to~$0$. Hence in this case, the contribution to the count is~$0$. As observed earlier, the contribution to the count when $A$ is chosen is~$0$. Thus, the score is at least~$[0 \mid 0]$.~\smallqed

\begin{subclaim}
\label{c:score-green-island.2}
If neither $v_3$ nor $v_4$ is a dotted amber vertex, and at least one of $v_3$ and $v_4$ is blue, then the average score is at least $[1]$.
\end{subclaim}
\proof Suppose that neither $v_3$ nor $v_4$ is a dotted amber vertex, and at least one of $v_3$ and $v_4$ is blue. Suppose firstly that the set $B$ is chosen. In this case, let $H$ be obtained from $G$ by removing the vertices $u$, $v$, and $(uv)_1$, adding the vertex $u$ to the set $B$, and marking its amber dotted neighbors on $C$. We gain~$2 \times 4 + 5 = 13$ by removing the vertices $u$, $v$, and $(uv)_1$, and we lose~$-12$ by adding the vertex $u$ to the set $B$, resulting in nett gain of~$+1$. Further, for each amber neighbor of $u$ on $C$ that is dotted, we gain~$+1$ from such a dotted amber neighbor noting that the degree drops from~$3$ to~$2$, while for each amber neighbor of $u$ on $C$ that is not dotted, we gain~$+1$ from the outgoing edge from $u$ to that vertex. Hence we have an excess weight of at least~$2$ if $u_2$ is a blue vertex, and an excess weight of at least~$3$ if $u_2$ is an amber vertex (from marking it if it is dotted, or deleting the outgoing edge $uu_2$ if it is not dotted). Hence, the contribution to the count when the set $B$ is chosen is at least~$2$.

Suppose next that the set $A$ is chosen. If $u_2$ is not a dotted blue vertex, then let $H$ be obtained from $G$ by removing the vertices $u$, $v$, and $(uv)_1$, adding the vertex $v$ to the set $A$, and marking its blue dotted neighbors on $C$, if any. This results in an excess weight of at least~$+1$. If $u_2$ is a dotted blue vertex and at least one of $v_3$ and $v_4$ is an amber vertex, then we mark the vertices $u$ and $v$, yielding a contribution of~$0$ to the count. If $u_2$ is a dotted blue vertex and both $v_3$ and $v_4$ are blue vertices, then let $H$ be obtained from $G$ by removing the vertices $u$, $v$, and $(uv)_1$, adding the vertex $v$ to the set $A$, and marking its blue dotted neighbors on $C$, if any. The vertices $v_3$ and $v_4$ both contribute~$1$ to the count, either by being marked dotted vertices, or by saving the cost of cutting the outgoing edges to~$v$. Cutting the edge $e_2$ cost~$3$ extra since it makes of the blue dotted vertex $u_2$ a degree~$1$ vertex. This results in a contribution of $2 \times 4 + 5 + 2 \times 1 - 12 - 3 = 0$, which counts~$0$ to the overall count. Hence, the contribution to the count when the set $A$ is chosen is at least~$0$. As observed earlier, the contribution to the count when $B$ is chosen is at least~$2$. Thus, the score is at least~$[0 \mid 2]$.~\smallqed

\begin{subclaim}
\label{c:score-green-island.3}
If neither $v_3$ nor $v_4$ is a dotted amber vertex, but both $v_3$ and $v_4$ are amber, then the average score is at least $[0]$.
\end{subclaim}
\proof Suppose that neither $v_3$ nor $v_4$ is a dotted amber vertex, but both $v_3$ and $v_4$ are amber. Suppose firstly that the set $A$ is chosen. In this case, we let $H$ be obtained from $G$ by removing the vertices $u$, $v$, and $(uv)_1$, adding the vertex $u$ to the set $A$, and marking the vertex $u_2$ if it is a dotted blue vertex. This results in an excess weight of~$2$ if $u_2$ is a blue vertex, and an excess weight of~$1$ if $u_2$ is an amber vertex. Hence, the contribution to the count when the set $A$ is chosen is at least~$1$.

Suppose next that the set $B$ is chosen. In this case, let $H$ be obtained from $G$ by removing the vertices $v$ and $(uv)_1$, adding the vertex $(uv)_1$ to the set $B$, and marking the vertex $u$. Both amber vertices $v_3$ and $v_4$ contribute~$1$ to the count for saving the cost of cutting the outgoing edges $e_3$ and $e_4$. This results in a contribution of~$5 + 4 - 12 + 2 \times 1 = -1$ to the overall count. Hence, the contribution to the count when the set $B$ is chosen is at least~$-1$. As observed earlier, the contribution to the count when $A$ is chosen is at least~$1$. Thus, the score is at least~$[1 \mid -1]$.~\smallqed

\medskip
In summary, by Claims~\ref{c:score-green-island.1},~\ref{c:score-green-island.2} and~\ref{c:score-green-island.3}, the average score of a \grei\ is at least $0$.~\smallqed

\begin{claim}\label{c:score-green-isthmus}
The average score of a \gisthmus\ is at least~$[0]$.
\end{claim}
\proof  Consider a \gisthmus\, shown in Figure~\ref{f:green-isthmus2}, where $u_1$ and $u_2$ are the ends of the edge $e_1$ and $e_2$, respectively, different from the vertex~$u$. Suppose that $u_1$ and $u_2$ are of different colors. Whether $A$ or $B$ is chosen, we do not delete the vertex $u$, and instead mark it. This contributes~$0$ to the count, possibly more since we save the cutting of the edge $e_1$ or $e_2$, yielding a score of $[1 \mid 1]$ if none of $u_1,u_2$ is marked and a score of at least $[0,0]$ is some are dotted.

\begin{figure}[htb]
\begin{center}
\input{green-isthmus2.tex}
\vskip -0.75 cm
\caption{A \gisthmus}
\label{f:green-isthmus2}
\end{center}
\end{figure}

Suppose $u_1$ and $u_2$ are of the same color, say amber. If the set $A$ is chosen, then we let $H$ be obtained from $G$ by removing the vertex $u$, resulting in an excess weight of~$+1$. Suppose now the set $B$ is chosen. If neither $u_1$ nor $u_2$ is dotted, then we let $H$ be obtained from $G$ by removing the vertices $u$ and $(uv)_1$, adding the vertex $(uv)_1$ to the set $B$, and marking the vertex~$v$. We get an extra contribution of $2 \times 1 = 2$ for not cutting the edges $e_1$ and $e_2$. This results in a contribution of~$5 + 4 - 12 + 2 = -1$ to the overall count. If at least one of $u_1$ and $u_2$ is a dotted amber vertex, say $u_1$, then we do not delete the vertex~$u$. The overall contribution is then~$0$, noting that the possible cut of the outgoing edge $e_2$ is already taken care of in the counting in the case when $u_2$ is not dotted. Hence, the contribution to the count when the set $B$ is chosen is at least~$-1$. As observed earlier, the contribution to the count when $A$ is chosen (and $u_1$ and $u_2$ are of the same colo) is at least~$1$. Thus, the score is at least~$[1 \mid -1]$.

In summary, if $u_1$ and $u_2$ are of different colors, then the score of at least $[0,0]$, and if $u_1$ and $u_2$ are of the same color, then the score of at least~$[1 \mid -1]$. Thus, the average score of a \gisthmus\ is at least~$[0]$.~\smallqed

\begin{claim}\label{c:score-red-isthmus}
The average score of a \risthmus\, is at least~$[0]$.
\end{claim}
\proof The proof is analogous to the proof of Claim~\ref{c:score-green-isthmus} with the only difference that the vertex $(uv)_2$ may be marked instead of $v$, which is now a vertex of degree~$2$. We gain an extra $+1$ in that case and all contributions are now nonnegative.  This yields a scoring of at least $[0 \mid 0]$.~\smallqed

\begin{claim}\label{c:score-black-detour}
The average score of a \bdetour\ or a \bspecial\  is at least~$[0]$.
\end{claim}
\proof Consider a \bdetour\ or a \bspecial\ as illustrated in Figure~\ref{f:black-detour2}. We may assume that the vertex $u$ is an amber vertex. If the set $A$ is chosen, then we do not delete the vertex $v$, and instead mark it. This contributes~$0$ to the count. If the set $B$ is chosen, then we do not delete the vertex $v$. This contributes~$0$ to the overall count if the vertex $u$ is a dotted amber vertex, and also contributes~$0$ to the overall count if the vertex $u$ is a not a dotted vertex since in this case we gain~$+1$ from deleting the vertex~$u$ but lose~$1$ from the degree of~$v$ dropping from~$3$ to~$2$. This yields a score of at least~$[0 \mid 0]$.~\smallqed

\begin{figure}[htb]
\begin{center}
\input{black-detour2.tex} \hspace{1cm}\input{special-black-edge.tex}
\vskip -0.25 cm
\caption{A \bdetour \, or a \bspecial}
\label{f:black-detour2}
\end{center}
\end{figure}

\begin{claim}\label{c:score-special-red}
The average score of a \rspecial\ is at least~$[\frac{1}{2}]$.
\end{claim}
\proof Consider a \rspecial, as illustrated in Figure~\ref{f:special-red-edge-cycle}. We may assume that the vertex $u$ is an amber vertex. If the set $\setA$ is chosen, then we delete $(uv)_1$ (and cut the edge $e$). We gain $5$ from deleting the vertex $(uv)_1$ and we lose~$3$ from the degree of $(uv)_2$ dropping from~$2$ to~$1$. The overall contribution is~$+2$. Suppose next that the set $\setB$ is chosen. If the vertex $u$ is dotted, then we do nothing (that is, we do not delete a vertex on the \rspecial), which contributes~$0$. If the vertex $u$ is not dotted, then we delete $(uv)_1$ and $(uv)_2$, and we add the vertex~$(uv)_2$ to the dominating set $\setB$, and mark the vertex~$v$. We save the cutting of the outgoing edge joining $u$ and $(uv)_1$, gaining an extra~$1$. The overall contribution is~$2 \times 5 - 12 + 1 = -1$ to the total counting. This yields a score of $[+2 \mid -1]$.~\smallqed

\begin{figure}[htb]
\begin{center}
\input{red-detour.tex}
\vskip -0.25 cm
\caption{A \rspecial}\label{f:special-red-edge-cycle}
\end{center}
\end{figure}


From here on, we start using the special artifacts on the \greencyclegraph\ as defined in Definition~\ref{defn:greencyclegraph}, namely \links\ and fibers. Recall that \links\ have two extremities on the \greencyclegraph, and we say a \link\ is \emph{well colored} if its extremities belong to different sets $\setA$ and $\setB$, while a \link\ is badly colored if its extremities are in the same set. We next introduce two types of \links, namely \emph{broad \links} and \emph{narrow \links}, that satisfy the following properties.

\begin{definition}[Average scores for links]\label{d:score-link}
\begin{itemize}
\item Broad \links\ will have an average score of at least $[+5]$ when the \link\ is well colored, $[-1]$ when badly colored  and when dotted, $[+\frac{1}{2}]$ if only one extremity is dotted, and $[+1]$ if both extremities are dotted.
\item Narrow \links\ will have an average score of $[+4]$ when the \link\ is well colored, $[-\frac{1}{2}]$ when badly colored, and $[0]$ when at least one extremity is dotted.
\end{itemize}
\end{definition}

Moreover, we introduce an artifact called a \emph{fiber} that we attach to a single vertex on the \greencyclegraph\, with the incentive to avoid dotting this vertex. A fiber will have the following property.

\begin{definition}[Scores for fibers]\label{d:score-fiber}
\begin{itemize}
\item Fibers will have a nonnegative score on average, except that they are worth $[-\frac{1}{2}]$ when the associated single vertex on the \greencyclegraph\ is dotted.
\end{itemize}
\end{definition}

\begin{definition}
\label{defn:narrow-link}
{\rm
If the center vertex $v$ of a black $3$-star is adjacent to three vertices $v_1$, $v_2$ and $v_3$ on the \green-\black-cycle $C$ as illustrated in Figure~\ref{f:isolate-3star}, then we call the $3$-star an \emph{isolated} $3$-\emph{star} with respect to the cycle $C$. Further, we introduce a narrow \link\ and a fiber in the \greencyclegraph, where a narrow \link\ is added between $v_1$ and $v_2$ and a fiber is added to $v_3$.
}
\end{definition}

\begin{figure}[htb]
\begin{center}
\input{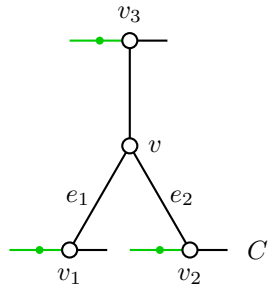}
\vskip -0.35 cm
\caption{An isolated $3$-star}\label{f:isolate-3star}
\end{center}
\end{figure}

In the following claim, we show that a broad \link\ and a narrow \link\ both have an average score at least as large as what the \links\ and fiber induce. Moreover, the average score is nonnegative if there are no \links\ or fibers.

\begin{claim}\label{c:score-3star}
The following properties hold for a narrow \link\ and fiber associated with an isolated $3$-star  in the \greencyclegraph.  \\[-24pt]
\begin{enumerate}
   \item If the \link\ is well colored and the fiber is not dotted, then the average score is at least $[5] \ge [4] + [0]$.
  \item If the \link\ is well colored and the fiber dotted, then the score is at least $[\frac{7}{2}] \ge [4] + [-\frac{1}{2}]$.
  \item If the \link\ is dotted, then the score is at least $[\frac{1}{2}]$, whether the fiber is dotted or not.
  \item If the \link\ is badly colored, the score is at least $[-\frac{1}{2}]$, whether the fiber is dotted or not.
  \end{enumerate}
\end{claim}
\proof
Let $v$ be the center of an isolated $3$-star $S_v$ with respect to a \green-\black-cycle $C$. Let $v_1$, $v_2$ and $v_3$ be the three neighbors of $v$ on the cycle $C$.

(a) If the \link\ is well colored and  the fiber is not dotted, then we can simply remove $v$ in all cases, and save from cutting the edges from $v$ to its neighbors. In this case, we may assume by symmetry that $v_1, v_3 \in \setA$ and $v_2 \in \setB$, yielding a score of $[5 \mid 6]$. This proves (a).

(b) Suppose the \link\ is well colored and the fiber is dotted. We may assume that $v_1$ is Amber, $v_2$ is blue, and $v_3$ is dotted Amber. We once again remove $v$. If the set \setA\ is chosen, we save one by cutting the edge $vv_2$, while if the set \setB\ is chosen, we save one by cutting the edge $vv_1$. However, when choosing the set \setB, we pay~$3$ for cutting the edge $vv_3$. This yields a score of $[5 \mid 2]$, with an average score of $[\frac{7}{2}]$, which proves (b).

(c) Suppose the \link\ is dotted, say $v_1$ is dotted Amber. Suppose firstly that at least one of $v_2$ and $v_3$ is not dotted blue. In this case, when choosing the dominating set \setA, we remove $v$ and get a score of at least~$4$ for removing $v$, and we pay~$3$ at most once (when at most one of $v_2$ and $v_3$ is dotted blue). When choosing the dominating set \setB, we mark~$v$ to obtain a score of~$0$. This yields a score of at least $[1 \mid 0]$, with an average score of $[\frac{1}{2}]$. Suppose next that both $v_2$ and $v_3$ are dotted blue. In this case, interchanging the roles of $v_1$ and $v_2$ we remove $v$ when choosing the dominating set \setA, and we mark~$v$ when choosing the dominating set \setA\ to obtain a score of at least $[0 \mid 1]$. In both cases, this yields an average score of $[\frac{1}{2}]$, whether the fiber is dotted or not. This proves (c).

(d) Suppose the \link\ is badly colored. We may assume that $v_1$ and $v_2$ both are Amber and are not dotted. Suppose firstly that $v_3$ is not dotted. In this case, when choosing the dominating set \setA, we remove $v$ and get a score of at least~$4$ for removing $v$. When choosing the dominating set \setB, we once again remove $v$ and get a score of~$4$ for removing $v$. If $v_3$ is Amber, then we save~$3$ for by cutting the three edges incident with~$v$ to obtain a final score of~$4 + 3 - 12 = -5$ for removing $v$, saving the three edges cut, but add $v$ to the dominating set. If $v_3$ is blue, then the vertex $v$ is not added to the dominating set to obtain a final score of~$4 + 2 = 6$. This yields a score of at least $[4 \mid -5]$, with an average score of $[-\frac{1}{2}]$.

Suppose next that $v_3$ is dotted. Suppose firstly that $v_3$ is dotted blue. In this case when choosing the dominating set \setA, we remove $v$ and pay $3$ for cutting the edge $vv_3$ to obtain a score of $4 - 3 = 1$. When choosing the dominating set \setB, we remove $v$ and save from cutting the two edges $vv_1$ and $vv_2$ to obtain a score of $4 + 2 = 6$. This yields a score of~$[1 \mid 6]$, with an average score of $[\frac{7}{2}]$. Suppose next that $v_3$ is dotted amber. In this case when choosing the dominating set \setA, we remove $v$ and obtain a score of $4$. When choosing the dominating set \setB, we do not remove~$v$ and save the two edges $vv_1$ and $vv_2$ cut. However, the degree of $v$ drops from~$3$ to~$1$, which costs~$4$. Thus we obtain a score of $2 - 4 = -2$. This yields a score of~$[4 \mid -2]$, with an average score of $[1]$. In summary, if the \link\ is badly colored, we obtain an average score of at least $[-\frac{1}{2}]$. This proves (d).~\smallqed

\medskip
We define next an isolated $2$-star with respect to a \green-\black-cycle.

\begin{definition}
\label{defn:isolated-2star}
{\rm
If the center vertex $u$ of a black $3$-star is adjacent to exactly two vertices on the \green-\black-cycle $C$ as illustrated in Figure~\ref{f:isolate-2star}, then we call the $3$-star an \emph{isolated} $2$-\emph{star} with respect to the cycle~$C$.
}
\end{definition}

\begin{figure}[htb]
\begin{center}
\input{isolated-2star.tex}
\vskip -0.35 cm
\caption{An isolated \twostar}
\label{f:isolate-2star}
\end{center}
\end{figure}

\begin{claim}\label{c:score-isolated-2star}
Let $C$ be an isolated \twostar\ with respect to the \green-\black-cycle $C$. If the neighbor of the center of the star not on the cycle $C$ is not adjacent to the center of any other \twostar\  associated with the cycle $C$, then the average score of the isolated \twostar\  is at least~$[0]$.
\end{claim}
\proof Consider an isolated \twostar\  with center $u$ that is adjacent to the vertices $u_1$ and $u_2$ on the \green-\black-cycle $C$, as illustrated in Figure~\ref{f:isolate-2star}, where the third neighbor, $u_3$ say, of $u$ does not belong to the cycle $C$ and is not adjacent to the center of any other \twostar\  associated with the cycle $C$. We may assume that the vertex $u_1$ is an amber vertex.

Suppose that the vertex $u_2$ is an amber vertex. When choosing the dominating set \setA, we remove $u$ and get a score of at least~$4$ for removing $u$. The degree of $u_3$ drops from~$3$ to~$2$ which costs~$1$ to yield a score of~$4 - 1 = 3$. When choosing the dominating set \setB, we do not delete the vertex $u$. If both $u_1$ and $u_2$ are dotted amber, then the degree of $u$ is unchanged, and this contributes~$0$ to the score. If exactly one of $u_1$ and $u_2$ are dotted amber, then the degree of $u$ drops by~$1$ which costs~$1$ but we regain the loss since we gain~$1$ for cutting the edge from~$u$ to the amber vertex that is not dotted, once again resulting in a score of~$1 - 1 = 0$. If neither $u_1$ nor $u_2$ is dotted, then the degree of $u$ decreases by~$2$, resulting in a weight loss of~$-4$. However, in our counting we have already attributed a cost of~$1$ to the cutting of each of the outgoing edges $uu_1$ and $uu_2$, which contributes~$2$ to the counting. Hence, the overall contribution in this case is~$2-4=-2$. Hence when the vertex $u_2$ is amber, we obtain a score of at least~$[3 \mid -2]$, with an average score of at least~$[\frac{1}{2}]$.

Suppose that the vertex $u_2$ is a blue vertex. Suppose that the vertex $u_2$ is not dotted blue. In this case when choosing the dominating set \setA, we remove $u$ and get a score of at least~$4$ for removing $u$. The degree of $u_3$ drops from~$3$ to~$2$ which costs~$1$. We have already counted a cost of~$1$ for the cutting of the outgoing edge $uu_2$. This yields a score of~$4-1+1=4$. Suppose that the vertex $u_2$ is dotted blue. In this case when choosing the dominating set \setA, we do not remove the vertex~$u$ but mark it, yielding a score of~$0$. A symmetrical argument (noting that the vertex $u_1$ is amber) yields a score of~$4$ or~$0$ when choosing the dominating set \setB\ in this case when the vertex $u_2$ is a blue vertex. Hence when the vertex $u_2$ is blue, we obtain a score of at least~$[0 \mid 0]$, with an average score of at least~$[0]$.~\smallqed

\medskip
We proceed further by formally defining a broad \link\ in the \greencyclegraph.

\begin{definition}
\label{defn:broad-link}
{\rm
Let $u$ and $v$ be the center vertices of isolated $2$-stars $S_u$ and $S_v$, respectively, with respect to the \green-\black-cycle $C$. Let $u_1$ and $u_2$ be the two neighbors of $u$ that belong to the cycle $C$, and let $v_1$ and $v_2$ be the two neighbors of $v$ that belong to the cycle $C$. If $u$ and $v$ are adjacent as illustrated in Figure~\ref{f:adjacent-2star}, then we say that the isolated $2$-stars $S_u$ and $S_v$ are adjacent. In this case, we introduce two broad \links\ in the \greencyclegraph, where we add a broad \link\ between $u_1$ and $u_2$ and a broad \link\ between $v_1$ and $v_2$.
}
\end{definition}

\begin{figure}[htb]
\begin{center}
\input{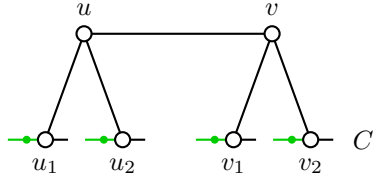}
\caption{Adjacent 2stars}\label{f:adjacent-2star}
\end{center}
\end{figure}

\begin{claim}\label{c:score-adjacent-2star}
The following properties hold for the two broad \links\ associated with adjacent  isolated $2$-stars in the \greencyclegraph. \\[-24pt]
\begin{enumerate}
   \item If a broad \link\ is well colored, then the average score of the \link\ is at least $[5]$.
  \item If a broad \link\ contains at least one dotted extremity, then the average score of the \link\ is at least $[\frac{1}{2}]$.
  \item If a broad \link\ is badly colored, then the average score of the \link\ is at least $[-1]$.
  \end{enumerate}
\end{claim}
\proof Consider two adjacent $2$-stars $S_u$ and $S_v$ with centers $u$ and $v$, where $u$ is adjacent to the vertices $u_1$ and $u_2$ on the \green-\black-cycle $C$ and $v$ is adjacent to the vertices $v_1$ and $v_2$ on $C$, as illustrated in Figure~\ref{f:adjacent-2star}. Let $W = \{u_1,u_2,v_1,v_2\}$. We proceed further with three subclaims.

\begin{subclaim}
\label{f:adjacent-2star-no-dotted}
If no vertex in $W$ is a dotted vertex, then the following properties hold for the average score of the two adjacent $2$-stars.
\\[-24pt]
\begin{enumerate}
   \item If both \links\ are badly colored, then the average score is at least $[-2] = [-1] + [-1]$.
  \item If one \link\ is badly colored and the other well colored, then the average score is at least $[4] = [-1] + [5]$.
  \item If both \links\ are well colored, then the average score is at least $[10] = [5] + [5]$.
\end{enumerate}
\end{subclaim}
\proof Suppose that no vertex in $W$ is a dotted vertex. Renaming vertices if necessary, we may assume that if $v_1$ and $v_2$ have the same color, then $u_1$ and $u_2$ have the same color. Further, we may assume that $u_1$ is an amber vertex.

(a) Suppose that both \links\ are badly colored, and thus that $u_1$ and $u_2$ are amber. Suppose firstly that both $v_1$ and $v_2$ are amber. If the set $A$ is chosen, then we delete both vertices $u$ and~$v$, which yields~$2 \times 4 = 8$ to the score. If the set $B$ is chosen, then we delete both vertices $u$ and~$v$, and we add the vertex $v$ to the dominating set. This yields $2 \times 4$ for the deletion of $u$ and $v$, $-12$ for adding the vertex $v$ to the dominating set, and $+4$ for the cost of cutting the four outgoing edges $uu_1$, $uu_2$, $vv_1$ and $vv_2$, which are already included in the count. The overall contribution if the set $B$ is chosen is therefore~$8 - 12 + 4 = 0$. This yields a score of $[8 \mid 0]$ in this case when both $v_1$ and $v_2$ are amber.
Suppose next that both $v_1$ and $v_2$ are blue. If the set $A$ is chosen, then we delete both vertices $u$ and~$v$, and we add the vertex $v$ to the dominating set. This yields $2 \times 4$ for the deletion of $u$ and $v$, $-12$ for adding the vertex $v$ to the dominating set, and $+2$ for the cost of cutting the two outgoing edges $vv_1$ and $vv_2$, which are already included in the count. The overall contribution if the set $A$ is chosen is therefore~$8 - 12 + 2 = -2$. If the set $B$ is chosen, then we delete both vertices $u$ and~$v$, and we add the vertex $u$ to the dominating set. By symmetry to the previous case, the overall contribution in this case is again~$-2$. The above yields a score of $[-2 \mid -2]$. This proves~(a).

(b) Suppose that one \link\ is badly colored and the other well colored. Without loss of generality, we may assume $u_1$, $u_2$ and $v_1$ are amber, and $v_2$ is blue. If the set $A$ is chosen, then we delete both vertices $u$ and~$v$, which yields~$2 \times 4 = 8$ to the count. Moreover, we gain~$+1$ for the outgoing edge~$vv_2$, which are already included in the count. The overall contribution if the set $A$ is chosen is therefore~$8 + 1 = 9$. If the set $B$ is chosen, then we delete both vertices $u$ and~$v$, and we add the vertex $u$ to the dominating set. This yields $2 \times 4$ for the deletion of $u$ and $v$, $-12$ for adding the vertex $u$ to the dominating set, and $+3$ for the cost of cutting the three outgoing edges $uu_1$, $uu_2$, and $vv_1$, which are already included in the count. The overall contribution if the set $B$ is chosen is therefore~$8 - 12 + 3 = -1$. The above yields a score of $[9 \mid -1]$. This proves~(b).

(c) Suppose that both \links\ are well colored. Without loss of generality, we may assume $u_1$ and $v_1$ are amber, and $u_2$ and $v_2$ are blue. If the set $A$ is chosen, then we delete both vertices $u$ and~$v$, yielding $2 \times 4 = 8$ to the count. Further, we count $+2$ for the cost of cutting the two outgoing edges $uu_2$ and $vv_2$, which are already included in the count. The overall contribution if the set $A$ is chosen is therefore~$8 + 2 = 10$. By symmetry, the overall contribution if the set $B$ is chosen is~$10$. The above yields a score of $[10 \mid 10]$. This proves~(c).~\smallqed

\begin{subclaim}
\label{f:adjacent-2star-both-links-dotted}
If both \links\ contain at least one dotted extremity, then the following properties hold for the average score of the two adjacent $2$-stars.
\\[-24pt]
\begin{enumerate}
\item If $W$ contains at least three vertices of the same color, then the average score is at least $[\frac{5}{2}]$.
\item If $W$ contains exactly two vertices of the same color and both $u$ and $v$ have a neighbor of each color, then the average score is at least $[2]$.
\item If $W$ contains exactly two vertices of the same color and both $u$ and $v$ have neighbors of different colors, then the average score is at least $[1]$.
\end{enumerate}
\end{subclaim}
\proof
Suppose both \links\ contain at least one dotted extremity. In particular, this implies that at least two vertices in $W$ are dotted.

(a) Suppose that $W$ contains at least three vertices of the same color. Suppose firstly that all four vertices of $W$ are colored the same, say with color Amber. If the set \setA\ is chosen, then we delete both $u$ and $v$, which yields $2 \times 4 = 8$ to the score. If the set \setB\ is chosen, then neither $u$ nor $v$ is deleted, which yields~$0$ to the score. Thus in this case, the score is at least~$[8 \mid 0]$. Suppose next that exactly three vertices in $W$ have the same color. We may assume that $u_1$, $u_2$ and $v_1$ are Amber. If the set \setA\ is chosen, then as before we delete both $u$ and $v$, which yields $8$ to the score. However in this case, if the blue vertex $v_2$ is dotted, then we pay~$3$ for cutting the edge $vv_2$, yielding an overall score of at least~$8 - 3 = 5$ when the set \setA\ is chosen. If the set \setB\ is chosen, then we do not delete $u$ and $v$ but we mark the vertex~$v$. We note that since at least one of the amber vertices $u_1$ and $u_2$ is dotted, the degree of the vertex~$u$ is at least~$2$. This yields a score of at least~$0$ when the set \setB\ is chosen. In summary, this yields a score of at least $[5 \mid 0]$. This proves~(a).

(b) Suppose that $W$ contains exactly two vertices of the same color and both $u$ and $v$ have a neighbor of each color. If the set \setA\ or the set \setB\ is chosen, then we delete both $u$ and $v$, which yields $2 \times 4 = 8$ to the score. When choosing the set \setA, if the blue neighbor of $u$ (resp., $v$) is dotted, then we pay~$3$ for cutting the edge from $u$ to that vertex. Analogously, when choosing the set \setB, if the amber neighbor of $u$ (resp., $v$) is dotted, then we pay~$3$ for cutting the edge from $u$ (resp., $v$) to that vertex. This yields an overall score of at least~$8 - 3 - 3 = 2$ when each set \setA\ and \setB\ is chosen. In summary, this yields a score of at least $[2 \mid 2]$. This proves~(b).

(c) Suppose that $W$ contains exactly two vertices of the same color and both $u$ and $v$ have neighbors of different colors. We may assume that both neighbors of $u$ are Amber, and therefore both neighbors of $v$ are Blue. If the set \setA\ is chosen, then we delete the vertex $u$ but not the vertex~$v$, and we cut the edge $uv$. The deletion of~$u$ contributes~$4$ to the score. If both $v_1$ and $v_2$ are dotted, then the degree of~$v$ drops from~$3$ to~$2$, and so cutting the edge $uv$ costs~$1$. If exactly one of $v_1$ and $v_2$ is dotted, then the degree of~$v$ drops from~$3$ to~$1$, which costs~$4$, but we gain~$1$ from deleting the outgoing edge from $v$ to its blue neighbor that is not dotted, and so cutting the edge $uv$ costs~$3$. Hence, the score is at least~$4 - 3 = 1$ if the set \setA\ is chosen. Identical arguments show that the score is at least~$1$ if the set \setB\ is chosen. In summary, this yields a score of at least $[1 \mid 1]$. This proves~(c).~\smallqed

\begin{subclaim}
\label{f:adjacent-2star-one-link-dotted}
If only one \link\ contains at least one dotted extremity, then the following properties hold for the average score of the two adjacent $2$-stars.
\\[-24pt]
\begin{enumerate}
\item If only one extremity of the dotted \link\ is dotted and the other \link\ is well colored, then the average score is at least $[\frac{11}{2}] = [\frac{1}{2}] + [5]$.
\item If both extremities of the dotted \link\ are dotted and the other \link\ is well colored, then the average score is at least $[6] = [1] +  [5]$.
\item If only one extremity of the dotted \link\ is dotted and the other \link\ is badly colored, then the average score is at least $[-\frac{1}{2}] = [\frac{1}{2}] + [-1]$.
\item If both extremities of the dotted \link\ are dotted and the other \link\ is badly colored, then the average score is at least $[\frac{1}{2}] = [\frac{3}{2}] + [-1]$.
\end{enumerate}
\end{subclaim}
\proof
(a) Suppose that only one extremity of the dotted \link\ is dotted and the other \link\ is well colored. We may assume that $u_1$ is dotted amber, $v_1$ is amber (and not dotted) and $v_2$ is blue (and not dotted). Suppose firstly that $u_2$ is colored Amber. By supposition, $u_2$ is not dotted. If the set \setA\ is chosen, then we delete both $u$ and $v$, which contributes $2 \times 4 = 8$ to the score. Since we gain~$1$ by cutting the outgoing edge from $v$ to the blue vertex $v_2$, this yields a score of at least~$9$. If the set \setB\ is chosen, then we delete the vertex~$v$, which contributes~$4$ to the score, but we do not delete the vertex~$u$. We gain~$1$ by cutting the outgoing edges from $u$ to the amber vertex $u_2$ and from $v$ to the amber vertex $v_1$. However, the degree of~$u$ drops from~$3$ to~$1$, which costs~$4$. Hence, the score when the set \setB\ is chosen is at least~$4 + 2 - 4 = 2$. Thus, the overall score is at least~$[9 \mid 2]$ if $u_2$ is colored Amber.

Suppose secondly that $u_2$ is colored blue. If the set \setA\ is chosen, then we delete both $u$ and $v$, which contributes $8$ to the score. Since we gain~$1$ by cutting the outgoing edges from $u$ to the blue vertex $u_2$ and from $v$ to the blue vertex $v_2$, this yields a score of at least~$10$. If the set \setB\ is chosen, then we delete both $u$ and $v$, which contributes $8$ to the score. We gain~$1$ by cutting the outgoing edge from $v$ to the amber vertex $v_1$. However, cutting the edge from $u$ to the dotted amber vertex $u_1$ costs~$3$. Hence, the score when the set~\setB\ is chosen is at least~$8 + 1 - 3 = 6$. Thus, the overall score is at least~$[10 \mid 6]$ if $u_2$ is colored blue. This proves~(a).

(b) Suppose that both extremities of the dotted \link\ are dotted and the other \link\ is well colored. We may assume that $u_1$ and $u_2$ are dotted. Further, we may assume that $u_1$ is dotted amber, $v_1$ is amber (and not dotted) and $v_2$ is blue (and not dotted). Suppose firstly that the dotted vertex $u_2$ is colored amber. If the set \setA\ is chosen, then as before this yields a score of at least~$9$. If the set \setB\ is chosen, then we delete the vertex~$v$, which contributes~$4$ to the score, but we do not delete the vertex~$u$. We gain~$1$ by cutting the outgoing edges from $v$ to the amber vertex $v_1$. However the degree of~$u$ drops from~$3$ to~$2$, which costs~$1$. Hence, the score when the set \setB\ is chosen is at least~$4 + 1 - 1 = 4$. Thus, the overall score is at least~$[9 \mid 4]$ if the dotted vertex $u_2$ is colored amber.

Suppose secondly that $u_2$ is colored blue. If the set \setA\ is chosen, then we delete both $u$ and $v$, which contributes $8$ to the score. We gain~$1$ by cutting the outgoing edges from $v$ to the blue vertex $v_2$. However, cutting the edge from $u$ to the dotted blue vertex $u_2$ costs~$3$. Hence, the score when the set~\setA\ is chosen is at least~$8 + 1 - 3 = 6$. If the set \setB\ is chosen, then we delete both $u$ and $v$, which contributes $8$ to the score. We gain~$1$ by cutting the outgoing edge from $v$ to the amber vertex $v_1$. However, cutting the edge from $u$ to the dotted amber vertex $u_1$ costs~$3$. Hence, the score when the set~\setB\ is chosen is at least~$8 + 1 - 3 = 6$. Thus, the overall score is at least~$[6 \mid 6]$ if the dotted vertex $u_2$ is colored blue. This proves~(b).

(c) Suppose that only one extremity of the dotted \link\ is dotted and the other \link\ is badly colored. We may assume that $u_1$ and $u_2$ are colored amber and that $v_1$ is dotted. By supposition, $v_2$ is not dotted.

Suppose firstly that $v_1$ or $v_2$ is amber. If the set \setA\ is chosen, then we delete both $u$ and $v$, which contributes $8$ to the score. If the dotted vertex $v_1$ is blue, then cutting the edge $vv_1$ costs~$3$. If the dotted vertex $v_1$ is amber, then cutting the outgoing edges $vv_1$ and $vv_2$ contributes at least~$0$ to the score. Thus, the score when the set \setA\ is chosen is at least~$8 - 3 = 5$. If the set \setB\ is chosen, then we delete both $u$ and $v$, and we add the vertex $v$ to the dominating set. This yields $2 \times 4$ for the deletion of $u$ and $v$, $-12$ for adding the vertex $v$ to the dominating set, and $+2$ for the cost of cutting the outgoing edges $uu_1$ and $uu_2$, which are already included in the count. If the dotted vertex $v_1$ is amber, then we mark the vertex $v_1$ (which is dominated by the deleted vertex $v$) which yields an additional score of~$+1$. If the dotted vertex $v_1$ is blue, then the vertex $v_2$ is amber and we gain $+1$ for the cost of cutting the outgoing edge $vv_2$. The overall contribution if the set $B$ is chosen is therefore at least~$8 - 12 + 2 + 1 = -1$. Thus, the overall score is at least~$[5 \mid -1]$ if $v_1$ or $v_2$ is amber.

Suppose next that both $v_1$ and $v_2$ are blue. Recall that $v_1$ is dotted and $v_2$ is not dotted. If the set \setA\ is chosen, then we delete $u$, which contributes~$4$ to the score, but do not delete~$v$. We gain~$1$ for the deletion of the outgoing edge~$vv_2$, but the cost of cutting the edges $uv$ and $vv_1$ is~$4$ since the degree of $v$ decreases from~$3$ to~$1$. Thus, the score when the set \setA\ is chosen is at least~$4 + 1 - 4 = 1$. If the set \setB\ is chosen, then we delete both $u$ and $v$, which contributes $8$ to the score, and we add either $u$ or $v$ to the dominating set, which costs~$-12$. We gain~$2$ by cutting the outgoing edge from $u$ to the amber vertices $u_1$ and $u_2$. Hence, the score when the set~\setB\ is chosen is at least~$8 - 12 + 2 = -2$. Thus, the overall score is at least~$[1 \mid -2]$ if both $v_1$ and $v_2$ are blue. This proves~(c).

(d) Suppose that both extremities of the dotted \link\ are dotted and the other \link\ is badly colored. We may assume that $u_1$ and $u_2$ are colored amber, and that both $v_1$ and $v_2$ are dotted. Suppose $v_1$ or $v_2$ is amber. Proceeding exactly as in the part~(c) above, the overall score in this case is at least~$[5 \mid -1]$. Suppose that both $v_1$ and $v_2$ are (dotted) blue. If the set \setA\ is chosen, then we delete $u$, which contributes~$4$ to the score, but do not delete~$v$. The cutting of the edge $uv$ only costs~$1$ since the degree of $v$ decreases from~$3$ to~$2$ (noting that the edges from $v$ to the dotted blue vertices $v_1$ and $v_2$ are not cut). Thus, the score when the set \setA\ is chosen is at least~$4 - 1 = 3$. If the set \setB\ is chosen, then we proceed exactly in part~(c) above and delete both $u$ and $v$ and add either $u$ or $v$ to the dominating set, to yield a score of at least~$-2$. Thus, the overall score is at least~$[3 \mid -2]$ if both $v_1$ and $v_2$ are blue. This proves~(d).~\smallqed

\medskip
The result of Claim~\ref{c:score-adjacent-2star} now follows from Claims~\ref{f:adjacent-2star-no-dotted},~\ref{f:adjacent-2star-both-links-dotted}, and~\ref{f:adjacent-2star-one-link-dotted}.~\smallqed

Note that until now, every vertex adjacent to a vertex not on the \green-\black\ cycle that gets removed in the graph is marked. This is not true for the two following claims, and an extra care is taken to avoid creating isolated vertices.

\begin{claim}
\label{score-common-nbr-2stars}
Consider two isolated $2$-stars with centers $u$ and $v$, respectively, with respect to the \green-\black-cycle $C$, where $u$ is adjacent to the vertices $u_1$ and $u_2$ on $C$ and $v$ is adjacent to the vertices $v_1$ and $v_2$ on $C$, and where $u$ and $v$ have a common neighbor~$w$ that has no neighbor on the cycle $C$. The structure of the resulting configuration is illustrated in Figure~\ref{f:common-nbr-two-2star}. Then the average score is at least~$[0]$.
\end{claim}

\begin{figure}[htb]
\begin{center}
\input{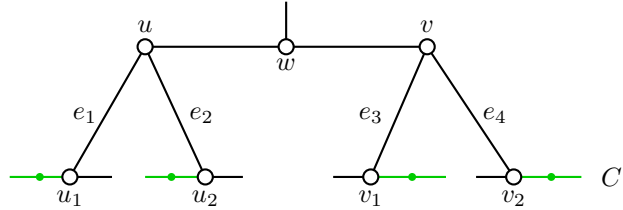}
\vskip -0.75 cm
\caption{A common neighbor of two 2stars}\label{f:common-nbr-two-2star}
\end{center}
\end{figure}

\proof
Consider the above configuration and denote $W = \{u_1,u_2,v_1,v_2\}$.
Let $x$ be the third neighbor of the vertex $w$ different from~$u$ and~$v$.
We proceed further with a series of six subclaims that analyse all possible cases that can occur.

\begin{subclaim}
\label{score-common-nbr-2stars.1}
If no vertex in $W$ is dotted, then the average score is at least $[2]$.
\end{subclaim}
\proof Suppose that no vertex in $W$ is dotted. In this case, we remove the vertices $u$, $v$ and $w$, we add the vertex $w$ to the set \setA\ (resp., \setB) when the set \setA\ (resp., \setB) is chosen. Further, we mark the vertex~$x$. The deletion of the three vertices contributes~$12$ to the score, and the addition of the vertex $w$ to the dominating set contributes~$-12$ to the score. If $a$ and $b$ are the number of vertices in $W$ that are amber and blue, respectively, then the overall contribution to the cost is~$b$ if the set \setA\ is chosen, and~$a$ if the set \setB\ is chosen. Thus, the overall score of the configuration is at least~$[b \mid a]$. We note that $a,b \ge 0$ and $a + b = 4$. Hence, the average score of the configuration is at least~$[2]$.~\smallqed

\begin{subclaim}
\label{score-common-nbr-2stars.2}
If at least one vertex in $W$ is dotted and all vertices in $W$ have the same color, then the average score is at least $[1]$.
\end{subclaim}
\proof  Suppose that at least one vertex in $W$ is dotted and all vertices in $W$ have the same color.  By symmetry and renaming vertices if necessary, we may assume that all vertices in $W$ are amber and that the vertex $u_1$ is dotted. If the set \setA\ is chosen, then we delete the vertices $u$ and $v$.
The deletion of the two vertices contributes $2 \times 4 = 8$ to the score, and since $w$ does not become an isolate vertex, cutting the edges $uw$ and $vw$ contributes $-4$ to the score. Hence, the score when the set \setA\ is chosen is at least~$8 - 4 = 4$.

Suppose the set \setB\ is chosen. If at least one of $v_1$ or $v_2$ is dotted, then we do not remove any vertex, yielding a contribution to the score of at least~$0$. If neither $v_1$ nor $v_2$ is dotted and $u_2$ is dotted, then we delete the vertices $v$ and $w$, add the vertex $w$ to the dominating set \setB, and mark the vertices $u$ and $x$, yielding a contribution to the cost of at least~$2 \times 4 - 12 + 2 = -2$ noting that we gain~$2$ from deleting the outgoing edges from $v$ to $v_1$ and $v_2$. If $u_1$ is the only vertex in $W$ that is dotted, then we delete the vertices $u$, $v$ and $w$,  we add the vertex $w$ to the dominating set \setB, and we mark the vertex~$x$, yielding a contribution to the score of at least~$3 \times 4 - 12 + 3 - 3 = 0$, noting that the dotted vertex $u_1$ decreases the score by~$3$ and the three outgoing edges $uu_2$, $vv_1$ and $vv_2$ each contribute~$1$ to the score. Hence, the score when the set \setB\ is chosen is at least~$-2$. Thus, the overall score of the configuration is at least~$[4 \mid -2]$, and so the average score of the configuration is at least~$[1]$.~\smallqed

\begin{subclaim}
\label{score-common-nbr-2stars.3}
If at least one vertex in $W$ is dotted and the vertices in $\{u_1,u_2\}$ are colored the same, and the vertices in $\{v_1,v_2\}$ are colored the same but with a different color, then the average score is at least $[\frac{5}{2}]$.
\end{subclaim}
\proof  By symmetry and renaming vertices if necessary, suppose that the vertices $u_1$ and $u_2$ are amber, where $u_1$ is dotted, and the vertices $v_1$ and $v_2$ are blue. Suppose the set \setA\ is chosen. If at least one of $v_1$ and $v_2$ is dotted, then we remove the vertex $u$, but remove neither $v$ nor $w$. The deletion of $u$ contributes~$4$ to the score, while the cutting of the edge $uw$ decreases the degree of the vertex $w$ by~$1$, which contributes~$-1$ to the score, yielding an overall contribution to the score of~$4 - 1 = 3$ in this case. If neither $v_1$ nor $v_2$ is dotted, then we remove the vertices $u$, $v$ and $w$, we add the vertex $w$ to the dominating set $A$, and we mark the vertex~$x$. This yields an overall contribution to the score of~$3 \times 4 - 12 + 2 = 2$, noting that the two outgoing edges $vv_1$ and $vv_2$ each contribute~$1$ to the score. Hence, the score when the set \setA\ is chosen is at least~$2$. If the set \setB\ is chosen, then we remove the vertex $v$, but remove neither $v$ nor $w$. The deletion of $v$ contributes~$4$ to the score, while the cutting of the edge $vw$ decreases the degree of the vertex $w$ by~$1$, which contributes~$-1$ to the score, yielding an overall contribution to the score of~$4 - 1 = 3$ when the set \setB\ is chosen. Thus, the overall score of the configuration is at least~$[2 \mid 3]$, and so the average score of the configuration is at least~$[\frac{5}{2}]$.~\smallqed

\begin{subclaim}
\label{score-common-nbr-2stars.4}
If exactly three vertices in $W$ have the same color, at least one of which is dotted, then the average score is at least $[\frac{1}{2}]$.
\end{subclaim}
\proof  Suppose that three vertices in $W$ are amber, say $u_1$, $u_2$ and $v_1$. Suppose firstly  that $u_1$ is dotted. If the set \setA\ is chosen, we remove the vertices $u$ and $v$. The deleting of the two vertices $u$ and $v$ contributes~$2 \times 4 = 8$ to the score, and since $w$ does not become an isolated vertex, cutting the edges $uw$ and $vw$ contributes $-4$ to the score. If the vertex $v_2$ is not a dotted (blue) vertex, then its contribution to the score is~$1$, while if the vertex $v_2$ is a dotted (blue) vertex, then its contribution to the score is~$-3$. Hence, the score when the set \setA\ is chosen is at least~$8 - 4 - 3 = 1$. If the set \setB\ is chosen, then we do not remove any of the vertices $u$, $v$ or $w$, but we mark the vertex~$v$, yielding a score of at least~$0$. Thus, the overall score of the configuration is at least~$[1 \mid 0]$, and so the average score of the configuration is at least~$[\frac{1}{2}]$.

Suppose now that neither $u_1$ nor $u_2$ is dotted, but $v_1$ is dotted.
If the set \setA\ is chosen, then we delete the vertex~$u$ and mark the vertex~$v$. The deleting of the vertex $u$ contributes~$4$ to the score. The cutting of the edge $uw$ contributes~$-1$ to the score, and the marked vertex $v$ contributes~$0$ to the score. Hence, the score when the set \setA\ is chosen is at least~$4 - 1 = 3$. If the set \setB\ is chosen, then we remove the vertices $u$, $v$ and $w$, we add the vertex $w$ to the dominating set \setB, and we mark the vertex~$x$. The deleting of the three vertices contributes~$12$ to the score, and the addition of the vertex $w$ to the dominating set contributes~$-12$ to the score. The cutting of the outgoing edge from $v$ to the dotted vertex $v_1$ contributes~$-3$ to the score, and the cutting of the two outgoing edges from $u$ to the amber vertices $u_1$ and $u_2$ contributes~$2$ to the score. Hence, the score when the set \setB\ is chosen is at least~$12 - 12 - 3 + 2 = -1$, yielding an overall score of at least~$[3 \mid -1]$, and so the average score of the configuration is at least~$[\frac{1}{2}]$.~\smallqed

\begin{subclaim}
\label{score-common-nbr-2stars.5}
If exactly three vertices in $W$ have the same color, none of which is dotted, but the fourth vertex is of a different color and dotted, then the average score is at least~$[0]$.
\end{subclaim}
\proof   Suppose that $u_1$ is dotted amber and all of $u_2$, $v_1$ and $v_2$ are blue and not dotted.
If the set \setA\ is chosen, then we delete the vertices $u$, $v$ and $w$, we add the vertex $w$ to the dominating set \setA, and we mark the vertex~$x$. The deleting of the three vertices contributes~$12$ to the score, and the addition of the vertex $w$ to the dominating set contributes~$-12$ to the score. The cutting of the outgoing edge from $u$ to the vertex $u_2$ contributes~$1$ to the score if $u_2$ is not dotted (blue) and contributes~$-3$ to the score if $u_2$ is dotted (blue). The cutting of the two outgoing edges from $v$ to the blue vertices $v_1$ and $v_2$ contributes~$2$ to the score. Hence, the score when the set \setB\ is chosen is at least~$12 - 12 - 3 + 2 = -1$. If the set \setB\ is chosen, then we remove the vertices $u$ and $v$, which contributes~$8$ to the score. The cutting of the edge $uw$ and $vw$ contributes~$-4$ to the score since the degree of the vertex~$w$ decreases from~$3$ to~$1$. The cutting of the outgoing edge from $u$ to the dotted (amber) vertex $u_1$ contributes~$-3$ to the score. Hence, the score when the set \setB\ is chosen is at least~$8 - 4 - 3 = 1$, yielding an overall score of at least~$[-1 \mid 1]$, and so the average score of the configuration is at least~$[0]$.~\smallqed

\begin{subclaim}
\label{score-common-nbr-2stars.6}
If each of $u$ and $v$ have neighbors of both colors, and at least one vertex is dotted, then the average score is at least~$[0]$.
\end{subclaim}
\proof  Suppose that both $u$ and $v$ have neighbors of different colors. Then in both cases when the set \setA\ is chosen and when the set \setB\ is chosen, we do not remove any vertices, but we mark both $u$ and $v$, yielding a score of at least~$[0 \mid 0]$.~\smallqed

\smallskip
The desired result in the statement of Claim~\ref{score-common-nbr-2stars} now follows from Claims~\ref{score-common-nbr-2stars.1} to~\ref{score-common-nbr-2stars.6}.~\smallqed

\smallskip
We define next an isolated $1$-star with respect to a \green-\black-cycle.

\begin{definition}
\label{defn:isolated-1star}
{\rm
If the center vertex of a black $3$-star is adjacent to exactly one vertex on the \green-\black-cycle $C$, then we call the $3$-star an \emph{isolated} $1$-\emph{star} with respect to the cycle~$C$.
}
\end{definition}

We consider next the cases where the center $u$ of an isolated $2$-star is adjacent to an isolated $1$-star centered at $v$ with respect to the \green-\black-cycle $C$, as shown in Figure~\ref{f:adjacent-2star-1star}. Let $w$ be the neighbor of $v$ not on the cycle $C$ and different from the vertex~$u$. If the vertex $w$ belongs to the \green-\black-cycle $C$, then this is covered by Claim~\ref{c:score-adjacent-2star}. Hence, we may assume in what follows that the vertex $w$ does not belong to $C$. We consider two cases, depending on the status of the third neighbor $w$ of $v$. We first consider the case when the vertex $w$ is the center $u$ of an isolated $2$-star, that is, $w$ is adjacent to two vertices of the \green-\black\ cycle $C$.

\begin{figure}[htb]
\begin{center}
\input{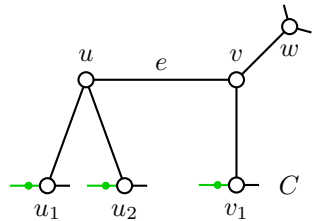}
\caption{A $2$-star adjacent to a $1$-star}
\label{f:adjacent-2star-1star}
\end{center}
\end{figure}

\begin{definition}
\label{defn:2star-1star-2star-fiber}
{\rm
Let $u$ and $w$ be the center vertices of isolated $2$-stars $S_u$ and $S_w$, respectively, with respect to the \green-\black-cycle $C$, and let $v$ be the center vertex of an isolated $1$-star $S_v$ with respect to $C$. Let $u_1$ and $u_2$ be the two neighbors of $u$ that belong to the cycle $C$, and let $w_1$ and $w_2$ be the two neighbors of $w$ that belong to the cycle $C$. If $S_v$ is adjacent to both $S_u$ and $S_w$ as illustrated in Figure~\ref{f:2star-1star-2star}, 
then we place a fiber on each of the vertices $u_1,u_2,w_1,w_2$ in the \greencyclegraph. We refer to the subgraph induced by the vertices of $S_u$, $S_v$ and $S_w$ as a $2$-star-$1$-star-$2$-star configuration.
}
\end{definition}

\begin{figure}[htb]
\begin{center}
\input{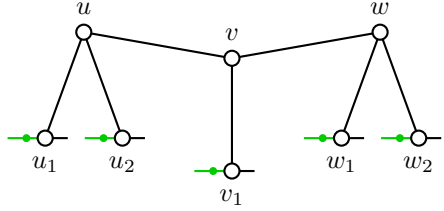}
\caption{A $2$-star-$1$-star-$2$-star configuration}
\label{f:2star-1star-2star}
\end{center}
\end{figure}

\begin{claim}\label{c:2star-1star-2star}
In a $2$-star-$1$-star-$2$-star configuration, the average score is at least $[0]$, unless two fibers associated with the configuration are dotted in which case the average score is at least $[-\frac{1}{2}]$.
\end{claim}
\proof
Consider a $2$-star-$1$-star-$2$-star configuration as defined in Definition~\ref{defn:2star-1star-2star-fiber}. Thus, $u$ and $w$ are the centers of isolated $2$-stars with respect to the \green-\black-cycle $C$, and $v$ is the center of an isolated $1$-star with respect to $C$. In our $2$-star-$1$-star-$2$-star configuration, the vertex $v$ is adjacent to both $u$ and $w$. Moreover, $u_1$, $u_2$, $v_1$, $w_1$ and $w_2$ are the neighbors of $u$, $v$ and $w$ on the \green-\black-cycle $C$, as illustrated in Figure~\ref{f:2star-1star-2star}. Let $W$ denote the set $\{u_1, u_2, w_1, w_2\}$.

Suppose firstly that no vertex in $W$ is dotted. In this case, we delete all three vertices $u$, $v$ and $w$ when each of the sets \setA\ and \setB\ is chosen, and we add the vertex~$v$ to the dominating set. Further, we mark the vertex~$v_1$. The deletion of the three vertices $u$, $v$ and $w$ contributes $12$ to the score, while the addition of the vertex $v$ to the dominating set contributes~$-12$ to the score. We save~$1$ from cutting each of the outgoing edges from $u$ and $w$ to vertices of $W$ of the opposite color to the chosen set. Further, if $v_1$ is not dotted, then we save~$1$ from cutting the outgoing edge $vv_1$ in the case when $v_1$ has the opposite color to the chosen set, while if $v_1$ is dotted, then we save~$1$ from marking the vertex. Thus the score of the configuration is of the form $[ a \mid b]$, where $a+b = 5$, implying that the average score is at least $[\frac{5}{2}]$.

Suppose that only one vertex in $W$ is dotted. By symmetry, we may assume that $u_1$ is dotted Amber. We proceed as in the previous case. That is, we delete all three vertices $u$, $v$ and $w$ when each of the sets \setA\ and \setB\ is chosen, we add the vertex~$v$ to the dominating set, and we mark the vertex~$v_1$. The contribution to the count is as before, except that by cutting the outgoing edge from $u$ to the dotted amber vertex $u_1$ when the set \setB\ is chosen, there is a cost of~$3$, implying that in this case the score of the configuration is of the form $[ a \mid b]$, where $a + b = 1$, yielding an average score of at least $[\frac{1}{2}]$.

Suppose that exactly two vertices in $W$ are dotted, and these two vertices belong to the same isolated $2$-star. By symmetry, we may assume that $u_1$ and $u_2$ are dotted. Suppose that $u_1$ and $u_2$ are both amber. Let $b$ be the number of blue vertices among $v_1$, $w_1$ and $w_2$. If the set \setA\ is chosen, then as before we delete all three vertices $u$, $v$ and $w$, we add the vertex~$v$ to the dominating set, and we mark the vertex~$v_1$. In this case, we gain an additional~$b$ by deleting the outgoing edge from $v$ to $v_1$ and from $w$ to $w_1$ and $w_2$ for the blue vertices among $v_1$, $w_1$ and $w_2$. This yields a score of $12 - 12 + b$ when the set \setA\ is chosen. If the set \setB\ is chosen, we delete the vertices $v$ and $w$, but we do not delete the vertex~$u$. Further, we add the vertex~$v$ to the dominating set, and we mark the vertices~$u$ and $v_1$. In this case, we gain~$3-b$ for the amber vertices among $v_1$, $w_1$ and $w_2$. This yields a score of $8 - 12 + 3 - b = -1 - b$ when the set \setB\ is chosen. Thus, the score of the configuration is at least $[b \mid -1 - b]$, yielding an average score of at least $[-\frac{1}{2}]$.

Suppose finally that at least one of $u_1$ and $u_2$ is dotted, and at least one of $w_1$ and $w_2$ is dotted. In this case, we do not delete any of the three vertices $u$, $v$ and $w$ when each of the sets \setA\ and \setB\ is chosen. However, whenever such a vertex is adjacent to a vertex of the opposite color to the chosen set, we mark the vertex. Suppose, for example, that the set \setA\ is chosen. If one of $u_1$ and $u_2$ is amber, then the vertex $u$ is marked. If both $u_1$ and $u_2$ are blue vertices, then the vertex $u$ is not marked. Further if both $u_1$ and $u_2$ are dotted (blue), then the degree of $u$ remains unchanged. If exactly one of $u_1$ and $u_2$ is dotted (blue), then the degree of $u$ drops from~$3$ to~$2$, which costs~$1$, but we gain~$1$ by deleting the outgoing edge from $u$ to its (blue) neighbor on the cycle $C$ that is not dotted. Since at least one of $u_1$ and $u_2$ is dotted, we conclude that in all cases, the contribution of $u$ to the count is~$0$. More generally, we note that if one of the vertices $u$, $v$ and $w$ is not marked when a particular set is chosen, say \setA, then such a vertex has degree at least~$2$ and its contribution to the count is~$0$. Thus, the score of the configuration in this case is at least $[0 \mid 0]$, yielding an average score of at least $[0]$.~\smallqed

\medskip
We proceed further with the following definition.

\begin{definition}
\label{defn:2star-1star-not-2star}
{\rm
Let $u$ be the center vertex of an isolated $2$-star $S_u$ with respect to the \green-\black-cycle $C$, and let $v$ be the center vertex of an isolated $1$-star $S_v$ with respect to $C$. Let $S_u$ and $S_v$ be adjacent, and so, $uv$ is an edge. Let $w$ be the third neighbor of $v$ that does not belong to the cycle $C$ and is different from $u$. If $w$ does not belong to the cycle $C$ and $w$ is not the center of an isolated $2$-star with respect to the cycle $C$, then we place a narrow \link\ between $u_1$ and $u_2$ in the \greencyclegraph. In this case, we refer to the subgraph induced by the vertices of $S_u$ and $S_v$ as a $2$-star-$1$-star-not-$2$-star configuration.
}
\end{definition}

\begin{claim}\label{c:2star-1star-any}
The following properties hold for the average score of a $2$-star-$1$-star-not-$2$-star configuration in the \greencyclegraph. \\[-24pt]
\begin{enumerate}
\item If the \link\ is badly colored, then the average score is at least $[-\frac{1}{2}]$.
\item If the \link\ contains at least one dotted vertex, then the average score is at least $[0]$.
\item If the \link\ is well colored, then the average score is at least $[4]$.
\end{enumerate}
\end{claim}
\proof
Consider a $2$-star-$1$-star-not-$2$-star configuration as defined in Definition~\ref{defn:2star-1star-not-2star}. Thus, $u$ is the center vertex of an isolated $2$-star with respect to the \green-\black-cycle $C$, and $v$ is the center vertex of an isolated $1$-star with respect to $C$, and $u$ and $v$ are adjacent. Let $w$ be defined as in Definition~\ref{defn:2star-1star-not-2star}. We note that in our $2$-star-$1$-star-not-$2$-star configuration, the vertex $v$ is adjacent to both $u$ and $w$. Let $u_1$ and $u_2$ be the neighbors of $u$ on the \green-\black-cycle $C$, and let $v_1$ be the neighbor of $v$ on the cycle $C$. The structure of the configuration is illustrated in Figure~\ref{f:adjacent-2star-1star}.

(a) Suppose firstly that the \link\ in the configuration (that we added between $u_1$ and $u_2$) is badly colored. We may assume that both $u_1$ and $u_2$ are amber (and not dotted). If the set \setA\ is chosen, then we delete $u$ but not $v$, and we mark $v$ if $v_1$ is amber. The deletion of $u$ contributes $4$ to the score. If $v_1$ is dotted blue, then the degree of $v$ drops from~$3$ to~$2$, and the contribution of cutting the edge $uv$ is~$-1$ in this case. If $v_1$ is blue and not dotted, then the degree of $v$ drops from~$3$ to~$1$, and the contribution of cutting the edge $uv$ is~$-4 + 1 = -3$ in this case noting that we gain~$1$ by cutting the outgoing edge from $v$ to the vertex $v_1$. Hence, the score when the set \setA\ is chosen is at least~$4 - 3 = 1$. If the set \setB\ is chosen, then we delete neither $u$ nor $v$, and we mark $v$. We gain~$2$ from cutting the two outgoing edges from $u$ to the amber vertices $u_1$ and $u_2$. However, the degree of $v$ drops from~$3$ to~$1$, which costs~$4$. Hence, the score when the set \setB\ is chosen is at least~$2 - 4 = -2$. Thus, the overall score of the configuration is at least~$[1 \mid -2]$ if the \link\ is badly colored, which yields an average score of at least $[-\frac{1}{2}]$.

(b) Suppose that the \link\ contains at least one dotted vertex. We may assume that $u_1$ is dotted amber. If the set \setA\ is chosen, then we delete neither $u$ nor $v$, and we mark $u$, yielding a score of at least~$0$. If the set \setB\ is chosen, we do not delete $u$ or $v$. Since $u_1$ is dotted and the edge $uu_1$ is not cut, we note that the degree of the vertex $u$ is at least~$2$. Thus, the overall score of the configuration is at least~$[0 \mid 0]$ if the \link\ contains at least one dotted vertex, which yields an average score of at least $[0]$.

(c) Suppose that the \link\ is well colored. We may assume that $u_1$ is blue and $u_2$ is amber (and neither $u_1$ nor $u_2$ is dotted). Further by symmetry on the roles of the colors, we may further assume that the vertex $v_1$ is blue. If the set \setA\ is chosen, then we delete $u$ but not $v$. The deletion of $u$ contributes~$4$ to the score, and we gain~$1$ by cutting the outgoing edge from $u$ to the blue vertex $u_1$. If $v_1$ is dotted blue, then the degree of~$v$ drops from~$3$ to~$2$, and the contribution of cutting the edge $uv$ is~$-1$. If $v_1$ is blue and not dotted, then the degree of $v$ drops from~$3$ to~$1$, and the contribution of cutting the edge $uv$ is~$-4 + 1 = -3$ noting that in this case we gain~$1$ by cutting the outgoing edge from $v$ to the vertex $v_1$. Hence, the score when the set \setA\ is chosen is at least~$4$ if the blue vertex $v_1$ is dotted and at least~$2$ if the blue vertex $v_1$ is not dotted.

If the set \setB\ is chosen, then we delete both $u$ and $v$, which contributes~$8$ to the score. Moreover, we gain~$1$ by cutting the outgoing edge from $u$ to the amber vertex $u_2$. If cutting the edge $vw$ results in the vertex $w$ having degree at least~$1$, then the contribution of cutting the edge $vw$ is at most~$-3$, yielding a score when the set \setB\ is chosen of at least~$8 + 1 - 3 = 6$. Combined with the score when the set \setA\ is chosen, the overall score of the configuration is at least~$[2 \mid 6]$ in this case, which yields an average score of at least $[4]$, as desired. Hence we may assume that the vertex $w$ becomes isolated upon the deletion of $u$ and $v$.

Suppose that $w$ is a vertex of degree~$2$ in $G$. In that case, its second neighbor $v'$ is not on the \green-\black-cycle $C$ by supposition. Since by assumption the vertex $w$ is isolated when the set \setB\ is chosen, the edge $wv'$ was cut and we infer that $v'$ is the center of an isolated $1$-star that is adjacent to the center $u'$ of an isolated $2$-star with respect to the \green-\black-cycle $C$. Further, the neighbors of $u'$ on the cycle $C$ are of different colors and the neighbor $v'_1$ of $v'$ on $C$ is colored blue, that is, a mirror situation occurs with $u'$ and $v'$ as with $u$ and $v$. Thus, if $v_1'$ is the neighbor of $v'$ on the cycle $C$, and if $u_1'$ and $u_2'$ are the neighbors of $u'$ on $C$, then the \link\ between $u_1'$ and $u_2'$ is well colored. Further, we may assume $u_1'$ is blue and $u_2'$ is amber (and neither $u_1'$ nor $u_2'$ is dotted), while $v_1'$ is blue, as illustrated in Figure~\ref{f:2star-1star-any-1star-2star} where in this case the vertex $w$ has degree~$2$ in $G$.

\begin{figure}[htb]
\begin{center}
\input{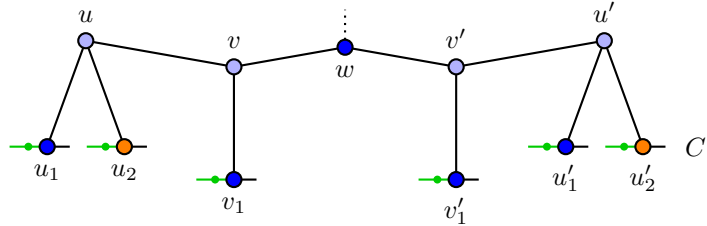}
\caption{A special case in the proof of Claim~\ref{c:2star-1star-any}(c)} \label{f:2star-1star-any-1star-2star}
\end{center}
\end{figure}

We now combine the scores of the two configurations, where the one configuration contains $u$ and $v$ and the other configuration contains $u'$ and $v'$. In this special case, we delete $u$, $v$, $w$, $u'$ and $v'$, and we add the vertex $w$ to the dominating set. The deletion of these five vertices contributes $4 \times 4 + 5 = 21$ to the score. We gain~$2$ by cutting the outgoing edge from $u$ to the amber vertex $u_2$ and the outgoing edge from $u'$ to the amber vertex $u_2'$. However, adding the vertex $w$ to the dominating set contributes~$-12$ to the score. Hence, the score when the set \setB\ is chosen is at least~$21 + 2 - 12 = 11$.

Recall that the score when the set \setA\ is chosen is at least~$4$ if the blue vertex $v_1$ is dotted and at least~$2$ if the blue vertex $v_1$ is not dotted. By symmetry, the score when the set \setA\ is chosen is at least~$4$ if the blue vertex $v_1'$ is dotted and at least~$2$ if the blue vertex $v_1'$ is not dotted. If at least one of $v_1$ and $v_1'$ is dotted, then the overall score of combining the two configurations is at least~$[6 \mid 11]$. Sharing this score equally among the two configurations yields a score of at least~$[3 \mid \frac{11}{2}]$ assigned to each configuration. This in turn yields an average score of at least $[\frac{17}{4}] > [4]$ for each configuration. Hence, we may assume that neither $v_1$ nor $v_1'$ is dotted. In this special case, when choosing the set \setA, we delete $u$, $v$, $w$, $u'$ and $v'$, and we add the vertex $w$ to the dominating set. The deletion of these five vertices contributes $21$ to the score. We gain~$4$ by cutting the outgoing edges to the blue vertices, namely $uu_1$, $vv_1$, $u'u_1'$, and $v'v_1'$ noting that neither $v_1$ nor $v_1'$ is dotted. Hence, the score when the set \setA\ is chosen is at least~$21 + 4 - 12 = 13$.
Thus, the overall score of combining the two configurations is at least~$[13 \mid 11]$.
Sharing this score equally among the two configurations yields a score of at least~$[\frac{13}{2} \mid \frac{11}{2}]$ assigned to each configuration.
This in turn yields once again an average score of at least $[\frac{24}{4}] > [4]$ for each configuration.

Suppose next that $w$ is a vertex of degree~$3$ in $G$. Recall that $w$ is not the center of an isolated $2$-star with respect to the \green-\black-cycle $C$. Thus, either $w$ is the center of a configuration of Claim~\ref{score-common-nbr-2stars}, or $w$ has at least one neighbor $v'$ that is the center of an isolated $1$-star that is adjacent to the center $u'$ of an isolated $2$-star with respect to the \green-\black-cycle $C$. In the latter, as before, we infer that the neighbors of $u'$ on the cycle $C$ are of different colors and the neighbor $v'_1$ of $v'$ on $C$ is colored blue, as illustrated in Figure~\ref{f:2star-1star-any-1star-2star} where in this case the vertex $w$ has degree~$3$ in $G$. We combine the scores of the two configurations and use analogous arguments to yield as before an average score of at least $[\frac{17}{4}] > [4]$ for each configuration. We note that in this case when the vertex $w$ is added to the dominating set, then we mark its third neighbor different from $v$ and $v'$.

Suppose finally that $w$ is the center of a configuration of Claim~\ref{score-common-nbr-2stars},
i.e.,  $w$ is adjacent to two centers of $2$-stars, namely $u'$ and $v'$ as illustrated in Figure~\ref{f:2star-1star-w-2star-2star}. Let $W'$ be the set of neighbors of $u'$ and $v'$ on the alternating \green-\black\ cycle $C$, namely $W' = \{ u'_1, u'_2, v'_1, v'_2 \}$.
In this case, associated with the entire subgraph, we need a nonnegative contribution for the two isolated $2$-stars with $w$ as a common neighbor,
and an average of $[4]$ for the link, so the average contribution needs to be at least $[4]$ on the whole subgraph.

\begin{figure}[htb]
\begin{center}
\input{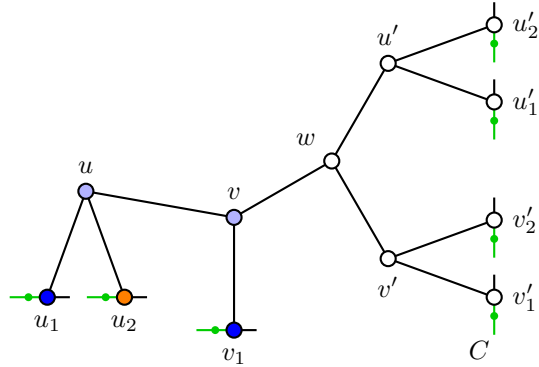}
\caption{Another special case in the proof of Claim~\ref{c:2star-1star-any}(c)}
\label{f:2star-1star-w-2star-2star}
\end{center}
\end{figure}

Assume first that one vertex in $W'$ is blue, say $u'_1$.
In this case, we remove also $u$, $v$, $w$, $u'$ and $v'$, and we add $v'$ to the set.
The contribution of removing the vertices is $20$, and the contribution of adding the vertex $v'$ to the dominating set is~$-12$. We may also lose $-3$ for cutting the edge $u'u'_2$ if $u'_2$ is dotted amber, and we gain~$1$ for deleting the edge $uu_2$. We mark any dotted amber neighbor of $v'$.
The total contribution to the score is therefore at least $[2 \mid 6]$, and thus the average score is at least $[4]$ as required.

Suppose now that all vertices in $W'$ are amber, and none is dotted.
In that case, we put $w$ in the set and remove the vertices $u$, $v$, $w$, $u'$ and $v'$.
The contribution is $20$ for deleting the vertices, $-12$ for adding the vertex $w$,
and we gain $5$ for deleting the outgoing edges.
The total contribution to the score is therefore at least $[2 \mid 13]$, yielding an average score of at least~$[\frac{15}{2}]$, which is more than enough for the desired average score of at least $[4]$.

Suppose now that all vertices in $W'$ are amber, and at least one is dotted
but all dotted vertices are on only one side, say only $u'_1$ and $u'_2$ are dotted.
In this case, we mark $u'$, delete $u$, $v$, $w$ and $v'$, and add $w$ to the set.
The contribution to the score is $16$ for deleting the vertices, $-12$ for adding vertex $w$. Moreover, we obtain~$3$ for deleting the outgoing edges. The total contribution to the score is therefore at least $[2 \mid 7]$, once again yielding an average score of at least $[4]$ as required.

Finally, suppose all vertices in $W$ are amber and there are dotted vertices on both sides,
say at least $u'_1$ and $v'_1$ are dotted.
In this case, we remove $u$ and $v$, with a contribution of $8$ to the score,
and cut the edge $vw$ with a contribution of $-1$to the score. We also gain~$1$ for cutting the edge $uu_2$. The total contribution to the score is therefore at least $[2 \mid 8]$, once again yielding an average score of at least $[4]$, as required. This completes the proof of Claim~\ref{c:2star-1star-any}.~\smallqed

\subsubsection{Assigning colors \Amber\ and \Blue\ to the \grees}
\label{s:set-colors}

We now define the coloring of the \gree\ in the \green-\black\ cycle to ensure that selecting one set among \setA\ or \setB\ ends up with a nonnegative score. Recall that our \greencyclegraph\ consists now of one alternating \green-\black\ cycle, to which are added
some \links\ joining two vertices in the \green-\black\ cycle, and fibers attached to only one vertex in the cycle, as in Definition~\ref{defn:greencyclegraph}.
The \grees\ are given colors among \Amber\ and \Blue, from which are deduced the colors of vertices
as from Definition~\ref{defn:AB-sets}.

Initially, we let all \grees\ of the \green-\black\ cycle be colored \Amber.
Let $w_B$ denote the number of well colored broad \links\ in the \greencyclegraph, $w_N$ the number of well colored narrow \links, $b_B$ the number of badly colored broad \links, and $b_N$ the number of badly colored narrow \links. By definition of our scoring,
if
\[
5w_B + 4w_N - b_B - \frac{1}{2}b_N \ge 0.
\]
then the average score of all configurations together is nonnegative. This necessarily implies that either the set \setA\ has a nonnegative score, or the set \setB\ does.

If the average score of all configurations is less than zero, then we apply a set of rules (to be defined shortly) iteratively to modify the \grees\ coloring,
as long as one of them apply.
After proving that this process necessarily ends,
we prove important properties of the resulting \greencyclegraph, and use a discharging procedure to guarantee
that the resulting average score is nonnegative, including the contributions of changes of colors and
of dotted edges and fibers.
From these arguments, we infer that the conclusion is the same, i.e.,  that one of \setA\ and \setB\ must have a nonnegative score.

The underlying idea of each of the following rules is the following.
We define a \emph{cut} that is a portion of the alternating \green-\black\ cycle
that goes from one \blae\ to another.
These two extremal \blaes\ are called the \emph{ends} of the cut.
The central idea is to switch the colors of all \Amber\ and \Blue\ \grees\ in that cut
(and thus also switching the colors of the vertices inside the cut)
whenever that increases the total average score of the \greencyclegraph.
During this process, we introduce in the \greencyclegraph\ changes of colors and dotted vertices as defined in Definition~\ref{defn:dotted}. Recall that each change of color contributes~$-4$ to the score of the resulting removed vertices.

We now formally define three rules, each corresponding to a different situation of a cut.
The first rule applies when neither \blae\ at the end of the cut corresponds to a change of color,
the second applies when both ends correspond to a change of color,
the third when precisely one end corresponds to a change of color.
Each rule simply describes necessary and sufficient condition such that
changing the colors on one side of the cut increases the average score,
or does not change the average score but reduces the total number of changes of color in the cycle,
or changes neither of these two parameters but decreases the number of dotted fibers.

For this, we identify now in variables all \links\ and fibers whose score is modified by such an operation.
In particular this implies all \links\ and fibers that are \emph{hit by a cut},
i.e.,  that are incident to an extremity of a \blae\ of the cut.
It also implies \links\ that are \emph{cut},
i.e.,  \links\ that have one extremity inside the cut and the other outside.
Note that \links\ with a dotted extremity do not change score when cut, and thus are not considered below.

For Rule 1, we use the following variables: \\ [-24pt]
\begin{enumerate}
\item[$\bullet$] $wc$ is the number of extremities of cut, not hit, well colored \links. \5
\item[$\bullet$] $bc$ is the number of extremities of cut, not hit, badly colored \links. \5
\item[$\bullet$] $wch$ is the number of extremities of cut, hit, well colored \links. \5
\item[$\bullet$] $bch$ is the number of extremities of cut, hit, badly colored \links. \5
\item[$\bullet$] $wh$ is the number of extremities of not cut, hit, well colored \links. \5
\item[$\bullet$] $bh$ is the number of extremities of not cut, hit, badly colored \links. \5
\item[$\bullet$] $oh$ is the number of link extremities hit whose opposite extremity is dotted. \5
\item[$\bullet$] $f$ is the number of hit fibers.
\end{enumerate}

Moreover, each of these \link\ variables may be restricted to broad or narrow links,
by the addition of a subscript $*_B$ or $*_N$, respectively. We remark that a \link\ that is hit on both extremities should be counted twice, within two variables.

\noindent
\noindent \textbf{Rule 1:} For cuts whose ends do not correspond to color changes, switching the colors in the cut will increase the score whenever:
\begin{equation}
\begin{split}
   6 (bc_B - wc_B) + \frac{9}{2} (bc_N - wc_N - wch_B - wh_B) -4 (wh_N + wch_N)
     &\\
 + \frac{3}{2}(bch_B+bh_B) +\frac{1}{2} (bch_N + bh_N) + \frac{1}{2} (oh_B - f) & > 8.
 \end{split}
\label{Rule1}
\end{equation}

A schematized green cycle graph coloring and application of Rule 1 is illustrated in Figure~\ref{figure-Rule-1}.

\begin{figure}[htb]
\begin{center}
\input{green-black-cycle-strategy3.tex}
\caption{Schematized green cycle graph coloring and application of Rule 1.}
\label{figure-Rule-1}
\end{center}
\end{figure}

\noindent
\textbf{Proof that Rule 1 increases the score:}
Based on Definitions~\ref{d:score-link} and \ref{d:score-fiber},
the score of each \link\ and fiber after implementation of Rule~1 changes as follows:
\begin{itemize}
\item A well colored \link\ that is cut but not hit has a score changing
from $+5$ to $-1$ if it is broad, from $+4$ to $-\frac{1}{2}$ if it is narrow.
The resulting contribution to the score is $-6 wc_B - \frac{9}{2} wc_N$.

\item A badly colored \link\ that is cut but not hit has a score changing
from $-1$ to $+5$ if it is broad, from $-\frac{1}{2}$ to  $+4$ if it is narrow.
The resulting contribution to the score is $6 bc_B + \frac{9}{2} bc_N$.

\item A well colored link that is hit, whether cut or not, has a score that changes
from $+5$ to $+\frac{1}{2}$ if it is broad, from $+4$ to $0$ if it is narrow.
The resulting contribution to the score is $-\frac{9}{2} (wch_B + wh_B) -4 (wch_N + wh_N)$.

\item A badly colored link that is hit, whether cut or not, has a score that changes
from $-1$ to $+\frac{1}{2}$ if it is broad, from $-\frac{1}{2}$ to $0$ if it is narrow.
The resulting contribution to the score is $ \frac{3}{2} (bch_B + bh_B) + \frac{1}{2} (bch_N + bh_N)$.

\item A link whose other extremity is dotted see its score changing from $\frac{1}{2}$ to $1$ if it is broad, but its score remains $0$ if it is narrow. The resulting contribution to the score is $\frac{1}{2} oh_B$.

\item A fiber that is hit has its score dropping from $0$ to $-\frac{1}{2}$. The resulting contribution to the score is $-\frac{1}{2}f$

\item Finally, the rule introduces two new changes of color. The resulting contribution to the score is $-8$.
\end{itemize}

Summing it all up, yields the formula given in~(\ref{Rule1}). This formula can be simplified when ignoring whether the \links\ are broad or narrow, taking the worse scenario. This provides the following simplified formula:

\noindent
\textbf{Simplified Rule 1:}
\begin{equation}
- 6wc + \frac{9}{2} (bc - wch - wh) + \frac{1}{2} (bch + bh - f) > 8.
\label{Simplified-Rule1}
\end{equation}

\smallskip
For Rule~2, we use similar definitions as above,
but all \links\ and fibers that are hit are incident to color changes, and so were dotted.

Before formally stating Rule~2, we add the following definitions of variables:
\\ [-24pt]
\begin{enumerate}
\item[$\bullet$] $wchd$ is the number of extremities of cut, hit, well colored dotted \links. \5
\item[$\bullet$] $bchd$ is the number of extremities of cut, hit, badly colored dotted \links. \5
\item[$\bullet$] $whd$ is the number of extremities of not cut, hit, well colored dotted \links. \5
\item[$\bullet$] $bhd$ is the number of extremities of not cut, hit, badly colored dotted \links. \5
\item[$\bullet$] $ohd$ is the number of dotted link extremities hit whose opposite extremity is also dotted. \5
\item[$\bullet$] $fd$ is the number of dotted hit fibers.
\end{enumerate}

\noindent \textbf{Rule 2:} For cuts whose both ends correspond to color changes,
switching the colors in the cut will not decrease the score (and reduce the number of color changes) whenever:
\begin{equation}
   \begin{split}
   6 (bc_B - wc_B)
   + \frac{9}{2} (bc_N - wc_N + whd_B + bchd_B)
    + 4 (whd_N + bchd_N)  & \\
    -\frac{3}{2} (bhd_B + wchd_B)
    + \frac{1}{2} (fd - bhd_N - wchd_N - ohd_B) & \ge -8
 \end{split}
  \label{Rule2}
\end{equation}

\noindent
\textbf{Proof that Rule~2 does not decrease the score:}
Based on Definitions~\ref{d:score-link} and \ref{d:score-fiber},
the score of each \link\ and fiber after implementation of Rule~2 changes as follows:
\begin{itemize}
\item A well colored \link\ that is cut but not hit has a score changing from $+5$ to $-1$ if it is broad, from $+4$ to $-\frac{1}{2}$ if it is narrow. The resulting contribution to the score is $-6 wc_B - \frac{9}{2} wc_N$.

\item A badly colored \link\ that is cut but not hit has a score changing from $-1$ to $+5$ if it is broad, from $-\frac{1}{2}$ to  $+4$ if it is narrow. The resulting contribution to the score is $6 bc_B + \frac{9}{2} bc_N$.

\item A well colored dotted link that is hit but not cut, or a badly colored dotted link that is hit and cut, has a score that changes  from $+\frac{1}{2}$ to $+5$ if it is broad,  and from $0$ to $+4$ if it is narrow. The resulting contribution to the score is $\frac{9}{2} (whd_B + bchd_B) + 4(whd_N + bchd_N)$.

\item A badly colored dotted link that is hit but not cut or a well colored dotted link that is hit and cut, has a score that changes from $+\frac{1}{2}$ to $-1$ if it is broad, from $0$ to $-\frac{1}{2}$ if it is narrow. The resulting contribution to the score is $ - \frac{3}{2} (bhd_B + wchd_B) - \frac{1}{2} (bhd_N + wchd_N)$.

\item A link whose other extremity is dotted see its score changing from $1$ to $\frac{1}{2}$ if it is broad, but its score remains $0$ if it is narrow. The resulting contribution to the score is $- \frac{1}{2} ohd_B$.

\item A fiber that is hit has its score changing from $-\frac{1}{2}$ to $0$. The resulting contribution to the score is $+\frac{1}{2} fd$

\item Finally, the rule removes two changes of color. The resulting contribution to the score is $+8$.
\end{itemize}

Summing it all up, yields the formula given in~(\ref{Rule2}). We also apply Rule~2 when there is equality since it will reduce the number of color changes, which we want to minimize if it does not increase the score. This formula can be simplified when ignoring whether the \links\ are broad or narrow, taking the worse scenario. This provides the following simplified formula:

\noindent
\textbf{Simplified Rule 2:}
\begin{equation}
- 6wc + \frac{9}{2} bc + 4 (whd + bchd) - \frac{3}{2}(bhd + wchd) + \frac{1}{2} (fd - ohd ) \ge -8.
\label{Simplified-Rule2}
\end{equation}

\smallskip
We present next Rule~3 which combines all the changes of the above rules.

\noindent \textbf{Rule 3:}
For cuts  with exactly one end corresponding to a color change, switching the colors in the cut will not decrease the score (and reduce the number of color changes) whenever:
  \begin{equation}
   \begin{split}
   6 (bc_B - wc_B)
   + \frac{9}{2} (bc_N - wc_N  - wch_B - wh_B + whd_B + bchd_B) & \\
    + 4 (whd_N + bchd_N - wh_N - wch_N)
    +\frac{3}{2} (bch_B + bh_B - bhd_B - wchd_B) & \\
    + \frac{1}{2} (bch_N + bh_N - bhd_N - wchd_N + oh_B - ohd_B -f + fd) & > 0
 \end{split}
 \label{Rule3}
 \end{equation}
or when the left term is equal to~$0$, but $fd > 0$.
\smallskip

\noindent
\textbf{Proof that Rule 3 does not decrease the score:} The proof is simply a combination of the two above proofs.

We also apply Rule~3 when there is equality since it will reduce the number of dotted fibers, which we want to minimize. This formula can be simplified when ignoring whether the \links\ are broad or narrow, taking the worse scenario. This provides the following simplified formula:

\noindent
\textbf{Simplified Rule 3:}
 \begin{equation}
 \begin{split}
  - 6wc + \frac{9}{2} (bc -wh -wch)  + 4 (whd + bchd)  -\frac{3}{2}(wchd + bhd) & \\
  + \frac{1}{2}(bh + bch + fd - f - ohd_B)  & > 0.
  \end{split}
\label{Simplified-Rule3}
\end{equation}

Suppose we apply these rules whenever such a rule can be applied. We show that the result has a positive overall score.
We first note that this process is finite.
Indeed, every application of a rule either increase the score (Rule 1, 2 or 3),
or does not change the score but reduces the number of color changes (Rule 2),
or does not change the score or the number of color changes,
but reduces the number of dotted fibers (Rule 3).

In the following, we prove some properties of the coloring resulting by applying the rules, until none of the three rules can be applied.
We generally consider a configuration with \blaes\ on the cycle
and \links\ of fiber extremities, implicitly using an orientation of the cycle
(clockwise or counterclockwise,
it does not matter provided it is consistent within the configuration).
For a \blae\ $e$, we call the \emph{first \blae\ after} $e$, denoted $e^{+1}$,
the \blae\ that is further by the orientation but adjacent to the same \gree.
We similarly define the \emph{first edge before} $e$ that we denote $e^{-1}$.
The first \blae\ before or after a vertex $x$ is the first \blae\ before or after
the \blae\ incident to $x$.

For a \blae\ $e$ or a vertex $x$, we say it is \emph{followed by} (or \emph{preceded by})
a \link\ extremity or a color change
when there is neither color changes
nor extremities of \link\ whose other extremity is not dotted in between.
In other words, a \blae\ $e$ is followed by a vertex $x$
if all \blae\ in between may be incident only to fibers or extremities of dotted \links,
which will not contribute to the score in the application of one of the above rules.
Implicitly, the \link\ extremity has its other extremity not dotted.

A \blae\ $e$ or a vertex $x$ is followed by two \link\ extremities $x_2$ and $x_3$
when it is followed by the \link\ extremity $x_2$ itself followed by the other \link\ extremity $x_3$,
and so on.

Observe that when $x$ is the extremity of a \link\ or a fiber,
it is adjacent to the center of a black star, and so by Claim~\ref{c:claim23},
the first \blae\ before or after $x$ is incident to at most one vertex
itself adjacent to the center of a black star
(which could then be the extremity of a \link\ or a fiber).
A similar property holds for three consecutive \blaes\ as stated by the following claim.

\begin{claim}\label{c:consecutive-blae}
If $x_1$, $x_2$ and $x_3$ are three extremities of \links\ or fibers,
then there is a \blae\ hitting nothing between $x_1$ and $x_3$.
\end{claim}
\begin{proof}
Suppose firstly  that two extremities (say $x_1$ and $x_2$) are incident to the same \blae\ $e$.
Then by Claim~\ref{c:claim23}, $e^{-1}$ and $e^{+1}$ are incident to no \link\ or fiber and thus
 $e^{+1}$ is a \blae\ hitting nothing between $x_1$ and $x_3$.
Suppose now that $x_1$, $x_2$ and $x_3$ belong to three consecutive \blaes,
say $u'_1u_2$, $u'_2u_3$ and $u'_3u_4$.
 By symmetry, suppose $x_2 = u_3$ is incident to a \link\ or a fiber,
 and thus that $u_3$ is adjacent to the center of a black star. We illustrate this configuration in Figure~\ref{f:consecutive-blae}.
 Then, by Claim~\ref{c:claim23},
 neither $u_2$ nor $u_4$ is adjacent to the center of a black star.
 By our assumption, this means that both $x_1 = u'_1$ and $x_3 = u'_3$
 are incident to a fiber or a \link, and thus adjacent to the center of a black star.
 But this is in contradiction with Claim~\ref{c:claim24}.\smallqed
\end{proof}

\begin{figure}[htb]
\begin{center}
\input{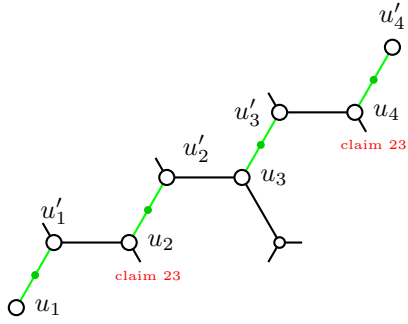}
\caption{An illustration of the proof of Claim~\ref{c:consecutive-blae}} \label{f:consecutive-blae}
\end{center}
\end{figure}

\begin{claim}\label{c:no-short-arch}
For every bad \link\ whose both extremities are in the same color region (that we call an arch),
 there are at least two \blaes\ not hitting either extremities within the arch.
 Moreover, one \blae\ is incident to no \link\ or fiber.
\end{claim}
\begin{proof}
Since a bad \link\ in a color area has both extremities of same parity,
the path between the extremities is composed of identical numbers of \grees\ and \blaes.
Thus if the path has $k$ \blaes, the cycle composed of this path and the arch is of length $3k+2$.
By the initial assumption, no cycle of length 5 or 8 belong to the graph, so there are at least three \blaes\ in the path, and two of them do not hit the extremities of the arch. The second part of the claim is now a direct consequence of Claim~\ref{c:consecutive-blae}.
\smallqed
\end{proof}

\begin{claim}\label{c:adjacent-reds}
 No \blae\ on the \greencyclegraph\ is incident to two badly colored links.
\end{claim}
\begin{proof}
Suppose, to the contrary, that a \blae\ $e$ is incident to two badly colored links.
By Claim~\ref{c:claim23}, $e^{-1}$ and $e^{+1}$ are incident to no \link\ or fiber.
 Thus we apply Rule~1 that match the number of color change on $e^{-1}$ and $e^{+1}$.
 We get that $bc =2$ and all other values are zero, and so the counting is at worse
 $2\times \frac{9}{2} > 8$, which is a contradiction. \smallqed
\end{proof}

\begin{claim}\label{c:blae-followed-by-three-red}
If a \blae\ $e$ that hits nothing is followed by three badly colored \link\ extremities $x_1$, $x_2$ and $x_3$, then  $x_1x_2$ is an arch.
\end{claim}
\begin{proof}
Let the \blae\  $e$ and extremities $x_1$, $x_2$ and $x_3$ be defined as in the claim statement,
and suppose, to the contrary, that $x_1x_2$ is not an arch.
 Let $e'$ be the first \blae\ after $x_2$.
 By Claim~\ref{c:claim23}, $e'$ is incident to at most one extremity of a \link\ or a fiber.
We now apply simplified Rule~1 from $e$ to $e'$.
We note that $wh + wch = 0$ or $x_3$ would not follow $x_2$. This yields $bc = 2$, $f + bh + bch + oh \le 1$, with all other terms in the simplified Rule~1 zero. The counting gives us at least $2\times\frac{9}{2} - \frac{1}{2} > 8$, and simplified Rule~1 applies, which is a contradiction. \smallqed
 \end{proof}

\begin{claim}\label{c:four-consecutive-red}
If an extremity $x_1$ of a badly colored \link\ is followed by three other badly colored \link\ extremities, namely $x_2$, $x_3$ and $x_4$, then $x_2x_3$ is an arch.
\end{claim}
\begin{proof}
Let $x_1$, $x_2$, $x_3$ and $x_4$ are defined as in the claim, and suppose, to the contrary, that $x_2x_3$ is not an arch. We consider the first \blae\ $e_\ell$ before $x_2$
and the first \blae\ $e_r$ after $x_3$,
and try to apply simplified Rule~1 on the cut from $e_\ell$ to $e_r$.
%
By Claim~\ref{c:claim23}, each of $e_\ell$ and $e_r$ is incident to at most one \link\ or fiber.
Moreover, if $e_\ell$ (resp., $e_r$) hits anything, then by Claim~\ref{c:consecutive-blae},
$e_\ell^{-1}$ (resp. $e_r^{+1}$) hits nothing.
We note that by hypothesis, if $e_\ell$ (resp., $e_r$) is incident to the extremity of a \link,
either it has its other extremity dotted, or it is $x_1$ (resp., $x_4$).

Suppose that $e_\ell$ is incident to $x_1$, and that
$x_1$ is the other extremity of an arch ending in $x_2$ or $x_3$.
Note that by Claim~\ref{c:no-short-arch}, the arch must be $x_1x_3$
and there is a \blae\ $e$ hitting nothing between $x_2$ and $x_3$.
Applying Rule~1 from $e_\ell^{-1}$ to $e$, we get that $bc=2$
(corresponding to the links ending in $x_1$ and $x_2$)
and nothing else is cut or hit.
The counting gives us at least $2\times\frac{9}{2} - \frac{1}{2} > 8$, and simplified Rule~1 applies, which is a contradiction.

By a symmetric argument, $e_r$ is not incident to $x_4$. Further,  $x_4$ does not form an arch with $x_2$ or $x_3$.
Now, if $e_\ell$ (resp., $e_r$) hits anything, we replace it with $e_\ell^{-1}$ (resp., $e_r^{+1}$).
The same counting applies except possibly $bc = 3$ or~$4$, which is even better. Thus once again, simplified Rule~1 applies, which is a contradiction.~\smallqed
\end{proof}

As a corollary of Claim~\ref{c:four-consecutive-red}, we obtain the following property of badly colored \link\ extremities.

\begin{claim}\label{c:no-5-in-a-row}
No badly colored \link\ extremity is followed by four badly colored \link\ extremities.
\end{claim}
\begin{proof}
Suppose, to the contrary, that $x_1$, $x_2$, $x_3$, $x_4$ and $x_5$ are four badly colored \link\ extremities following a badly colored \link\ extremity. By Claim~\ref{c:four-consecutive-red}, $x_2x_3$ and $x_3x_4$ are both arches, a contradiction.
\end{proof}

We now consider fiber and \link\ extremities hit by a color change.

\begin{claim}\label{c:hits-only-one}
In the \greencyclegraph, no color change hits a fiber.
Moreover, if a color change hits two \link\ extremities, both also have their other extremity dotted.
\end{claim}
\begin{proof}
Consider a \blae\ $e$ on the \greencyclegraph\ corresponding to a color change. We note that both extremities of $e$ are Amber or both are Blue. Suppose firstly  that the color change $e = u'_iu_{i+1}$ hits two extremities,
either of some \links\ or fibers.
By Claim~\ref{c:claim23}, the \blaes\ $e^{-1}$ and $e^{+1}$ hit no \link\ or fiber.

Suppose, to the contrary, that one extremity is of a \link\ and one of a fiber.
By symmetry, say $u'_i$ hits the extremity of a \link\ and $u_{i+1}$ hits the extremity of a fiber.
If the \link\ incident to $u'_i$ is badly colored, we use Rule~3
(or Rule~2 if $e^{-1}$ is also a change of color)
on the cut from $e^{-1}$ to $e$.
In this case, $bchd = 1$ and $fd = 1$, with all other terms in the associated formula for Rule~3 equal to zero.
This yields $+4+\frac{1}{2} > 0$, and simplified Rule 3 applies, a contradiction.
If the \link\ incident to $u'$ is well colored, we use Rule~3
(or Rule~2 if $e^{+1}$ is also a change of color)
on the cut from $e$ to $e^{+1}$.
In this case, $whd = 1$ and $df = 1$, with all other terms in the associated formula for Rule~3 equal to zero.
This yields $+4+\frac{1}{2} > 0$ and Rule 3 applies, a contradiction.
If the \link\ incident to $u'$ has its other extremity dotted, we use Rule~3
(or Rule~2 if $e^{+1}$ is also a change of color)
on the cut from $e$ to $e^{+1}$.
In this case, $ohd = 1$ and $df = 1$, with all other terms in the associated formula for Rule~3 equal to zero.
If $e^{+1}$ is also a color change, we get $\frac{1}{2}(+1 - 1) > -8$ and Rule~2 applies,
otherwise we get $\frac{1}{2}(+1 - 1) = 0$ and Rule~3 applies
since its application reduces the number of dotted fibers. In both cases, we obtain a contradiction.

Suppose next that the color change $e$ hits two fiber extremities.
We use Rule~3 on the cut from $e^{-1}$ to $e$.
In this case, $df = 2$, with all other terms in the associated formula for Rule~3 equal to zero.
This yields $2\frac{1}{2} > 0$, and Rule~3 applies, a contradiction.

Suppose now that the color change $e$ hits two \link\ extremities.
If both \links\ are well colored or both are badly colored,
then we use Rule~3 on the cut from $e^{-1}$ to $e$.
We have either $bhd = bchd = 1$ or $whd = wchd = 1$, with all other terms in the associated formula for Rule~3 equal to zero. This yields $4 - \frac{3}{2} > 0$, and Rule 3 applies, a contradiction.
If one \link\ is well colored and the other badly, say $u'_i$ is well colored,
then we use Rule 3 on the cut from $e$ to $e^{+1}$.
In this case, $whd = 1$ and $bchd = 1$, with all other terms in the associated formula for Rule~3 equal to zero.
This yields $4+4 > 0$, and Rule~3 applies, a contradiction.

Suppose now that one \link\ has its other extremity dotted, say $u'_i$.
If the other \link\ is well colored, then we apply Rule~3 from $e^{-1}$ to $e$,
where $whd = 1$ and $ohd=1$, with all other terms in the associated formula for Rule~3 equal to zero. The counting gives $+4-\frac{1}{2} >0$, and Rule~3 applies, a contradiction.
If the other \link\ is badly colored, then we apply Rule~3 from $e$ to $e^{+1}$.
In this case, $bchd = 1$ and $ohd=1$, with all other terms in the associated formula for Rule~3 equal to zero.  The counting is the same, and so Rule~3 applies, a contradiction.
Therefore, only the case when the other \link\ also has its other extremity dotted remains,
which may occur according to the claim.

Suppose now that the color hits only one extremity, of a fiber.
By Claim~\ref{c:consecutive-blae},
at least one of the \blaes\ $e^{-1}$ and $e^{+1}$ hit no \link\ or fiber, name it $e'$.
Applying Rule~3 from $e$ to $e'$, we get $df = 1$, with all other terms in the associated formula for Rule~3 equal to zero. This yields $\frac{1}{2} > 0$, and Rule 3 applies, a contradiction.\smallqed
\end{proof}

\begin{claim}\label{c:cost-colchange}
The contribution of the links that are hit by one color change when applying Rule~2 or~3 is at least $-\frac{3}{2}$. Moreover, if the color change does not hit a broad \link\ extremity,  then the contribution is at least $-\frac{1}{2}$.
\end{claim}
\begin{proof}
Using Claim~\ref{c:hits-only-one},
we know that for a color change when applying Rule~2 or~3,
 either $ohd \le 2$ and $wch+wh+bch+bh = 0$,
 or $ohd = 0$ and $wch+wh+bch+bh \le 1$. In both cases, $fd = 0$.

First assume that the \links\ hit may be broad. The contribution of $ohd$ is then $-\frac{1}{2}$ and the contribution of other links that are hit (from $wch+wh+bch+bh$) is then at least $-\frac{3}{2}$.
 In the first situation, the contribution of \links\ that are hit by this color change
 is then at least $2\times-\frac{1}{2} = -1$
 and in the second situation it is at least $-\frac{3}{2}$. This yields a minimum contribution of at least  $-\frac{3}{2}$ for each color change.

If we add the assumption that the color change hits no broad \link, then $ohd$ contributes $0$ and
the other \links\ that are hit contribute at least $-\frac{1}{2}$.
In the first situation, the contribution is~$0$ and in the second situation, the contribution is at least $-\frac{1}{2}$. This yields a minimum contribution of at least  $-\frac{1}{2}$ for each color change. \smallqed
\end{proof}

In view of Claim~\ref{c:cost-colchange}, we can therefore reformulate Rules 2 and 3,
 introducing $cc_B$ for counting the ends of the cut corresponding to color changes
 that hit at least one broad \link\
 ($0 \le cc_B \le 2$ for Rule 2, $0 \le cc_B \le 1$ for Rule 3).
 In Rule~2, one color change contributes at least $-\frac{1}{2} - cc_B$,
 while in Rule~3, the two color changes contribute $2\times -\frac{1}{2} - cc_B$.

\noindent
 \textbf{Reformulated Rule 2:}
\begin{equation}
6 (bc_B - wc_B)
   + \frac{9}{2} (bc_N - wc_N)
     - cc_B \ge -7
 \end{equation}

\noindent
 \textbf{Reformulated simplified Rule 2:}
\begin{equation}
- 6 wc
   + \frac{9}{2} bc
     - cc_B \ge -7
 \end{equation}

\noindent
  \textbf{Reformulated Rule 3:}
\begin{equation}
  \begin{split}
   6 (bc_B - wc_B)
   + \frac{9}{2} (bc_N - wc_N  - wch_B - wh_B )
    + 4 ( - wh_N - wch_N) & \\
    +\frac{3}{2} (bch_B + bh_B) - cc_B
    + \frac{1}{2} (bch_N + bh_N  + oh_B -f ) & > \frac{1}{2}
 \end{split}
 \end{equation}

\noindent
\textbf{Reformulated simplified Rule 3:}
\begin{equation}
    -6wc + \frac{9}{2} (bc - wch - wh)
    - cc_B
    + \frac{1}{2} (bch + bh  -f ) > \frac{1}{2}
 \end{equation}

\begin{claim}\label{c:no-consecutive-colchange}
No color change is followed by another color change.
\end{claim}
\begin{proof}
Suppose, to the contrary, that $e$ is a color change followed by another color change $e'$.
In this case, we apply the reformulated Rule~2 from $e$ to $e'$,
where $wc+bc = 0$ and $cc_B \le 2$. The counting gives at least $-2 \ge -7$, and the reformulated simplified Rule~2 applies, a contradiction.
\end{proof}

\begin{claim}\label{c:colchange-separated-by-one}
If a color change $e$ is followed by one \link\ extremity and then another color change $e'$,
then the \link\ is well colored, broad, and both color changes hit a broad \link\ extremity.
\end{claim}
\begin{proof}
Let the color changes $e$ and $e'$ be defined as above.
We want to apply reformulated Rule~2 from $e$ to $e'$.
Suppose firstly  that the \link\ is not a broad well colored link. In this case,  we get $wc_B = 0$, $wc_N+bc_N + bc_B = 1$, and $cc_B \le 2$. The counting gives at least $-\frac{9}{2} - 2 = -\frac{13}{2} \ge -7$, and the reformulated simplified Rule~2 applies, a contradiction. This proves the first part of the claim, that the \link\ is a well colored broad \link. Suppose next that $wc_B = 1$. Suppose that at most one of the color changes hits a broad link,
i.e.,  that is, $cc_B \le 1$. In this case, the reformulated Rule~2 gives at least $-6 - 1  = -7$,
and again reformulated Rule~2 applies, a contradiction.
\smallqed
\end{proof}

\begin{claim}\label{c:2bad-after-colchange}
If there is a bad \link\ extremity $x_1$ preceded by a color change,
then it is immediately followed by a well colored \link\ extremity $x_2$,
i.e.,  either the two \link\ extremities are on the same \blae, or
 the well colored \link\ extremity $x_2$ is on the first \blae\ after the badly colored \link\ extremity $x_1$.
\end{claim}
\begin{proof}
Suppose that after a color change on some \blae\ $e$,
there is an extremity $x_1$ of a badly colored \link\
not immediately followed by a well colored \link\ extremity.
In particular, the first \blae\ $e'$ after $x_1$ does not hit a well colored \link.
We want to apply (reformulated simplified) Rule 3 from $e$ to $e'$. By Claim~\ref{c:no-short-arch},
$e'$ does not hit the other extremity of the badly colored \link\ ending on $x_1$, and so $bc = 1$.
By supposition, $wh + wch = 0$ and $wc = 0$.
By Claim~\ref{c:claim23}, $e'$ hits at most one fiber, and so $f \le 1$.
Applying the reformulated counting where $cc_B \le 1$ gives at least
$\frac{9}{2} -\frac{1}{2} -1 > \frac{1}{2}$, and reformulated Rule~3 applies, a contradiction. \smallqed
\end{proof}

\begin{claim}\label{c:colchange-separated-by-two}
If a color change $e$ is followed by two \link\ extremities and then another color change,
then both \links\ are distinct and well colored.
\end{claim}
\begin{proof}
Suppose at least one of the \links\ is badly colored. In this case, we apply (reformulated) Rule 2 from one color change to the other. By supposition, we have $wc \le 1$, $bc \ge 1$, and $cc_B \le 2$.
The counting thus gives at least $-6 + \frac{9}{2} -2 = -\frac{7}{2} \ge -7$,
and reformulated Rule 2 applies, a contradiction. If the two \link\ extremities are those of a same \link, then $wc = 0$, $bc = 0$, and $cc_B \le 2$, and the counting yields $-2 \ge -7$, implying once again that reformulated Rule~2 applies, a contradiction.
\end{proof}

\begin{claim}\label{c:3-extremities-between-colchange}
Suppose a color change is followed by three \link\ extremities, then another color change.
Then at least two extremities are those of well colored \links..
If only two are well colored, then
either both well colored \links\ are broad,
 or both changes of color hit a broad \link,
 or the badly colored \link\ is narrow and at least one color change hits a broad \link.
\end{claim}
\begin{proof}
Suppose that a color change $e$ is followed by three \link\ extremities, $x_1$, $x_2$ and $x_3$,
and then another color change $e'$.
First suppose only one \link\ extremity is of a well-colored \link.
By Claim~\ref{c:2bad-after-colchange}, $x_2$ is well colored, and $x_1$ and $x_3$ are badly colored.
Moreover, $x_2$ is immediately after $x_1$. Applying Claim~\ref{c:2bad-after-colchange} in the other direction, $x_3$ is also immediately after $x_2$. But this contradicts Claim~\ref{c:consecutive-blae}.

Assume now only two \link\ extremities are well colored.
We apply Rule~2 from $e$ to $e'$.
If neither color change hit a broad \link\,
then at least one of the well colored \link\ is narrow.
In this case we have $wc_N \ge 1$, $wc \le 2$, $bc = 1$, and $cc_B=0$.
The counting gives at least $-6 -\frac{9}{2} + \frac{9}{2} = -6 \ge -7$, and reformulated Rule 2 applies, a contradiction.

Suppose now the third \link\ extremity is badly colored and broad, and at least one color change does not hit a broad \link. In this case, we have $wc = 2$, $bc_B = 1$ and $cc_B\le 1$.
The counting gives at least $-2\times 6 + 6 -1 = -7$,
and the reformulated Rule~2 applies, a contradiction.
 \smallqed
\end{proof}

\begin{claim}\label{c:colchangeSeparatedBy4}
If a color change is followed by four \link\ extremities, and then another color change,
then at most two \link\ extremities are badly colored, and if there are two, they form an arch.
\end{claim}
\begin{proof}
Suppose that two color changes $e$ and $e'$ are separated by four \link\ extremities.
 Suppose, to the contrary, that there are at least two badly colored \link\ extremities
 which do not form an arch.
In this case we apply reformulated Rule 2 from $e$ to $e'$, and we get
$bc\ge 2$, $wc\le 2$, and $cc_B\le 2$.
The counting gives at least $2\frac{9}{2} - 12 - 2 = -5 \ge -7$ and reformulated Rule~2 applies, a contradiction.~\smallqed
\end{proof}

\begin{claim}\label{c:2good-among-4-after-colchange}
If a color change is followed by at least four \link\ extremities, then either there are two well colored \link\ extremities among the first four extremities, or the color change is followed by five \link\ extremities $x_1, x_2, \ldots, x_5$ such that the following properties hold in the configuration.
\\ [-24pt]
\begin{enumerate}
\item Among $x_1$ and $x_2$, one is the extremity of a badly colored \link\ (name it $x_b$),
 and the other is the extremity of a well colored \link\ (name it $x_w$).
\item If $x_b = x_1$, then $x_b$ is immediately followed by $x_2$.
\item The edge $x_bx_3$ is a badly colored arch.
\item The badly colored \link\ extremity $x_4$ is immediately followed by a well colored \link\ extremity $x_5$.
\item If $x_bx_3$ is broad and $x_w$ is the end of a narrow \link, then the color change hits a broad link.
\end{enumerate}
\end{claim}

\begin{figure}[htb]
\begin{center}
  \input{2good-among-4-after-colchange-special-case.tex}
 \caption{Illustration of the special case of Claim~\ref{c:2good-among-4-after-colchange} when $x_b = x_1$}
 \end{center}
\end{figure}

\begin{proof}
Let $e$ be a color change, followed by five \link\ extremities $x_1, x_2, \ldots, x_5$.
Suppose that no two of the first four \link\ extremities are well colored, i.e.,  only one of $x_1$, $x_2$, $x_3$, $x_4$ is well colored. We show that properties (a) to (e) hold.
By Claim~\ref{c:2bad-after-colchange}, the well colored \link\ extremity must be $x_1$ or $x_2$,
name it $x_w$ and the other $x_b$. Furthermore by Claim~\ref{c:2bad-after-colchange}, if $x_b$ is $x_1$, then it is immediately followed by $x_w = x_2$ (and so, properties (a) and (b) hold). Moreover, $x_3$ and $x_4$ are the extremities of badly colored links.

Suppose firstly that $x_bx_3$ is not an arch.
Let $e_3$ be the first \blae\ after $x_3$.
Our intention is to apply Rule~3 from $e$ to $e_3$.
In this case, $bc = 2$ and $wc = 1$.
By Claim~\ref{c:claim23}, $e'$ hits at most one fiber or link,
 which may be $x_4$ or a \link\ with the other extremity dotted.
Thus in the counting, we get
$bh + bch + f + oh \le 1$ and $wh + wch = 0$.
Using the reformulated simplified Rule~3 with $cc_B \le 1$, we get
$2 \times \frac{9}{2} - 6 - \frac{1}{2} - 1 = \frac{3}{2} > \frac{1}{2}$,
and reformulated Rule~3 applies, a contradiction.
Hence, $x_bx_3$ is a badly colored arch, which proves part~(c).

By Claim~\ref{c:no-short-arch},
there is a \blae\ $e_2$ between $x_b$ and $x_3$ hitting nothing.
 If $x_b = x_1$, since $x_2$ is immediately after $x_1$,
 this \blae\ is also between $x_2$ and $x_3$.
 Let $e_4$ be the first \blae\ after $x_4$ not hitting it.
 If there is no well colored \link\ extremity immediately after $x_4$,
 then $e_4$ does not hit the extremity of a well colored \link. In this case, we wish to apply Rule~1 from $e_2$ to $e_4$, with $bc = 2$ and $wc = 0$. We also have by supposition that $wh = wch = 0$.
Since $e_2$ hits nothing, by Claim~\ref{c:claim23} on $e_4$, we have $bch + bh + f + oh \le 1$.
The computation thus gives us at least
 $2\frac{9}{2} - \frac{1}{2} > 8$, and Rule~1 applies, a contradiction. Hence, there must be a well colored \link\ extremity $x_5$ immediately after $x_4$, which proves part~(d).

Finally, let $x_bx_3$ be broad and let $x_w$ be the end of a narrow \link.
Suppose, to the contrary, that  the color change does not hit the end of a broad link.
We wish to apply reformulated Rule~3 from $e$ to $e_2$.
In this case, $bc_B = 1$, $wc_N = 1$, $wch_B+ bch_B + wch_N + bch_N + oh_B =0 $, and $cc_B = 0$.
The computation gives $6 - \frac{9}{2} \ge \frac{1}{2}$, and reformulated Rule~3 applies, a contradiction. This proves part~(e).~\smallqed
\end{proof}

\begin{claim}\label{c:3goodFor5betweenColchange}
If a color change is followed by at least five \link\ extremities and then a color change,
then at least three of the \link\ extremities are well colored.
\end{claim}
\begin{proof}
First, suppose the area between the two color changes contains six or more \link\ extremities.
By Claim~\ref{c:2bad-after-colchange}, one of the first two extremities in both directions is
of a well colored \link, and those two are necessarily distinct.
By Claim~\ref{c:2good-among-4-after-colchange}, there must also be two well colored \link\ extremities
among the first five.
If the second well colored \link\ extremity is the same as the first one in the other direction,
then there are exactly six \link\ extremities between the two color changes,
and we are in the special situation of Claim~\ref{c:2good-among-4-after-colchange}.
But then $x_5$ is immediately after $x_4$,
and from the other direction, $x_5$ is immediately after $x_6$,
contradicting Claim~\ref{c:consecutive-blae}.

Suppose now the area between the two color changes contains exactly five \link\ extremities.
We denote the vertices that belong to these extremities by $x_1, x_2, x_3, x_4$ and $x_5$.
Suppose two or fewer of these \links\ are well colored.
By Claim~\ref{c:2bad-after-colchange},
one of $x_1$ and $x_2$ is the extremity of a well colored \link\
and also one of $x_4$ and $x_5$.
By supposition, all three other extremities are extremities of badly colored \links.
In particular, the \link\ ending in $x_3$ is the extremity of a badly colored \link.
This link can form an arch with only one of the two badly colored \link\ extremities.
We may assume it does not form an arch with $x_1$ or $x_2$.
Let $e^\dagger$ be the \blae\ immediately after $x_3$.
 The edge $e^\dagger$ may hit $x_4$ but in that case, $x_4$ is the extremity of a badly colored \link.
 Indeed, if $x_4$ is well colored, then by Claim~\ref{c:2bad-after-colchange}
 (considered in the opposite direction),
 it is immediately followed by $x_5$, and this contradicts Claim~\ref{c:consecutive-blae}.
Hence, $e^\dagger$ hits no well colored \link.
Applying the reformulated simplified Rule~3 between $e$ and $e^\dagger$, we have $bc = 2$, $wc = 1$, $0 \le bh + bch + f \le 1$, $wh = wch = 0$, $cc_B \le 1$. The counting therefore yields at least $2\frac{9}{2} - 6 - \frac{1}{2} - 1 = \frac{3}{2} > \frac{1}{2}$, and so reformulated Rule~3 applies, a contradiction. \smallqed
\end{proof}

\smallskip
We proceed further by defining what we have coined a ``heavy configuration.''

\begin{definition}[Heavy configurations]
Consider the extremity of a well colored \link. We say it is \emph{heavy} if it is in one of the following configurations: \\ [-24pt]
\begin{enumerate}
\item
preceded by a color change and followed by a bad \link\ extremity then a \blae\ hitting no fiber or well colored \link\ extremity,
\item
followed by a bad \link\ extremity then a color change,
\item
preceded by a color change and followed by a well colored narrow \link\ then a color change,
\item
followed by three badly colored \link\ extremities, the first two of which do not form an arch.
\end{enumerate}

We also have two special heavy configurations for broad links: \\ [-24pt]
\begin{enumerate}
\item[$\bullet$] \textbf{super heavy:} preceded and followed by color changes. \1
\item[$\bullet$] \textbf{slightly heavy:} preceded by a color change and
followed by a well colored \link\ then a color change or
preceded by a color change and followed by a \blae\ hitting nothing.
\end{enumerate}
\end{definition}


\begin{remark}
From Claim~\ref{c:four-consecutive-red},
if a well colored \link\ extremity is followed by four badly colored \link\ extremities,
then it is necessarily heavy from configuration $(d)$.
\end{remark}

\begin{claim}\label{c:heavy-conf}
No well colored \link\ has both extremities heavy or super heavy,
Moreover, no well colored \link\ has one extremity super heavy and the other one slightly heavy.
\end{claim}
\begin{proof}
Suppose one \link\ is well colored and has both extremities heavy or super heavy.
In this case we apply some rule that cuts the \link\ on each side, and the \link\ remains well colored.
We thus proceed by checking that if we don't have to pay for the cut of that \link,
then we can apply one of the three rules on each of the described configuration.

\textbf{Heavy configuration (a):}
Suppose the \link\ extremity $x_1$ is preceded by a color change $e$,
    and followed by a bad \link\ extremity $x_2$,
    itself followed by a \blae\ $e'$ hitting nothing else except possibly a badly colored \link.
In this case, if we apply reformulated simplified Rule~3 from $e$ to $e'$
    and ignore the cut of the well colored \link\ ending in $x_1$, then we get
    $bc = 1$, $cc_B \le 1$, $wc = 0$ (since we ignore the cut of the \link\ incident to $x_1$),
    $wh + wch = 0$ and $bh + bch + f \le 1$.
    The counting gives $\frac{9}{2} - 1 -\frac{1}{2}  > \frac{1}{2}$,
    and reformulated simplified Rule~3 applies, a contradiction.

\textbf{Heavy configuration (b):}
This case is similar to the above one.
  Denote by $e$ the \blae\ of the color change, $x_1$ and $x_2$ the \link\ extremities, where $x_2$ is the heavy extremity of the well colored \link\ (and $x_1$ is the extremity of a badly colored \link).
  From Claim~\ref{c:2bad-after-colchange}, $x_2$ is immediately after $x_1$, and thus by Claim~\ref{c:consecutive-blae}, there is a \blae\ $e'$ hitting nothing immediately after $x_2$.
  We apply reformulated simplified Rule~3 from $e$ to $e'$,
  and ignore the cut of the well colored \link\ ending in $x_2$. This yields
    $bc = 1$, $cc_B \le 1$, $wc = 0$ (since we ignore the cut of the \link\ incident to $x_2$),
    and $wh + wch + bh +bch + f = 0$ .
    The counting gives $\frac{9}{2} - 1 > \frac{1}{2}$, and reformulated simplified Rule~3 applies, a contradiction.

\textbf{Heavy configuration (c):}
Let $e$ and $e'$ be the two color changes.
We apply reformulated Rule~2 from $e$ to $e'$, where in this case we ignore the cut of one \link\
(or have no cut at all if the two extremities are those of an arch).
This yields $wc_N\le 1$, $cc_B \le 2$,
and the contribution of all other terms in the reformulated Rule~2 is zero.
The counting gives $-\frac{9}{2} - 2 > -7$, and reformulated Rule~2 applies, a contradiction.

\textbf{Heavy configuration (d):}
Let $x_1, x_2, x_3$ and $x_4$ be the extremities of the \links,
 here $x_1$ is the well colored \link\ extremity.
We note that by Claim~\ref{c:blae-followed-by-three-red}, since $x_2x_3$ is not an arch,
there is no \blae\ hitting nothing between $x_1$ and $x_2$.
Thus, $x_2$ must be immediately after $x_1$, and by Claim~\ref{c:consecutive-blae},
there is a \blae\ $e$ hitting nothing just before $x_1$.
Suppose there is no \blae\ hitting nothing between $x_3$ and $x_4$
(and thus $x_3x_4$ is not an arch by Claim~\ref{c:no-short-arch}).
In this case by Claim~\ref{c:consecutive-blae}, the first edge after $x_4$ hits nothing,
which contradicts Claim~\ref{c:blae-followed-by-three-red}.
Hence let $e'$ be a \blae\ between $x_3$ and $x_4$ hitting nothing.
We now apply simplified Rule~1 from $e$ to $e'$. This yields $bc = 2$, and the contribution of all other terms in the reformulated Rule~1 is zero.
The counting gives us $2\frac{9}{2} > 8$, and simplified Rule~1 applies, a contradiction.

We also need to show that a super heavy configuration satisfies the same, i.e.,  that we can apply one of the three rules if we ignore the cost of cutting the link.
We actually show that when ignoring the cost of the cut,
we can get a positive score with an extra saving of~$2$.

\textbf{Super heavy configuration:}
Suppose the \link\ extremity $x_1$ is surrounded by two color changes $e$ and $e'$.
In this case if we apply reformulated simplified Rule~2 from $e$ to $e'$
    and ignore the cut of the well colored \link\ ending in $x_1$, we get
    $bc = 0 = wc$ (since we ignore the cut of the \link\ incident to $x_1$) and $cc_B \le 2$.
    The counting gives $-2  \ge -7 + 2$ and reformulated simplified Rule~2 applies,
    with an extra saving of~$2$.

Finally, we show thanks to this saving, that a link cannot have one extremity super heavy and one slightly heavy.

\textbf{Slightly heavy configuration :}
Suppose the link extremity $x_1$ is preceded by a color change $e$ and
followed by a link extremity $x_2$ then a color change $e'$.
We apply reformulated simplified Rule~2 from $e$ to $e'$.
We ignore the cost of cutting the \link\ ending in $x_1$, and obtain $wc\le 1$, $cc_B \le 2$, and the contribution of all other terms in the simplified Rule~2 is zero.
The counting gives $-6 - 2 = -8$ which gives $-6$ thanks to the saving of $2$ from the super heavy other extremity, and reformulated simplified Rule~2 applies, a contradiction.

Suppose now the link extremity is next to a color change and a \blae\ hitting nothing.
We apply reformulated simplified Rule 3 from the color change to the \blae, we get
$cc_B \le 1$, and the contribution of all other terms in the reformulated simplified Rule~3 is zero. Thus the counting gives $-1$, which yields $+1$ thanks to the extra saving of~$2$, and so reformulated simplified Rule~3 applies, a contradiction. This completes the proof of Claim~\ref{c:heavy-conf}.~\smallqed
\end{proof}

\subsubsection{Discharging procedure}

Now, we want to show that at least one of the colorings Amber and Blue gives a positive score.
For that purpose, we show that the sum of the scoring for Amber and Blue is positive, and therefore infer that at least one is positive.


We say a well colored \link\ extremity is \emph{close} to a color change if they are separated by nothing else than badly colored \link\ extremities (or fibers or extremities of \links\ whose other extremity is dotted).

To prove the sum is positive, we will use some discharging, so that no negative weight remains in the graph. The process is the following:
\begin{itemize}
\item Well colored \links\ receive 4 or 5 depending whether they are broad or narrow.
 Those are shared among the extremities of the \link,
 giving $\frac{5}{2}$ to any extremity of a \link\ in a heavy configuration or of a broad \link\ on a slightly heavy configuration,  $\frac{7}{2}$ to any extremity of a (necessarily broad) \link\ in a super heavy configuration,
 and at least $\frac{3}{2}$ to other extremities.
 From Claim~\ref{c:heavy-conf}, no well colored \link\ has two heavy extremities, or one super heavy and one slightly heavy or heavy, and thus all \links\ have enough for this transfer.

\item Well colored \link\ extremities close to color changes
 give $\frac{7}{4}$ to any close color change hitting a broad \link, and~$2$ to any close color change hitting no broad \link.

 \item Well colored \link\ extremities transfer weight to surrounding badly colored \link\ extremities according to the following rules:
\begin{itemize}
 \item The transferred weight is $\frac{1}{2}$ if the badly colored \link\ extremity is broad, and
 $\frac{1}{4}$ if it is narrow.
\item Well colored heavy \link\ extremities always transfer weights to badly colored \link\ extremities at distance~$1$.
  \item Well colored heavy \link\ extremities close to no color change transfer weights to badly colored \link\ extremities at distance~$2$ when necessary.
  \item Well colored \link\ extremities close to no color change transfer weights to badly colored \link\ extremities at distance~$1$.
 \end{itemize}

 \item Every \link\ near a \link\ extremity close to a color change transfer to it a weight of $\frac{1}{2}$ as a backup if required.
 \item Every broad \link\ extremity hit by a color transfers a weight of $\frac{1}{2}$ to the color change.
 \end{itemize}

We need to prove that the resulting discharging results in nonnegative weights in all configurations. We discuss the configurations in turn.
\begin{description}

\item[color changes:] Every color change must get~$4$: By Claim~\ref{c:no-consecutive-colchange},
every color change gets at least $\frac{7}{4}$ on each side when it hits a broad link
(from which it receives $\frac{1}{2}$),
or at least~$2$ from each side otherwise.

\item[fibers hit by a color change:] By Claim~\ref{c:hits-only-one}, no color change hits a fiber.

\item[badly colored link extremity:] We need to show that each such extremity receives a weight of~$\frac{1}{2}$:

\begin{itemize}

\item Consider first a bad \link\ extremity \emph{far} from any color change,
i.e.,  such that none of the closest well colored \links\ extremities is close to a color change.
In this case we consider \emph{nearby} well colored \link\ extremities
(i.e.,  closest in each direction)
and the \emph{area} of badly colored \link\ extremities in between.
If in this area there are at most two badly colored \link\ extremity, then
each receives a half from its closest well colored \link\ extremity.
If the area contains three badly colored \link\ extremities,
say $x_1$, $x_2$ and $x_3$, then one of $x_1x_2$ and $x_2x_3$ is not an arch,
implying that $x_1$ or $x_3$ is next to a heavy well colored \link\ extremity
(by configuration (d)).
In this case, $x_1$ and $x_3$ receive $\frac{1}{2}$ from their closest well colored \link\ extremity,
and $x_2$ receive $\frac{1}{2}$ from the heavy extremity.
Finally, the area may not contain more than five bad \link\ extremities
(Claim~\ref{c:no-5-in-a-row}), and if it contains four,
then by Claim~\ref{c:four-consecutive-red}, the two middle ones form an arch.
We therefore inder that the well colored \link\ extremities nearby are both heavy
(by configuration (d) again)
and all four bad \link\ extremity receives $\frac{1}{2}$.

\item Suppose that a badly colored \link\ extremity is next to a color change.
It cannot be next to two color changes by Claim~\ref{c:colchange-separated-by-one}.
Therefore by Claim~\ref{c:2bad-after-colchange},
it is also next to a well colored \link\ extremity,
which is heavy from configuration~(c).
Thus it receives a weight of $\frac{1}{2}$ from this heavy link extremity.

\item Suppose that both nearby well colored \link\ extremities
are next to a color change.
In this case by Claim~\ref{c:2good-among-4-after-colchange} and~\ref{c:3goodFor5betweenColchange}, the area contains at most two badly colored \link\ extremities.
If there are two, then both nearby well colored \link\ extremities are heavy
from configuration~(a). If these is only one, then
at least one is heavy from configuration~(a) and Claim~\ref{c:consecutive-blae}.
Therefore each receives some weight from a nearby heavy extremity.

\item Suppose finally that one nearby well colored \link\ extremity
is next to a color change. In this case, we call the extremities $x_1$, $x_2$, $x_3$ \ldots where $x_1$ is close to the color change.
By Claim~\ref{c:2good-among-4-after-colchange}, the area contains at most three
badly colored \link\ extremities.
If there is only one, then $x_2$ receives the necessary weight from $x_3$ (which is well colored).
If there are two, $x_1$ is heavy by configuration (a),
and thus covers $x_2$, and $x_3$ received weight from $x_4$.
If there are three, then we are in the special configuration of
Claim~\ref{c:2good-among-4-after-colchange}
and $x_2x_3$ form an arch. In this case, both $x_1$ and $x_5$ are heavy, $x_1$ covers $x_2$ and $x_5$ covers the two others.
\end{itemize}

\item[well colored link extremity close to two color changes:]
We note that by Claim~\ref{c:2bad-after-colchange},
this case may occur only if there are at most three \link\ extremities between the two color changes.
If there are two or three extremities,
Claims~\ref{c:colchange-separated-by-two} and \ref{c:3-extremities-between-colchange},
respectively, show that two of the \links\ are well colored.
Therefore this situation occurs only when the two color changes are separated by exactly one \link\ extremity,
in the condition of Claim~\ref{c:colchange-separated-by-one}.
However, then both color changes hit a broad link extremity,
 and the \link\ extremity is a super heavy extremity, and so it receives $\frac{7}{2}$ and gives
 twice $\frac{7}{4}$ which makes it nonnegative.

\item[well colored \link\ extremity close to a color change and possibly giving backup:]
Suppose that a \link\ extremity close to a color change needs to give a backup.
Then it is close to another \link\ extremity itself next to a color change,
and thus there are exactly two well colored \link\ extremities between the two color changes.
By Claim~\ref{c:3goodFor5betweenColchange}, there are at most four \link\ extremities between the two color changes.

If there are four extremities only two of which are well colored, we denote them $x_1$, $x_2$, $x_3$ and $x_4$, and show the two well colored extremities are both heavy, from which we infer that neither needs a backup.
By Claim~\ref{c:colchangeSeparatedBy4}, the two badly colored \link\ extremities
are those of an arch.
By Claim~\ref{c:2bad-after-colchange}, only one extremity of $x_1$ and $x_2$ is badly colored,
and also one of $x_3$ and $x_4$.
If $x_1$ (resp., $x_4$) is badly colored, then $x_2$ (resp., $x_3$) is well colored and a heavy configuration~(b).
If $x_1$ is well colored, then $x_2$ is badly colored.
In this case, $x_1$ is a heavy configuration (a), unless $x_3$ is well colored and immediately follows $x_2$.
However then $x_4$ is badly colored and immediately follows $x_3$ by Claim~\ref{c:2bad-after-colchange}, which contradicts Claim~\ref{c:consecutive-blae}.

Suppose there are three extremities and only two are well colored.
If the badly colored \link\ extremity is $x_1$ or $x_3$, then $x_2$ is a heavy configuration (b).
Otherwise, by Claim~\ref{c:consecutive-blae},
there is a \blae\ hitting nothing between $x_1$ and $x_3$, say between $x_1$ and $x_2$.
In this case, $x_1$ is a heavy configuration (a).
Hence at least one well colored \link\ extremity is heavy.
Now following Claim~\ref{c:3-extremities-between-colchange}, we consider three cases:

\textbf{Case~1. Both color changes hit a broad \link:}
In this case, each well colored link extremity gives $\frac{7}{4}$ to the
change of color next to it, and the heavy link extremity gives $\frac{1}{2}$
to the badly colored \link\ extremity in the area.
Thus, the heavy link extremity receives $\frac{5}{2}$, and gives $\frac{7}{4} + \frac{1}{2}$, and so it has sufficient weight to discharge $\frac{1}{4}$ for a backup.
The other \link\ extremity receives at least $\frac{3}{2}$, and needs to give only
$\frac{7}{4}$, and so it requires a weight of at most $\frac{1}{4}$ as a backup, which the heavy \link\ extremity can discharge to it. Therefore, in this case both extremities end up with a nonnegative weight.

\textbf{Case~2. Only one color change hits a broad \link\ but the badly colored \link\ is narrow:}
The total share received by the two \link\ extremities is $\frac{3}{2}+\frac{5}{2}$
and it should cover $2+\frac{7}{4}$ for the color changes
and $\frac{1}{4}$ for the badly colored \link\ extremity, implying that the total share it receives is what is needed. However for clarity, we include the details as follows. If the broad well colored \link\ extremity is next to the color change hitting a broad \link, then the extremity gives $\frac{7}{4}$ to the color change,
$\frac{1}{4}$ to the badly colored \link\ extremity,
and $\frac{1}{2}$ to the narrow well colored \link\ extremity as backup,
and therefore it can transfer a weight of~$2$ to the color change.
If the broad well colored \link\ extremity is next to the color change not hitting a broad \link,
then the extremity gives~$2$ to the color change,
$\frac{1}{4}$ to the badly colored \link\ extremity,
and $\frac{1}{4}$ to the narrow well colored \link\ extremity as backup,
which can thus give $\frac{7}{4}$ to the color change hitting no broad \link\ extremity.
%

\textbf{Case~3. Both extremities  are those of broad \links:}
In this case, both well colored \link\ extremity are heavy or slightly heavy.
Indeed, let $x_1$, $x_2$ and $x_3$ be the three link extremities.
By Claim~\ref{c:consecutive-blae},
there is a \blae\ hitting nothing between $x_1$ and $x_3$, say between $x_2$ and $x_3$.
We note that by Claim~\ref{c:2bad-after-colchange}, the badly colored \link\ extremity
is not $x_3$.
Thus among $x_1$ and $x_2$, there is one badly colored \link\ extremity
and a heavy \link\ extremity (from configuration~(a) or~(b)).
Moreover, since $x_3$ is the extremity of a broad \link\ and is surrounded by a color
change and a \blae\ hitting nothing, it is slightly heavy.
Thus both receive $\frac{5}{2}$ and neither need to give backup.
%

Suppose now there are only two extremities. In this case, by Claim~\ref{c:colchange-separated-by-two} we infer that both extremities are well colored.
Moreover either one is a heavy configuration~(c),
or both are broad and slightly heavy.
In the first case, the heavy extremity has enough to give~$2$ to the color change and
$\frac{1}{2}$ as backup, while in the second case, both extremities receive enough weight, and so neither needs backup.

\item[well colored \link\ extremity close to a color change but not giving backup:]
If the extremity is heavy, it may give at most~$2$ to the color change
and twice $\frac{1}{2}$ to adjacent badly colored \link\ extremities.
Therefore it ends up with a positive weight with a backup of at most $\frac{1}{2}$.
If the extremity is narrow, then it may give at most~$2$ to the adjacent color change,
and so it requires at most $\frac{1}{2}$ as a backup.

\item[well colored link extremities not adjacent to color changes :]
A heavy well colored \link\ extremity not next to a color change may have to give
$\frac{1}{2}$ to badly colored \link\ extremities at distance at most~$2$,
so it may be required to transfer to at most such extremities.
It may also give backup up to $\frac{1}{2}$.
Suppose this sums up to more than $\frac{5}{2}$,
implying that it gives backup and to two badly colored links on both sides.
Since the next well colored link extremity on each direction needs a backup,
it is close to a color change. Since there are two badly colored \link\ extremities
on each side, the next well colored links are heavy,
and so there are actually three badly colored \link\ extremities on each sides, and
both must have the properties of the special configurations in Claim~\ref{c:2good-among-4-after-colchange}.
However in that configuration, $x_5$ is immediately after $x_4$ and this property occurs on both
sides, which contradicts Claim~\ref{c:consecutive-blae}.
Therefore a heavy well colored \link\ extremity not adjacent to color changes ends up with
a nonnegative weight.

We next consider a light well colored \link\ extremity $x_0$.
Suppose again that it ends up with a negative weight, that is, it gives weight
to badly colored \link\ extremities close by on each side $x_{-1}$ and $x_1$,
and gives backup on each direction to the closest well colored \link\ extremity.
We note that by Claim~\ref{c:consecutive-blae},
there is a \blae\ hitting nothing between $x_{-1}$ and $x_1$, say between
$x_{-1}$ and $x_0$. Thus the closest well colored \link\ extremity on that
direction $x_{-3}$ is heavy, and if it requires backup, it must be next to
two other badly colored \link\ extremities $x_{-4}$ and $x_{-2}$.
However in this case, there are three bad \link\ extremities following the change of color
before $x_{-4}$, meaning that this correspond to the special configuration of
Claim~\ref{c:2good-among-4-after-colchange},
and thus $x_0$ is immediately after $x_{-1}$,
which contradicts our assumption that a \blae\ hitting nothing was between $x_{-1}$
and $x_0$.
\end{description}

The above arguments and discussion show that the discharging procedure is such that upon completion of the discharging no negative weight remains in the graph. Therefore, at least one of the colorings Amber and Blue gives a positive score.

\subsection{Red edges}
\label{S:red-edges}

By our results in Section~\ref{S:max-green-black-paths}, there are no \grees\ in the colored multigraph.
In this section, we show that there are no \redes.
Recall that by Claim~\ref{c:red-matching} the \redes\ form a matching in the colored multigraph~$M_G$.
Further recall that by Claim~\ref{c:no-long-green},
there are no long  \redes\ in the colored multigraph $M_G$.
Thus by our earlier observations, there are no marked vertices.

\begin{claim}
\label{no-adj-vertices-with-red-edges}
No adjacent vertices are incident to distinct \redes.
\end{claim}
\proof Suppose that there are adjacent vertices $v$ and $w$ incident with distinct \redes.
Let $vy$ and $wz$ be \redes\ incident with $v$ and $w$, respectively, where the edge $vw$ is a \blae.
This structure is illustrated in Figure~\ref{f:red-edges-adj}.
Let $H$ be obtained from $G$ by removing the vertices $v$ and $w$,
and removing the vertices $(vy)_1$, $(vy)_2$, $(wz)_1$, and $(wz)_2$.
We note that $\w(G) = \w(H) + 2 \times 4 + 4 \times 5 - 4 \times 1 = 24$.
Every MD-set of $H$ can be extended to a MD-set of $G$
by adding to it the vertices $(v,y)_1$ and $(w,z)_1$.
Thus, $\mdom(G) \le \mdom(H) + \alpha$ and $\w(G) \ge \w(H) + \beta$,
where $\alpha = 2$ and $\beta = 24$, contradicting Fact~\ref{new:fact1}.~\smallqed

\begin{figure}[htb]
\begin{center}
    \input{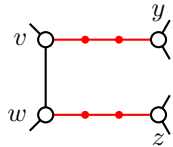}
    \caption{Adjacent vertices incident with distinct \redes}
    \label{f:red-edges-adj}
    \end{center}
\end{figure}

\begin{claim}
\label{c:no-red-edges}
There is no \rede.
\end{claim}
\proof Suppose that there is a \rede\ $uv$.
Let $u_1$ and $u_2$ be the neighbors of $u$ not on the \rede\ $uv$,
and let $v_1$ and $v_2$ be the neighbors of $v$ not on the \rede\ $uv$.
Thus, $uu_1$, $uu_2$, $vv_1$, and $vv_2$ are all \blaes.
By Claim~\ref{no-adj-vertices-with-red-edges},
none of vertices $u_1$, $u_2$, $v_1$ and $v_2$ is incident with a \rede.

Suppose firstly that $u_1$ or $u_2$ is adjacent to $v_1$ or $v_2$.
Renaming vertices if necessary, we may assume that $u_1v_1$ is an edge.
As observed earlier, $u_1v_1$ is a \blae.
In this case, let $H$ be obtained from $G$ by removing the vertices $u$, $v$, $u_1$ and $v_1$,
removing the vertices $(uv)_1$ and $(uv)_2$, and marking the vertices $u_2$ and $v_2$.
We note that $\w(G) = \w(H) + 4 \times 4 + 2 \times 5 - 2 \times 1 = 24$.
Every MD-set of $H$ can be extended to a MD-set of $G$
by adding to it the vertices $u$ and $v$.
Thus, $\mdom(G) \le \mdom(H) + \alpha$ and $\w(G) \ge \w(H) + \beta$,
where $\alpha = 2$ and $\beta = 24$, contradicting Fact~\ref{new:fact1}.

Hence, the set $\{u_1,u_2,v_1,v_2\}$ is an independent set.
Let $H$ be obtained from $G$ by removing the vertices $u$, $v$, $u_1$, $u_2$ and $v_1$,
removing the vertices $(uv)_1$ and $(uv)_2$, and marking the vertex~$v_2$.
Since the graph $G$ contains no $4$- or $7$-cycles, every vertex of $H$ is adjacent in $G$
to at most one vertex in $V(G) \setminus V(H)$,
implying that $\w(G) = \w(H) + 5 \times 4 + 2 \times 5 - 6 \times 1 = 24$.
Every MD-set of $H$ can be extended to a MD-set of $G$
by adding to it the vertices $u$ and $v$.
Thus, $\mdom(G) \le \mdom(H) + \alpha$ and $\w(G) \ge \w(H) + \beta$,
where $\alpha = 2$ and $\beta = 24$, contradicting Fact~\ref{new:fact1}.~\smallqed

\begin{figure}[htb]
\begin{center}
    \input{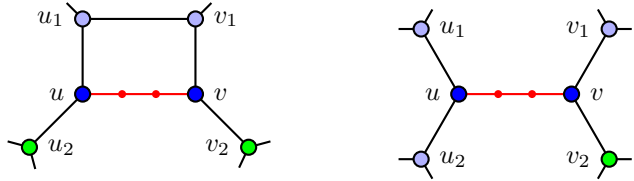}
    \caption{The two cases of the proof of Claim~\ref{c:no-red-edges}}
    \label{f:no-red-edge}
    \end{center}
\end{figure}

\subsection{The cubic graph $G$} \label{S:G-cubic}

By our earlier observations, there are no \grees\  and no \redes\ in the colored multigraph,
implying that the graph $G$ is a cubic graph and contains no marked vertex.
Since there are no marked vertices in $G$, we note that $\mdom(G) = \gamma(G)$ and $\w(G) = \nw(G)$.
As a consequence of Fact~\ref{fact1}, we may infer the following fact.

\begin{fact}
\label{new:fact1}
If $H$ is a proper subgraph of $G$, and $\alpha$ and $\beta$ are two integers
such that $\gamma(G) \le \gamma(H) + \alpha$ and $\nw(G) \ge \nw(H) + \beta$,
then $12\alpha > \beta$.
\end{fact}

Recall that by supposition, the girth of $G$ is at least~$6$ with no cycles of length~$7$ or~$8$.
Let $g$ denote the girth of $G$, and so $g = 6$ or $g \ge 9$.
Let $C \colon v_1 v_2 \ldots v_g v_1$ be a cycle in $G$ of smallest length, namely~$g$.
Let $u_i$ be the neighbor of $v_i$ not on $C$ for all $i \in [g]$.
Since the girth of $G$ is equal to~$g \ge 6$, the vertices $u_1, u_2, \ldots, u_g$ are distinct
and do not belong to the cycle $C$. We define the \emph{boundary} of the cycle $C$ in $G$, denoted $\partial(C)$, to be the set of all vertices in $G$ that do not belong to the cycle $C$ but are adjacent in $G$ to at least one vertex on the cycle. Thus, $\partial(C) = \{u_1,u_2, \ldots, u_g\}$. We show firstly that the graph $G$ contains no $6$-cycle.

\begin{claim}
\label{girth:6}
The graph $G$ contains no $6$-cycle.
\end{claim}
\proof
Suppose, to the contrary, that the girth of $G$ is equal to~$6$. Thus, $g = 6$ and in this case $C \colon v_1v_2 \ldots v_6v_1$ is a $6$-cycle in $G$. Suppose that the set $\partial(C)$ is an independent set in $G$. Since the graph $G$ contains no~$7$-cycle, the vertices $u_1$ and $u_4$ have no common neighbor. We now consider the graph $H - (V(C) \cup \{u_1,u_4\})$. By our earlier observations, the graph $H$ has minimum degree~$2$. Further, $n_3(H) = n_3(G) - 8$ and $n_2(H) = 8$. Thus, $\nw(G) = \nw(H) + 8 \times 4 - 8 = 24$. Every $\gamma$-set of $H$ can be extended to a $\gamma$-set of $G$ by adding to it the vertices $v_1$ and $v_4$, and so $\gamma(G) \le \gamma(H) + 2$. Hence, $\gamma(G) \le \gamma(H) + \alpha$ and $\nw(G) \ge \nw(H) + \beta$, where $\alpha = 2$ and $\beta = 24$, contradicting Fact~\ref{new:fact1}. Hence, the set $\partial(C)$ is not an independent set.

Since the graph $G$ contains no $4$-cycle and no $5$-cycle, the only possible edges in the subgraph $G[\partial(C)]$ induced by $\partial(C)$ are the edges $u_iu_{i+3}$ where $i \in [3]$. Renaming vertices if necessary, we may assume that $u_1u_4$ is an edge of $G$. We note that $N_G[v_1] \cup N_G[v_4] = V(C) \cup \{u_1,u_4\}$. In this case, we consider the graph $H = G - (V(C) \cup \{u_1,u_4\})$. Since the girth of $G$ is equal to~$6$, every vertex in $H$ is adjacent in $G$ to at most one vertex in $V(C) \cup \{u_1,u_4\}$, implying that the graph $H$ has minimum degree~$2$. Further, $n_3(H) = n_3(G) - 8$ and $n_2(H) = 6$. Thus, $\nw(G) = \nw(H) + 8 \times 4 - 6 = 26$. Every $\gamma$-set of $H$ can be extended to a $\gamma$-set of $G$ by adding to it the vertices $v_1$ and $v_4$, and so $\gamma(G) \le \gamma(H) + 2$. Hence, $\gamma(G) \le \gamma(H) + \alpha$ and $\nw(G) \ge \nw(H) + \beta$, where $\alpha = 2$ and $\beta = 26$, contradicting Fact~\ref{new:fact1}.~\smallqed

\medskip
By Claim~\ref{girth:6}, the graph $G$ contains no $6$-cycle. By supposition, the girth of $G$ is at least~$6$ with no cycles of length~$7$ or~$8$. Therefore, the girth of $G$ is at least~$9$, that is, $g \ge 9$. This implies that every vertex in $V(G) \setminus (V(C) \cup \partial(C))$ is adjacent in $G$ to at most one vertex in the boundary, $\partial(C)$, of $C$.

\begin{claim}
\label{girth:not0mod3}
$g \not\equiv 0 \, (\modo \, 3)$.
\end{claim}
\proof Suppose, to the contrary, that $g \equiv 0 \, (\modo \, 3)$. Thus, $g = 3k$ for some integer $k \ge 3$. Hence, $|V(C)| = g = 3k$. In this case, we let
\[
S = \bigcup_{i=1}^{k} \{v_{3i-1}\} \hspace*{0.5cm} \mbox{and} \hspace*{0.5cm} U = \bigcup_{i=1}^{k} \{u_{3i-1}\}.
\]

We note that $U \subset \partial(C)$ and $|S| = |U| = k$. The subgraph of $G$ that contains the set $V(C) \cup U$ is illustrated in Figure~\ref{f:0mod3}. In this case, we consider the graph $H = G - (V(C) \cup U)$ of order~$n - 4k$. Every vertex in $H$ is adjacent in $G$ to at most one vertex in $V(C) \cup U$. Thus, the graph $H$ has minimum degree~$2$. We note that each vertex in $U$ is adjacent in $G$ to two vertices in $V(H)$, and each vertex in $V(C) \setminus S$ is adjacent in $G$ to exactly one vertex in $V(H)$, while each vertex in $S$ has no neighbor in $G$ that belongs to $V(H)$. These observations imply that $\nw(G) = \nw(H) + 4 \times 4k - 4k = 12k$. Every $\gamma$-set of $H$ can be extended to a $\gamma$-set of $G$ by adding to it the set $S$, and so $\gamma(G) \le \gamma(H) + |S| = \gamma(H) + k$. Hence, $\gamma(G) \le \gamma(H) + \alpha$ and $\nw(G) \ge \nw(H) + \beta$, where $\alpha = k$ and $\beta = 12k$, contradicting Fact~\ref{new:fact1}.~\smallqed

\tikzstyle{every node}=[circle, draw, fill=black!0, inner sep=0pt, minimum width=.2cm]

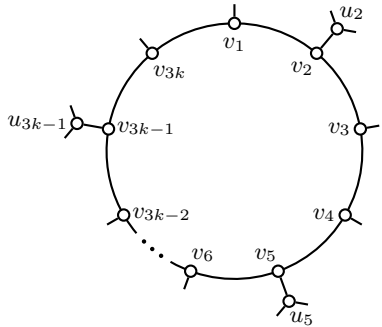
\begin{figure}[htb]
\begin{center}
\begin{tikzpicture}
	[thick,scale=.85,
		novertex/.style={circle,draw,inner sep=0pt,minimum size=0mm},
		dot/.style={circle,draw,fill=black,inner sep=0pt,minimum size=.3mm},
		vertex/.style={circle,draw,inner sep=0pt,minimum size=1.5mm},
		vertexlabel/.style={circle,draw=white,inner sep=0pt,minimum size=0mm}]
	\coordinate (A) at (0,0);
	\def \k {9}
	\draw (A)
	{
		node at +(90-0*360/\k:2) [vertex, label=-90-0*360/\k:\small{$v_1$}]  (u1){}
		node at +(90-1*360/\k:2) [vertex, label=-90-1*360/\k:\small{$v_2$}]  (u2){}
		node at +(90-2*360/\k:2) [vertex, label=180:\small{$v_3$}]  (u3){}
		node at +(90-3*360/\k:2) [vertex, label=180:\small{$v_4$}]  (u4){}
		node at +(90-4*360/\k:2) [vertex, label=-90-4*360/\k:\small{$v_5$}]  (u5){}
		node at +(90-5*360/\k:2) [vertex, label=-90-5*360/\k:\small{$v_6$}]  (u6){}
		node at +(90-5.5*360/\k-5:2) [dot]  {}
		node at +(90-5.5*360/\k:2) [dot]  {}
		node at +(90-5.5*360/\k+5:2) [dot]  {}
		node at +(90-6*360/\k:2) [vertex, label=0:\small{$v_{3k-2}$}]  (u7){}
		node at +(90-7*360/\k:2) [vertex, label=0:\small{$v_{3k-1}$}]  (u8){}
		node at +(90-8*360/\k:2) [vertex, label=-90-8*360/\k:\small{$v_{3k}$}]  (u9){}
		(A) +(90-5.25*360/\k:2) arc (90-5.25*360/\k:360+90-5.75*360/\k:2)
		(u1)--+(90-0*360/\k:.3)
		(u2) node at +(90-1*360/\k:.5) [vertex, label=90-1*360/\k:\small{$u_2$}]  (v2){}
		(u2)--(v2)--+(90-1*360/\k+60:.3) (v2)--+(90-1*360/\k-60:.3)
		(u3)--+(90-2*360/\k:.3)
		(u4)--+(90-3*360/\k:.3)
		(u5) node at +(90-4*360/\k:.5) [vertex, label=90-4*360/\k:\small{$u_5$}]  (v5){}
		(u5)--(v5)--+(90-4*360/\k+60:.3) (v5)--+(90-4*360/\k-60:.3)
		(u6)--+(90-5*360/\k:.3)
		(u7)--+(90-6*360/\k:.3)
		(u8) node at +(90-7*360/\k:.5) [vertex, label=90-7*360/\k+8:\small{$u_{3k-1}$}]  (v8){}
		(u8)--(v8)--+(90-7*360/\k+60:.3) (v8)--+(90-7*360/\k-60:.3)
		(u9)--+(90-8*360/\k:.3)
	};
\end{tikzpicture}
\end{center}
\vskip -0.6 cm \caption{The case when $g \equiv 0 \, (\modo \, 3)$ and $g \ge 9$}
\label{f:0mod3}
\end{figure}

By Claim~\ref{girth:not0mod3}, $g \not\equiv 0 \, (\modo \, 3)$. Thus by our earlier observations, $g \ge 10$. Let $N_G(u_i) = \{v_i,x_i,w_i\}$ for all $i \in [g]$. By the girth condition, the sets $N_G(u_i)$ and $N_G(u_j)$ are vertex disjoint for $i,j \in [g]$ and $i \ne j$, that is, $N_G(u_i) \cap N_G(u_j) = \emptyset$ for $1 \le i < j \le g$.

\begin{claim}
\label{girth:not2mod3}
$g \not\equiv 2 \, (\modo \, 3)$.
\end{claim}
\proof Suppose, to the contrary, that $g \equiv 2 \, (\modo \, 3)$.  Thus, $g = 3k+2$ for some integer $k \ge 3$. Hence, $|V(C)| = g = 3k+2$. Recall that $N_G(u_1) = \{v_1,x_1,w_1\}$ and $N_G(u_2) = \{v_2,x_2,w_2\}$. In this case, we let
\[
S = \{u_1,u_2\} \cup \left( \bigcup_{i=1}^{k} \{v_{3i+1}\} \right) \hspace*{0.5cm} \mbox{and} \hspace*{0.5cm} U = \bigcup_{i=1}^{k} \{u_{3i+1}\}.
\]

We note that $|S| = k+2$ and $|U| = k$. Let $U_{12} = \{u_1,u_2,x_1,x_2,w_1,w_2\}$. The girth condition implies that $U \cap U_{12} = \emptyset$. The subgraph of $G$ that contains the set $V(C) \cup U \cup U_{12}$ is illustrated in Figure~\ref{f:2mod3}. We now consider the graph $H = G - (V(C) \cup U \cup U_{12})$ of order~$n - (3k+2) - k - 6 = n - 4k - 8$. Since the girth $g = 3k+2 \ge 11$, every vertex in $H$ is adjacent in $G$ to at most one vertex in $V(C) \cup U \cup U_{12}$. Thus, the graph $H$ has minimum degree~$2$. We note that each vertex in $U \cup \{x_1,x_2,w_1,w_2\}$ is adjacent in $G$ to two vertices in $V(H)$, each vertex in $V(C) \setminus (S \cup \{v_1,v_2\})$ is adjacent in $G$ to exactly one vertex in $V(H)$, while each vertex in $S \cup \{u_1,u_2,v_1,v_2\}$ has no neighbor in $G$ that belongs to $V(H)$. These observations imply that $\nw(G) = \nw(H) + 4 \times (4k+8) - (4k + 8) = 12(k + 2)$. Every $\gamma$-set of $H$ can be extended to a $\gamma$-set of $G$ by adding to it the set $S$, and so $\gamma(G) \le \gamma(H) + |S| = \gamma(H) + k + 2$. Hence, $\gamma(G) \le \gamma(H) + \alpha$ and $\nw(G) \ge \nw(H) + \beta$, where $\alpha = k+2$ and $\beta = 12(k+2)$, contradicting Fact~\ref{new:fact1}.~\smallqed

\begin{figure}[htb]
\begin{center}
\begin{tikzpicture}
	[thick,scale=.85,
		novertex/.style={circle,draw,inner sep=0pt,minimum size=0mm},
		dot/.style={circle,draw,fill=black,inner sep=0pt,minimum size=.3mm},
		vertex/.style={circle,draw,inner sep=0pt,minimum size=1.5mm},
		vertexlabel/.style={circle,draw=white,inner sep=0pt,minimum size=0mm}]
	\coordinate (C) at (12,0);
	\def \k {8}
	\draw (C)
	{
		node at +(90-0*360/\k:2) [vertex, label=-90-0*360/\k:\small{$v_3$}]  (u1){}
		node at +(90-1*360/\k:2) [vertex, label=-90-1*360/\k:\small{$v_4$}]  (u2){}
		node at +(90-2*360/\k:2) [vertex, label=-90-2*360/\k:\small{$v_5$}]  (u3){}
		node at +(90-2.5*360/\k-5:2) [dot]  {}
		node at +(90-2.5*360/\k:2) [dot]  {}
		node at +(90-2.5*360/\k+5:2) [dot]  {}
		node at +(90-3*360/\k:2) [vertex, label=-90-3*360/\k:\small{$v_{3k}$}]  (u4){}
		node at +(90-4*360/\k:2) [vertex, label={[label distance=-5]90:\small{$v_{3k+1}$}}]  (u5){}
		node at +(90-5*360/\k:2) [vertex, label=-90-5*360/\k:\small{$v_{3k+2}$}]  (u6){}
		node at +(90-6*360/\k:2) [vertex, label=-90-6*360/\k:\small{$v_1$}]  (u7){}
		node at +(90-7*360/\k:2) [vertex, label=-90-7*360/\k:\small{$v_2$}]  (u8){}
		(u2) node at +(90-1*360/\k:.5) [vertex, label=90-1*360/\k:\small{$u_4$}]  (v2){}
		(u5) node at +(90-4*360/\k:.5) [vertex, label=90-4*360/\k:\small{$u_{3k+1}$}]  (v5){}
		(u7) node at +(90-6*360/\k:.5) [vertex, label={[label distance=-1]180:\scriptsize{$u_1$}}](v7){}
		(v7) node at +(90-6*360/\k+60:.7) [vertex, label={[label distance=-1]0:\scriptsize{$x_1$}}](v'7){}
		(v7) node at +(90-6*360/\k-60:.7) [vertex, label={[label distance=-1]0:\scriptsize{$w_1$}}](v''7){}
		(u8) node at +(90-7*360/\k:.5) [vertex, label=90-7*360/\k:\scriptsize{$u_2$}]  (v8){}
		(v8) node at +(90-7*360/\k+60:.7) [vertex, label=-90-7*360/\k:\scriptsize{$x_2$}]  (v'8){}
		(v8) node at +(90-7*360/\k-60:.7) [vertex, label=-90-7*360/\k:\scriptsize{$w_2$}]  (v''8){}
		(C) +(90-2.25*360/\k:2) arc (90-2.25*360/\k:360+90-2.75*360/\k:2)
		(u1)--+(90-0*360/\k:.3)
		(u2)--(v2)--+(90-1*360/\k+60:.3) (v2)--+(90-1*360/\k-60:.3)
		(u3)--+(90-2*360/\k:.3)
		(u4)--+(90-3*360/\k:.3)
		(u5)--(v5)--+(90-4*360/\k+60:.3) (v5)--+(90-4*360/\k-60:.3)
		(u6)--+(90-5*360/\k:.3)
		(u7)--(v7) (v7)--(v'7) (v7)--(v''7)
		(v'7)--+(90-6*360/\k+60+30:.3) (v'7)--+(90-6*360/\k+60-30:.3)
		(v''7)--+(90-6*360/\k-60+30:.3) (v''7)--+(90-6*360/\k-60-30:.3)
		(u8)--(v8) (v8)--(v'8) (v8)--(v''8)
		(v'8)--+(90-7*360/\k+60+30:.3) (v'8)--+(90-7*360/\k+60-30:.3)
		(v''8)--+(90-7*360/\k-60+30:.3) (v''8)--+(90-7*360/\k-60-30:.3)
	};
\end{tikzpicture}\end{center}
\vskip -0.6 cm \caption{The case when $g \equiv 2 \, (\modo \, 3)$ and $g \ge 11$}
\label{f:2mod3}
\end{figure}
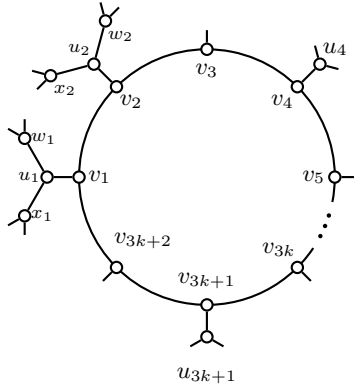

\medskip
By Claim~\ref{girth:not0mod3} and~\ref{girth:not2mod3}, we have $g \equiv 1 \, (\modo \, 3)$. Thus, $g = 3k+1$ for some integer $k \ge 3$.

\begin{claim}
\label{girth:k3}
$k = 3$.
\end{claim}
\proof Suppose, to the contrary, that $k \ge 4$, and so $g = 3k+1 \ge 13$. Hence, $|V(C)| = g = 3k+1$. Recall that $N_G(u_1) = \{v_1,x_1,w_1\}$. In this case, we let
\[
S = \{u_1\} \cup \left( \bigcup_{i=1}^{k} \{v_{3i}\} \right) \hspace*{0.5cm} \mbox{and} \hspace*{0.5cm} U = \bigcup_{i=1}^{k} \{u_{3i}\}.
\]

We note that $|S| = k+1$ and $|U| = k$. Let $U_{1} = \{u_1,x_1,w_1\}$. The girth condition implies that $U \cap U_{1} = \emptyset$. The subgraph of $G$ that contains the set $V(C) \cup U \cup U_{1}$ is illustrated in Figure~\ref{f:1mod3}. We now consider the graph $H = G - (V(C) \cup U \cup U_{1})$ of order~$n - (3k+1) - k - 3 = n - 4k - 4$. Since the girth $g = 3k+1 \ge 13$, every vertex in $H$ is adjacent in $G$ to at most one vertex in $V(C) \cup U \cup U_{1}$. Thus, the graph $H$ has minimum degree~$2$. We note that each vertex in $U \cup \{x_1,w_1\}$ is adjacent in $G$ to two vertices in $V(H)$, each vertex in $V(C) \setminus (S \cup \{v_1\})$ is adjacent in $G$ to exactly one vertex in $V(H)$, while each vertex in $S \cup \{u_1,v_1\}$ has no neighbor in $G$ that belongs to $V(H)$. These observations imply that $\nw(G) = \nw(H) + 4 \times (4k+4) - (4k + 4) = 12(k + 1)$. Every $\gamma$-set of $H$ can be extended to a $\gamma$-set of $G$ by adding to it the set $S$, and so $\gamma(G) \le \gamma(H) + |S| = \gamma(H) + k + 1$. Hence, $\gamma(G) \le \gamma(H) + \alpha$ and $\nw(G) \ge \nw(H) + \beta$, where $\alpha = k+1$ and $\beta = 12(k+1)$, contradicting Fact~\ref{new:fact1}.~\smallqed

\begin{figure}[htb]
\begin{center}
\begin{tikzpicture}
	[thick,scale=.85,
		novertex/.style={circle,draw,inner sep=0pt,minimum size=0mm},
		dot/.style={circle,draw,fill=black,inner sep=0pt,minimum size=.3mm},
		vertex/.style={circle,draw,inner sep=0pt,minimum size=1.5mm},
		vertexlabel/.style={circle,draw=white,inner sep=0pt,minimum size=0mm}]
	\coordinate (B) at (0,0);
	\def \k {10}
	\draw (B)
	{
		node at +(90-0*360/\k:2) [vertex, label=-90-0*360/\k:\small{$v_2$}]  (u1){}
		node at +(90-1*360/\k:2) [vertex, label=-90-1*360/\k:\small{$v_3$}]  (u2){}
				node at +(90-2*360/\k:2) [vertex, label=180:\small{$v_4$}]  (u3){}
		node at +(90-3*360/\k:2) [vertex, label=180:\small{$v_5$}]  (u4){}
		node at +(90-4*360/\k:2) [vertex, label=-90-4*360/\k:\small{$v_6$}]  (u5){}
		node at +(90-5*360/\k:2) [vertex, label=-90-5*360/\k:\small{$v_7$}]  (u6){}
		node at +(90-5.5*360/\k-5:2) [dot]  {}
		node at +(90-5.5*360/\k:2) [dot]  {}
		node at +(90-5.5*360/\k+5:2) [dot]  {}
		node at +(90-6*360/\k:2) [vertex, label=30:\small{$v_{3k-1}$}]  (u7){}
		node at +(90-7*360/\k:2) [vertex, label=0:\small{$v_{3k}$}]  (u8){}
		node at +(90-8*360/\k:2) [vertex, label=-0:\small{$v_{3k+1}$}]  (u9){}
		node at +(90-9*360/\k:2) [vertex, label=-90-9*360/\k:\small{$v_1$}]  (u10){}
		(u2) node at +(90-1*360/\k:.5) [vertex, label=90-1*360/\k:\small{$u_3$}]  (v2){}
		(u5) node at +(90-4*360/\k:.5) [vertex, label=90-4*360/\k:\small{$u_6$}]  (v5){}
		(u8) node at +(90-7*360/\k:.5) [vertex, label=90-7*360/\k:\footnotesize{$u_{3k}$}]  (v8){}
		(u10) node at +(90-9*360/\k:.5) [vertex, label=90-9*360/\k:\scriptsize{$u_{1}$}]  (v10){}
		(v10) node at +(90-9*360/\k+60:.7) [vertex, label=-90-9*360/\k:\scriptsize{$x_1$}]  (v'10){}
		(v10) node at +(90-9*360/\k-60:.7) [vertex, label=-90-9*360/\k:\scriptsize{$w_1$}]  (v''10){}
		(B) +(90-5.25*360/\k:2) arc (90-5.25*360/\k:360+90-5.75*360/\k:2)
		(u1)--+(90-0*360/\k:.3)
		(u2)--(v2)--+(90-1*360/\k+60:.3) (v2)--+(90-1*360/\k-60:.3)
		(u3)--+(90-2*360/\k:.3)
		(u4)--+(90-3*360/\k:.3)
		(u5)--(v5)--+(90-4*360/\k+60:.3) (v5)--+(90-4*360/\k-60:.3)
		(u6)--+(90-5*360/\k:.3)
		(u7)--+(90-6*360/\k:.3)
		(u8)--(v8)--+(90-7*360/\k+60:.3) (v8)--+(90-7*360/\k-60:.3)
		(u9)--+(90-8*360/\k:.3)
		(u10)--(v10) (v10)--(v'10) (v10)--(v''10)
		(v'10)--+(90-9*360/\k+60+30:.3) (v'10)--+(90-9*360/\k+60-30:.3)
		(v''10)--+(90-9*360/\k-60+30:.3) (v''10)--+(90-9*360/\k-60-30:.3)
	};
\end{tikzpicture}
\end{center}
\vskip -0.6 cm \caption{The case when $g \equiv 1 \, (\modo \, 3)$ and $g \ge 13$}
\label{f:1mod3}
\end{figure}
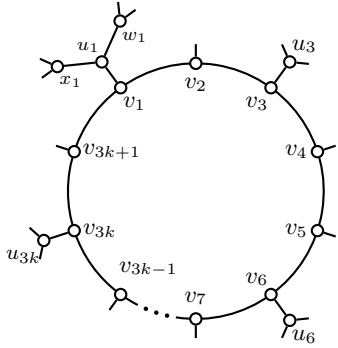

By Claim~\ref{girth:k3}, the girth of $G$ is~$g = 10$. Hence, $|V(C)| = 10$. We proceed now as in the proof of Claim~\ref{girth:k3}. In this case when $g = 10$, we have $S = \{u_1,v_3,v_6,v_9\}$, $U = \{u_3,u_6,u_9\}$, and $U_{1} = \{u_1,x_1,w_1\}$. The girth condition once again implies that $U \cap U_{1} = \emptyset$. The subgraph of $G$ that contains the set $V(C) \cup U \cup U_{1}$ is illustrated in Figure~\ref{f:g9}. We now consider the graph $H = G - (V(C) \cup U \cup U_{1})$ of order~$n(H) = n - 16$. If every vertex in $H$ is adjacent in $G$ to at most one vertex in $V(C) \cup U \cup U_{1}$, then proceeding exactly as in the proof of Claim~\ref{girth:k3} we produce a contradiction.

\begin{figure}[htb]
\begin{center}
\begin{tikzpicture}
	[thick,scale=.85,
		novertex/.style={circle,draw,inner sep=0pt,minimum size=0mm},
		dot/.style={circle,draw,fill=black,inner sep=0pt,minimum size=.3mm},
		vertex/.style={circle,draw,inner sep=0pt,minimum size=1.5mm},
		vertexlabel/.style={circle,draw=white,inner sep=0pt,minimum size=0mm}]
	\coordinate (B) at (0,0);
	\def \k {10}
	\draw (B)
	{
		node at +(90-0*360/\k:2) [vertex, label=-90-0*360/\k:\small{$v_2$}]  (u1){}
		node at +(90-1*360/\k:2) [vertex, label=-90-1*360/\k:\small{$v_3$}]  (u2){}
				node at +(90-2*360/\k:2) [vertex, label=180:\small{$v_4$}]  (u3){}
		node at +(90-3*360/\k:2) [vertex, label=180:\small{$v_5$}]  (u4){}
		node at +(90-4*360/\k:2) [vertex, label=-90-4*360/\k:\small{$v_6$}]  (u5){}
		node at +(90-5*360/\k:2) [vertex, label=-90-5*360/\k:\small{$v_7$}]  (u6){}
		node at +(90-6*360/\k:2) [vertex, label=30:\small{$v_{8}$}]  (u7){}
		node at +(90-7*360/\k:2) [vertex, label=0:\small{$v_{9}$}]  (u8){}
		node at +(90-8*360/\k:2) [vertex, label=-0:\small{$v_{10}$}]  (u9){}
		node at +(90-9*360/\k:2) [vertex, label=-90-9*360/\k:\small{$v_1$}]  (u10){}
		(u2) node at +(90-1*360/\k:.5) [vertex, label=90-1*360/\k:\small{$u_3$}]  (v2){}
		(u5) node at +(90-4*360/\k:.5) [vertex, label=90-4*360/\k:\small{$u_6$}]  (v5){}
		(u8) node at +(90-7*360/\k:.5) [vertex, label=90-7*360/\k:\footnotesize{$u_{9}$}]  (v8){}
		(u10) node at +(90-9*360/\k:.5) [vertex, label=90-9*360/\k:\scriptsize{$u_{1}$}]  (v10){}
		(v10) node at +(90-9*360/\k+60:.7) [vertex, label=-90-9*360/\k:\scriptsize{$x_1$}]  (v'10){}
		(v10) node at +(90-9*360/\k-60:.7) [vertex, label=-90-9*360/\k:\scriptsize{$w_1$}]  (v''10){}
		(B) +(90-5.25*360/\k:2) arc (90-5.25*360/\k:360+90-5.25*360/\k:2)
        (u1)--+(90-0*360/\k:.3)
		(u2)--(v2)--+(90-1*360/\k+60:.3) (v2)--+(90-1*360/\k-60:.3)
		(u3)--+(90-2*360/\k:.3)
		(u4)--+(90-3*360/\k:.3)
		(u5)--(v5)--+(90-4*360/\k+60:.3) (v5)--+(90-4*360/\k-60:.3)
		(u6)--+(90-5*360/\k:.3)
		(u7)--+(90-6*360/\k:.3)
		(u8)--(v8)--+(90-7*360/\k+60:.3) (v8)--+(90-7*360/\k-60:.3)
		(u9)--+(90-8*360/\k:.3)
		(u10)--(v10) (v10)--(v'10) (v10)--(v''10)
		(v'10)--+(90-9*360/\k+60+30:.3) (v'10)--+(90-9*360/\k+60-30:.3)
		(v''10)--+(90-9*360/\k-60+30:.3) (v''10)--+(90-9*360/\k-60-30:.3)
	};
\end{tikzpicture}
\end{center}
\vskip -0.6 cm \caption{The case when $g = 10$}
\label{f:g9}
\end{figure}
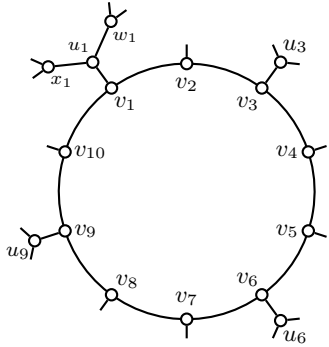

Hence, we may assume that there is a vertex $z$ in $H$ that is adjacent in $G$ to two or more vertices in $V(C) \cup U \cup U_{1}$. Since the girth of $G$ is~$g = 10$, the vertex $z$ is adjacent in $G$ to exactly two vertices in $V(C) \cup U \cup U_{1}$, namely to one of the neighbors of $u_1$ different from $v_1$ and to the vertex $u_6$. Thus, $z$ is adjacent to either $x_1$ or $w_1$, and $z \in \{x_6,w_6\}$. Renaming the vertices $x_1, w_1, x_6$ and $w_6$ if necessary, we may assume that $z = w_6$ and that $z$ is adjacent to $w_1$. Thus, $Q_1 \colon v_1u_1w_1w_6u_6v_6$ is a path in $G$. By symmetry, renaming the vertices $u_i$ and $w_i$ if necessary, we may assume that $Q_i \colon v_iu_iw_iw_{5+i}u_{5+i}v_{5+i}$ is a path in $G$ for all $i \in [5]$. We now consider the cycle $C$ and the paths $Q_1$ and $Q_2$, and let
\[
S = \{u_1,u_2,u_6,u_7,v_4,v_9\} \hspace*{0.5cm} \mbox{and} \hspace*{0.5cm} Q = V(Q_1) \cup V(Q_2) \cup \{u_4,u_9,x_1,x_2,x_6,x_7\}.
\]

The subgraph of $G$ that contains the set $V(C) \cup Q$ is illustrated in Figure~\ref{f:cubic-10cycle}.

\begin{figure}[htb]
\begin{center}
    \input{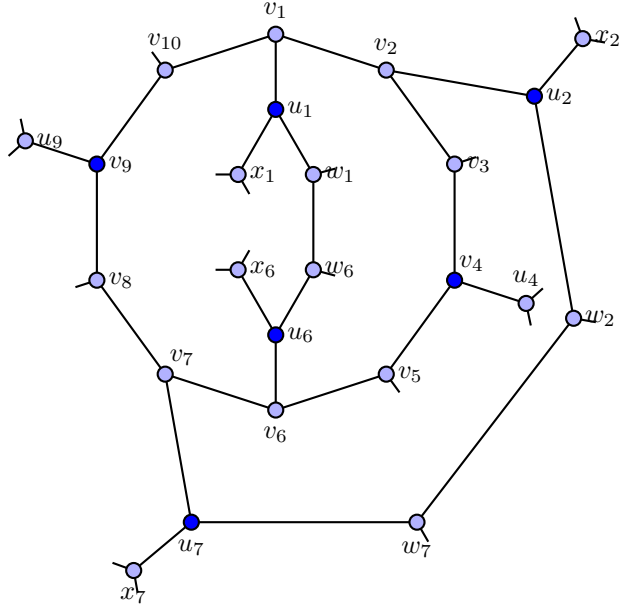}
    \caption{A subgraph of $G$ induced by the set $V(C) \cup Q$}
    \label{f:cubic-10cycle}
    \end{center}
\end{figure}

We now consider the graph $H = G - (V(C) \cup Q)$ of order~$n(H) = n - 24$. Since the girth of $G$ is $g = 10$, every vertex in $H$ is adjacent in $G$ to at most one vertex in $V(C) \cup Q$, except possibly for two vertices, namely a vertex adjacent to both $x_1$ and $x_7$, and a vertex adjacent to both $x_2$ and $x_6$. Thus every vertex in $H$ has degree at least~$2$, except possibly for two vertices which have degree~$1$ in $H$. If $H$ has no vertex of degree~$1$, then $n_1(H) = 0$, $n_2(H) = 20$, and $\nw(G) = \nw(H) + 4 \times 24 - 1 \times 20 = \nw(H) + 6 \times 12 + 4$. If $H$ has one vertex of degree~$1$, then $n_1(H) = 1$, $n_2(H) = 18$, and $\nw(G) = \nw(H) + 4 \times 24 - 1 \times 18 - 4 \times 1 = \nw(H) + 6 \times 12 + 2$. If $H$ has two vertices of degree~$1$, then $n_1(H) = 2$, $n_2(H) = 16$, and $\nw(G) = \nw(H) + 4 \times 24 - 1 \times 16 - 4 \times 2 = \nw(H) + 6 \times 12$. In all cases, $\nw(G) \ge \nw(H) + 6 \times 12$. Every $\gamma$-set of $H$ can be extended to a $\gamma$-set of $G$ by adding to it the set $S$, and so $\gamma(G) \le \gamma(H) + |S| = \gamma(H) + 6$. Hence, $\gamma(G) \le \gamma(H) + \alpha$ and $\nw(G) \ge \nw(H) + \beta$, where $\alpha = 6$ and $\beta = 6 \times 12$, contradicting Fact~\ref{new:fact1}. This completes the proof of Theorem~\ref{theo:main}.~\QED

\section{Concluding remarks}

The best general upper bound to date on the domination number of a connected cubic graph $G$ of order~$n \ge 10$ is due to Kostochka and Stocker~\cite{KoSt09} (see, Theorem~\ref{t:KoSt}) in 2009 who showed that $\gamma(G) \le \frac{5}{14}n = \left( \frac{1}{3} + \frac{1}{42} \right)n$.

In 2010 Verstra\"{e}te~\cite{JacVert10} posed a most intriguing conjecture that if the girth, $g(G)$, of a cubic graph $G$ of order~$n$ is at least~$6$, then $\gamma(G) \le \frac{1}{3} n$. The girth requirement in Verstraete's Conjecture is essential, since there are connected cubic graphs $G$ of arbitrarily large order~$n$ that contain $4$-cycles and $5$-cycles and satisfy $\gamma(G) > \frac{1}{3} n$. The best known girth condition (prior to this paper) guaranteeing that the domination number of a cubic graph $G$ of order~$n$ is at most the magical threshold of $\frac{1}{3}n$ is due to L\"{o}wenstein and Rautenbach~\cite{LoRa08} in 2008 who showed that if $g(G) \ge 83$, then $\gamma(G) \le \frac{1}{3} n$ holds. In this paper, we prove Verstraete's conjecture when there is no $7$-cycle and no $8$-cycle, that is, we show that if $G$ is a cubic graph of order~$n$ and girth~$g(G) \ge 6$ that does not contain a $7$-cycle or $8$-cycle, then $\gamma(G) \le \frac{1}{3}n$.

Equally appealing to Verstraete's conjecture is a conjecture inspired by Kostochka in 2009 that if $G$ is a bipartite graph of order~$n$, then $\gamma(G) \le \frac{1}{3} n$. In this paper, we prove the Kostochka's related conjecture when there is no $4$-cycle and no $8$-cycle, that is, we show that if $G$ is a bipartite cubic graph of order~$n$ that contains no $4$-cycle and no $8$-cycle, then $\gamma(G) \le \frac{1}{3}n$.

It would be extremely interesting to prove Verstraete's conjecture in general when $7$-cycles and $8$-cycles are allowed in the mix, and to prove the Kostochka inspired conjecture in general when $4$-cycles and $8$-cycles are present. The considerable effort made in this paper suggests that this may be difficult. With considerably more work, the methods employed in the paper (using the new concept of marked domination in graphs, and using colored multigraphs, matchings, and intricate discharging methods) may have some hope of proving these two $\frac{1}{3}$-domination conjectures if we allow $8$-cycles back into the mix. However, completely new methods and ideas are needed to prove these conjectures in general (when we allow $7$-cycles in the case of Verstraete's conjecture and when we allow $4$-cycles in the case of the Kostochka inspired conjecture).

\section{Acknowledgements}

The authors wish to thank Justin Southey and Mickael Montassier for helpful discussions.

M. A. Henning research was supported in part by the South African National Research Foundation (grants 132588, 129265) and the University of Johannesburg.

\end{document}